\documentclass[12pt, a4paper,leqno]{amsart}
\usepackage{amsmath,amsthm,amscd,amssymb,amsfonts,amsbsy,mathtools}
\usepackage{latexsym}
\usepackage{txfonts}
\usepackage{exscale}
\usepackage{marginnote}
\usepackage{cancel}
\usepackage{soul}

\usepackage[colorlinks,citecolor=red,pagebackref,hypertexnames=false]{hyperref}

\usepackage{textcomp}

\calclayout
\allowdisplaybreaks

\numberwithin{equation}{section}

\usepackage{pgf}
\usepackage{color}

\theoremstyle{plain}
\newtheorem{theorem}{Theorem}
\newtheorem{lemma}{Lemma}[section]
\newtheorem{corollary}{Corollary}

\newtheorem{proposition}{Proposition}

\theoremstyle{definition}

\theoremstyle{remark}
\newtheorem{remark}{Remark}[section]

\newcommand{\R}{\mathbb{R}}

\newcommand{\E}{\mathbb{E}}

\newcommand{\cH}{\mathcal H}

\newcommand{\dist}{\,\mathrm{dist}\,}

\newcommand{\wt}{\widetilde}

\newcommand{\cF}{{\mathcal F}}
\newcommand{\cS}{{\mathcal S}}
\newcommand{\cP}{{\mathcal P}}
\newcommand{\cI}{{\mathcal I}}
\newcommand{\cB}{{\mathcal B}}
\newcommand{\cN}{{\mathcal N}}
\newcommand{\cQ}{{\mathcal Q}}
\newcommand{\cG}{{\mathcal G}}

\newcommand{\cE}{{\mathcal E}}
\newcommand{\cC}{{\mathcal C}}

\newcommand\red[1]{{\color{red}#1}}

\newcommand{\Z}{{\mathbb{Z}}}

\newcommand{\eps}{\varepsilon}

\newcommand{\bp}{\noindent {\it Proof}.\,\,}
\newcommand{\ep}{\hfill$\Box$ \vskip 0.08in}

\newcommand{\N}{\mathbb{N}}

\newcommand{\fl}[1]{\left \lfloor  #1 \right \rfloor}
\newcommand{\cl}[1]{\left \lceil  #1 \right \rceil}

\newcommand{\ipc}[2]{\left \langle #1 ,\, #2 \right \rangle}
\newcommand{\Ev}[1]{\mathbb{E} \left( #1 \right)}  

\newcommand{\card}[1]{{\rm Card}\left\{\, #1 \, \right\}}
\newcommand{\cards}[1]{{\rm Card}\left( #1 \right)}

\renewcommand{\P}{\mathbb{P}}
\renewcommand{\Pr}[1]{ {\P}\left\{\, #1 \, \right\}}
\newcommand{\Prr}[1]{{\P}\left (\, #1 \, \right)}

\begin{document}

\title[The landscape law]{The landscape law for tight binding Hamiltonians}
\author[D. Arnold,  M. Filoche, S. Mayboroda, W. Wang, and S. Zhang]{Douglas Arnold,  Marcel Filoche, Svitlana Mayboroda, Wei Wang, and Shiwen Zhang
}
\newcommand{\Addresses}{{
  \bigskip
  \vskip 0.08in \noindent --------------------------------------

  \footnotesize

\medskip

D.~ Arnold, \textsc{School of Mathematics, University of Minnesota, 206 Church St SE, Minneapolis, MN 55455 USA}\par\nopagebreak
  \textit{E-mail address}: \texttt{arnold@umn.edu }

\vskip 0.4cm

M.~Filoche, \textsc{Physique de la Mati\`ere Condens\'ee, Ecole Polytechnique, CNRS, Institut Polytechnique de Paris, Palaiseau, France}\par\nopagebreak
  \textit{E-mail address}: \texttt{marcel.filoche@polytechnique.edu}
\vskip 0.4cm

  S.~Mayboroda, \textsc{School of Mathematics, University of Minnesota, 206 Church St SE, Minneapolis, MN 55455 USA}\par\nopagebreak
  \textit{E-mail address}: \texttt{svitlana@math.umn.edu}

\vskip 0.4cm

  W. ~Wang, \textsc{School of Mathematics, University of Minnesota, 206 Church St SE, Minneapolis, MN 55455 USA}\par\nopagebreak
  \textit{E-mail address}: \texttt{wang9585@umn.edu }
  
\vskip 0.4cm

S.~Zhang, \textsc{School of Mathematics, University of Minnesota, 206 Church St SE, Minneapolis, MN 55455 USA}\par\nopagebreak
  \textit{E-mail address}: \texttt{zhan7294@umn.edu }
  
}}

\begin{abstract} The present paper extends the landscape theory pioneered in \cite{FM, ADFJM-CPDE, DFM} to the tight-binding Schr\"odinger operator on $\Z^d$. In particular, we establish upper and lower bounds for the integrated density of states in terms of the counting function based upon the localization landscape. 
 \end{abstract}

\maketitle



\tableofcontents

\section{Introduction and main results}\label{sec:intro}
In this paper, we consider the discrete tight-binding Schr\"odinger operator $H=-\Delta+V$ on $\Z^d$. The traditional approaches to the estimates for the integrated density of states of its continuous analogue in $\R^d$ can roughly be split into two groups, both pertaining to the asymptotic regimes. The first one are the results akin to the Weyl law, estimating the asymptotics of the spectrum as the eigenvalue $\mu\to +\infty$ in terms of the volume in the phase space of the set $\bigl\{(\xi, x):\,|\xi|^2+V(x)\leq \mu\bigr\}$, or  in terms of the associated counting function according to the Fefferman-Phong uncertainty principle \cite{Fe}. Informally speaking, letting $\mu \to +\infty$ corresponds to considering length scales which tend to zero, and hence such estimates, by design, are not relevant for a tight-binding model on $\Z^d$. And indeed, deterministic potentials on $\Z^d$ were typically treated by more ad hoc approaches specific to their structure: periodicity, symmetries, etc, see, e.g., \cite{DLY,HJ}. The second type of results pertains to the case when the potential is random. Then the integrated density of states exhibits the so-called Lifschitz tails, an exponential asymptotic behavior at the 
edge of spectrum. This regime is rather well-understood in both  $\Z^d$ and $\R^d$, but is in essence probabilistic, restricted to disordered potentials and insensitive to their individual features exhibited, for instance, on finite sets. 

The present paper introduces another approach. It takes advantage of the landscape function $u$ from \cite{FM,DFM} to build a box-counting somewhat analogous to the Weyl law and the Fefferman-Phong uncertainty principle, but associated to a {\it different} potential, the reciprocal of the landscape $1/u$. The use of the landscape in place of the original potential $V$ allows one to work in a non-asymptotic regime, contrary to aforementioned results, and in some sense to bring the ideas behind the original Weyl law to the lattice {\it without restrictions on the potential or the pertinent eigenvalues}, for both deterministic and random scenarios. In practice this approach provides a ``black box", in which the landscape, evaluated directly from the original Hamiltonian, yields an accurate approximation for the integrated density of states without any adjustable parameters, for  deterministic  and random potentials alike -- see the numerical experiments in \cite{FM, ADFJM-PRL, ADFJM-SIAM}. The present paper addresses the estimates from above and below and makes the first step towards the  mathematically rigorous understanding of the precision of the landscape predictions in the aforementioned works.

 In order to describe our main results, we introduce some notations.   Let $\Lambda= (\Z/K\Z)^d\cong \{\bar 1,\cdots,\bar K\}^d$ be an integer torus, where $K\in\N$, $K\geq 3$, and  $\bar k$, $1\leq k\leq K,$ is the congruence class, modulo $K$. For simplicity, we will omit the bar from $\bar k$ when it is clear.   Let $V=\{v_n\}_{n\in {\Lambda}}\in \ell^\infty(\Lambda)$ be a real-valued, non-constant, non-negative potential. We denote by $V_{\max}=\max_{n\in {\Lambda}}v_n$ the amplitude of the potential. 
The  tight binding  Hamiltonian $H $ is the linear operator on $\cH :=\ell^2({\Lambda}) \cong \R^{K^d}$ defined by
\begin{equation}
    (H \varphi)_n=-\sum_{|m-n|_1=1}\left(\varphi_m-\varphi_n\right)\,+ v_n \varphi_n,\ \ n\in {\Lambda}, \label{eq:opH-intro}
\end{equation}
where $|n|_1:=\sum_{i=1}^d|n_i|$ is the $1$-norm on $\Lambda$. We may think of $\varphi$ either as a periodic sequence $\varphi_n$ indexed by $n\in \Z^d$ or as a periodic function $\varphi(n)$ on $\Z^d$. We are interested in the normalized integrated density of states of $H$, i.e., the eigenvalue counting function per unit volume:
\begin{equation}\label{eq:Ndef-intro}
   N(\mu):=K^{-d}\,  \times  \left\{ \, \textrm{the number of eigenvalues\ } \lambda \ {\rm of}\ H\ {\rm such \ that}\ \lambda\le \mu  \right\} .
\end{equation}

In 2012, a new concept called the localization landscape was introduced in \cite{FM}. Given an operator $H$ as above, the discrete localization landscape function is the unique  solution $u=\{u_n\}_{n\in \Lambda}\in \cH$ to the equation $(H u)_n=1$. When $V=0$ and $u$ vanishes on the boundary, the landscape is simply the torsion function of the Dirichlet Laplacian. 

Next, let us define the landscape box counting function $N_u$.  We start by defining, for any positive integer $s$, a partition $\cP(s)$
of the set $\{1,\cdots,K\}^d$ into subsets which are mostly boxes of side length $s$, as follows. Writing $K = qs + r$
(where the quotient $q$ and the remainder $r$ are non-negative integers and $r < s$), we define a
partition $\cP_1(s)$ of the set $\{1,\cdots,K\}$ into $q$ subsets of $s$ consecutive elements, and, if $r>0$, one
additional subset of cardinality $r$. The partition $\cP(s)$ then consists of the boxes defined by the
Cartesian products of $d$ subsets from $\cP_1(s)$, see Figure \ref{fig:Disbox}. 
\begin{figure}
    \centering
    \includegraphics[width=0.5\textwidth]{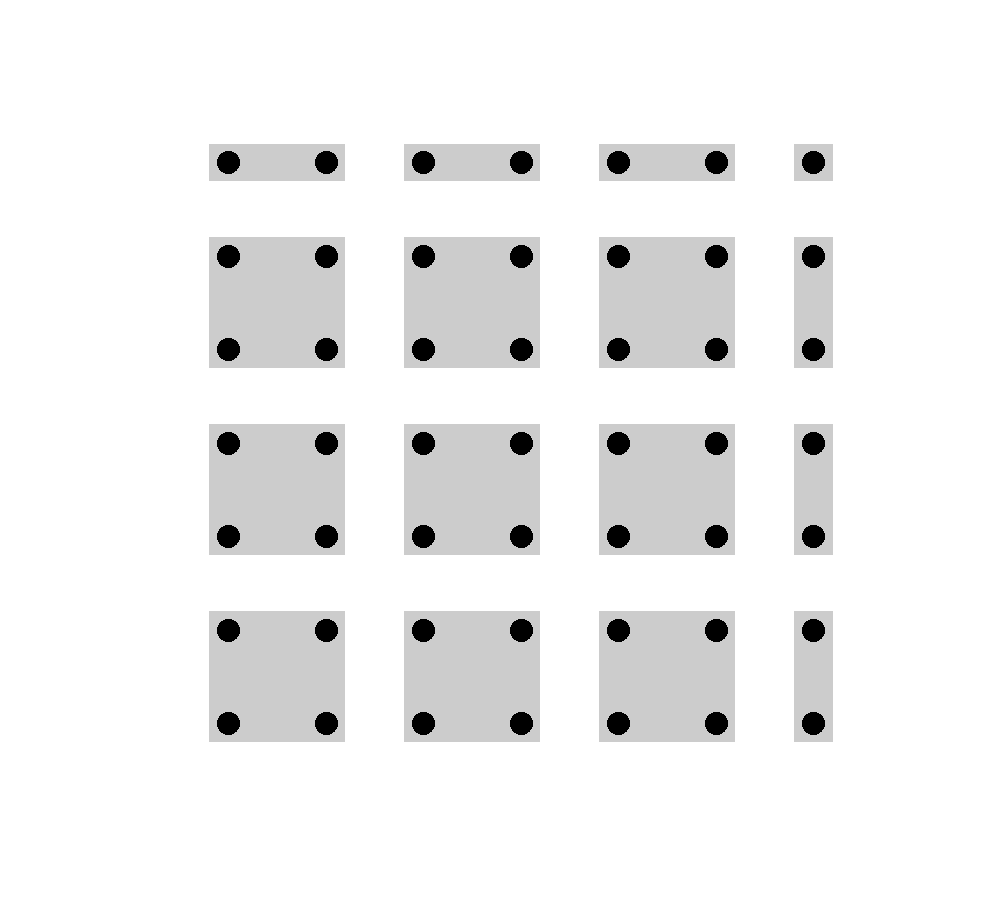}
    \caption{The partition $\cP(2)$ for $\Lambda=\{1,\cdots,7\}^2$.}
    \label{fig:Disbox}
\end{figure}
For given $\mu>0$, we then set $s(\mu)=\cl{\mu^{-1/2}}$, and define $N_u(\mu)$ as the
number of boxes on which the minimum of $1/u_n$ does not exceed $\mu$, normalized by the size of the
set $\Lambda$:
\begin{equation}\label{eq:Nudef-intro}
    N_u(\mu)=K^{-d}\times
     \, \left\{ \textrm{the number of \ } Q\in \cP\big(s(\mu)\big)\ {\rm such \ that}\ \min_{n\in Q}\frac{1}{u_n}\le \mu \right\}.
\end{equation}
Our goal is to estimate the integrated density of states $N$ in terms of the landscape box counting function $N_u$.
To this end, we will establish several estimates, stated below as Theorems~\ref{thm:NNu-intro},
\ref{thm:NNu-intro-2}, and \ref{thm:NNu-intro-scaling}, which we collectively refer to as the \emph{Landscape Law}.

The first result, which will be proven in Section~\ref{subsec:ub}, shows that, after a proper scaling, $N_u$ provides an upper bound for $N(\mu)$ over the whole range of $\mu$.
\begin{theorem}\label{thm:NNu-intro} Let $V\in \ell^\infty(\Lambda)$ be a non-constant, non-negative potential. Then there is a dimensional constant  $C_1>0$, such that 
\begin{equation}\label{eq:N<Nu-intro}
     N\left(\mu \right)\le N_u(C_1\mu) \quad \mbox{for all $\mu>0$.}
\end{equation}
\end{theorem}
\noindent
In saying that $C_1$ is a dimensional constant, we mean that it depends only on $d$, and, in particular, is independent of $K$ and $V$. In fact, we can take $C_1=4d$ in \eqref{eq:N<Nu-intro}.

The next theorem, proved in Section~\ref{sec:non-scaling}, contains the key estimate for obtaining a lower bound for $N(\mu)$.
\begin{theorem}\label{thm:NNu-intro-2}
Retain the hypotheses in Theorem \ref{thm:NNu-intro}. 	Then there are dimensional constants $c_\ast,c_0,c_1,C_0,\alpha_0,c'_0,c'_1,C'_0$ (in particular, independent of $K$ and $V$), such that 
\begin{equation}\label{eq:NNu-intro-2}
N(\mu)\ge c_0{\alpha^d}N_u(c_1 \alpha^{d+2}\mu)-C_0N_u(c_1 \alpha^{d+4}\mu)  \quad \mbox{for all $0<\mu\le c_\ast\alpha^{-4}$ and $0<\alpha<\alpha_0$,}
\end{equation} 
 and
\begin{equation}\label{eq:intro-mu>1}
    N(\mu)\ge c_0' N_u(c'_1 \alpha^{2}\mu)-C_0' N_u(c'_1 \alpha^{4}\mu)  \quad \mbox{for all $\mu> c_\ast\alpha^{-4}$ and $0<\alpha<\alpha_0$.}
\end{equation}

\end{theorem}

In order to get a positive lower bound on $N(\mu)$,  we will remove the negative correction terms 
on the right-hand side of \eqref{eq:NNu-intro-2} and \eqref{eq:intro-mu>1} 
through several complementary mechanisms. We roughly divide potentials into two regimes, corresponding to 
 potentials satisfying a  certain scaling condition  and to certain disordered potentials.  Note that  these regimes  could overlap and  do not between them cover all potentials. 

\subsection{Deterministic potentials subject to the doubling scaling estimates} 
\begin{theorem}\label{thm:NNu-intro-scaling}
Retain the hypotheses of Theorem \ref{thm:NNu-intro} and assume that $u$ satisfies the scaling condition 
\begin{equation}\label{eq:scaling}
	\sum_{n\in 3Q}u_n^2\le C_S\left(\sum_{n\in Q}u_n^2+\ell^{d+4}\right)
	\end{equation}
	for every cube $Q\subset \Lambda$ of side length $\ell$. Here, $3Q$ is the tripled cube  concentric with $Q$ (see the definition in \eqref{eq:3Q}).
Then there exist positive constants $C_1, \, c_3$ depending on dimension only and  
$c_2$ depending on $d$, $V_{\max}$ and $C_S$  such that 
	\begin{equation}\label{eq:N>Nu2-intro}
	c_3 N_u(c_2\,\mu)\leq N(\mu)\leq N_u(C_1\mu) \quad \mbox{for all $\mu>0$}.
	\end{equation}
\end{theorem}

Assumption~\eqref{eq:scaling} is analogous to doubling hypotheses which are commonly used in the continuous case for elliptic PDEs. Such estimates are standard consequences of the Harnack and De Giorgi--Nash--Moser arguments which hold for homogeneous equations and for the Schr\"odinger equation with relatively slowly varying potentials, for instance, within the Kato class. See the discussion in \cite{DFM, Ku, HL}.

The scaling condition \eqref{eq:scaling} also holds whenever $V$ is periodic.  Indeed, suppose that $\{v_n\}$ is periodic in each of the $d$ coordinate directions with period vector $\vec p=(p_1,\cdots,p_d)\in \N^d$ (see, e.g., \cite{DLY,Ea,HJ,RS}). Assume that $K$ is divisible by each $p_i$.
Then, as we show in Section \ref{sec:peri}, the scaling condition \eqref{eq:scaling} is satisfied with $C_S$  depending on $d,V_{\max}$, and $\vec p$, but not on $K$, which yields 
\begin{corollary}\label{cor:periodic}
 Let $H=-\Delta+V$ be as in \eqref{eq:opH-intro}, with a non-trivial periodic potential  $V=\{v_n\}_{n\in\Lambda}$ as above.  Then 
\begin{equation} 
\label{twosided} c_3\, N_u(c_2\,\mu)\, \leq \, N(\mu)\, \le\, N_u(C_1\, \mu)   \quad \mbox{for all $\mu>0$},\end{equation}
where  $C_1,c_3 $ are  dimensional constants and $c_2$ depends on $d,V_{\max}$, and $\vec p$ only. 
\end{corollary}

Perhaps the major example when \eqref{eq:scaling} might fail is that of disordered systems. Indeed, if on a 1-dimensional lattice we could have an arbitrarily  long  region of $V=0$ followed by an arbitrarily long  region of $V=1$, that would correspond to a region of $-\Delta u=1$ followed by a region of $-\Delta u+u =1$. In the first case $u$ is quadratic and in the second exponential, which clearly destroys the ``doubling" required by \eqref{eq:scaling}. Fortunately, there is a complementary mechanism to obtain an improved estimate akin to \eqref{eq:N>Nu2-intro} from \eqref{eq:NNu-intro-2}.  

To illustrate it, suppose that for $\mu$ belonging to some interval on the positive half-line, we have the bounds
\begin{equation*}
a\, \mu^\beta\, \le \,  N_u(\mu) \,  \le \,  b \,  \mu^\beta,
\end{equation*}
where the power $\beta>d/2$.  Substituting these bounds into \eqref{eq:NNu-intro-2} and choosing $\alpha$ sufficiently small,
it is easy to deduce the lower bound 
\eqref{eq:N>Nu2-intro} on $N(\mu)$ for $\mu$ in the same interval. 
A similar argument can be used to obtain a lower bound for $N(\mu)$, or, more precisely, for the expectation of $N(\mu)$, in the case of Anderson potentials or any disordered potential near  a fluctuation boundary. This is basically due to the fact that the aforementioned exponential nature of the Lifshitz tails ``beats" the negative polynomial correction in \eqref{eq:NNu-intro-2}. 

\subsection{Disordered potentials}\label{subs-disorder} 
Assume that the values $\{v_n\}_{n\in \Lambda}$ are given by independent, identically distributed (i.i.d.) random variables, with common probability measure $P_0$ on $\R$. Denote by $F(\delta)=P_0(v_n\le \delta)$ the common cumulative distribution function of $v_n$ and by 
\[{\rm supp}\,P_0=\Big\{ \,  \mu\in\R:\,P_0\bigl(v_n\in (\mu-\eps,\mu+\eps)\bigr)>0, \,\forall \,\eps>0\, \Big\}\]
the support of the measure $P_0$. 
We   assume that $\inf {\rm supp}\,P_0=0$ and $\sup {\rm supp}\,P_0=V_{\max}>0$. We denote by $\E(\cdot)$ the expectation with respect to the product measure on $\R^{|\Lambda|}$ generated by $P_0$.  
\begin{theorem}\label{thm:NNu-iid-intro}
Let $V=\{v_n\}_{n\in\Lambda}$ be an Anderson-type potential  as above. Let $C_1$ be as in Theorem \ref{thm:NNu-intro}. Then there are  constants $c_5,c_6>0$ depending  on $d$,  the expectation of the random variable, and $V_{\max}$, such that 
\begin{equation} \label{eq:NNu-And-intro}
   c_5\, \E N_u(c_6\,\mu) \le \E N(\mu)\le \E N_u(C_1\, \mu) \ \ \ {\textrm{for all}}\ \  \mu>0.
\end{equation}
Furthermore, there is a  constant $\mu_\ast>0$ depending on $d$ and expectation of the random variable, such that if, in addition, $\mu\le \mu_\ast$, then \eqref{eq:NNu-And-intro} holds with the constants $c_5,c_6$ independent of $V_{\max}$.
\end{theorem}

We note that in the course of the proof of Theorem~\ref{thm:NNu-iid-intro} we prove the following universal bound on Lifschitz tails in terms of the  cumulative  distribution function $F$. 

\begin{proposition}\label{lLif} Retain the setting of Theorem~\ref{thm:NNu-iid-intro}.
Then there are constants $\mu_0,K_\ast,c_i$, depending only on  the dimension and the expectation of the random variable (but independent of $K$), such that  
\begin{equation}\label{eq:N-tail-intro-new}
c_1 \mu^{d/2}F(c_2 \mu)^{  c_3\mu^{-d/2}}\le \E N(\mu)\le  c_4\, \mu^{d/2}\, F(c_5\, \mu)^{\, c_6\, \mu^{-d/2}} \   \ {\rm for \ all \ } \mu\in(K_\ast/K^2,\mu_0).
\end{equation}
\end{proposition}
To the best of our knowledge, this statement has never been formulated in this generality, even though perhaps it would not surprise the experts. The more traditional, weaker double log asymptotics are now considered classical (see \cite{Li, KM2, Si1, Ki}) and for certain classes of random potentials they have been improved in \cite{BiKo, Ko, KM1} and other works. Here, Proposition~\ref{lLif} does not carry any a priori assumptions on the underlying probability distribution, and is a by-product of the landscape  method.  

\subsection{The dual landscape and computational examples} 
Contrary to the continuous case, the spectrum of the discrete Schr\"odinger operator $H=-\Delta+V$ is a compact subset in $[0,4d+V_{\max}]$. The eigenvalue counting near the top of the spectrum for $\wt \mu$ close to $4d+V_{\max}$ can be converted into the counting near the bottom of the spectrum for $\mu=4d+V_{\max}-\wt \mu$ close to $0$.  Such a conversion  is obtained via  a dual model $\wt H=-\Delta+V_{\max}-V$, see \cite{LMF, WZ}. One defines the dual landscape function $\wt u$ as the solution to $(\wt H\, \wt u)_n=1$ and the box-counting function $N_{\wt u}$ using \eqref{eq:Nudef-intro}, leading to   
\begin{corollary}\label{cor:duallandscape}
Retain the definitions in Theorem \ref{thm:NNu-iid-intro}. Let $C_1$ be as in Theorem \ref{thm:NNu-intro}. Suppose $K\ge 3$ is even. There are positive constants $\wt c_5,\wt c_6$ depending on $d$, the expectation of the random variable, and $V_{\max}$, such that 
\begin{equation}\label{eq:NNu-top-intro}
1-\E N_{\wt u}(C_1\,\wt \mu) \le \E N(\mu)\le  1-\wt c_5\E N_{\wt u}(\wt c_6\,\wt \mu) \ \ {\textrm{for all}}\  \mu<4d+V_{\max}, 
\end{equation}
where $\wt \mu=4d+V_{\max}-\mu$ and $\wt u$ is the landscape function for $-\Delta+V_{\max}-V$. 
\end{corollary}

If one carefully tracks the values of the constants in \eqref{eq:NNu-And-intro} and \eqref{eq:NNu-top-intro} obtained in the proofs, they are of course far from optimal. However, the formulas emphasize the correct features of the spectrum and, as we have mentioned above, the numerical experiments actually yield even more satisfactory results than the formal estimates seem to warrant. In \cite{DMZDAWF} and accompanying numerical work still in preparation, we show that there are very stable constants $c_1,c_2$ such that a practical landscape law holds: $N(\mu)\approx c_1N_u(c_2\mu)$, see Figure \ref{fig:NNu}. 
\begin{figure}
    \centering 
    \includegraphics[width=0.48\textwidth]{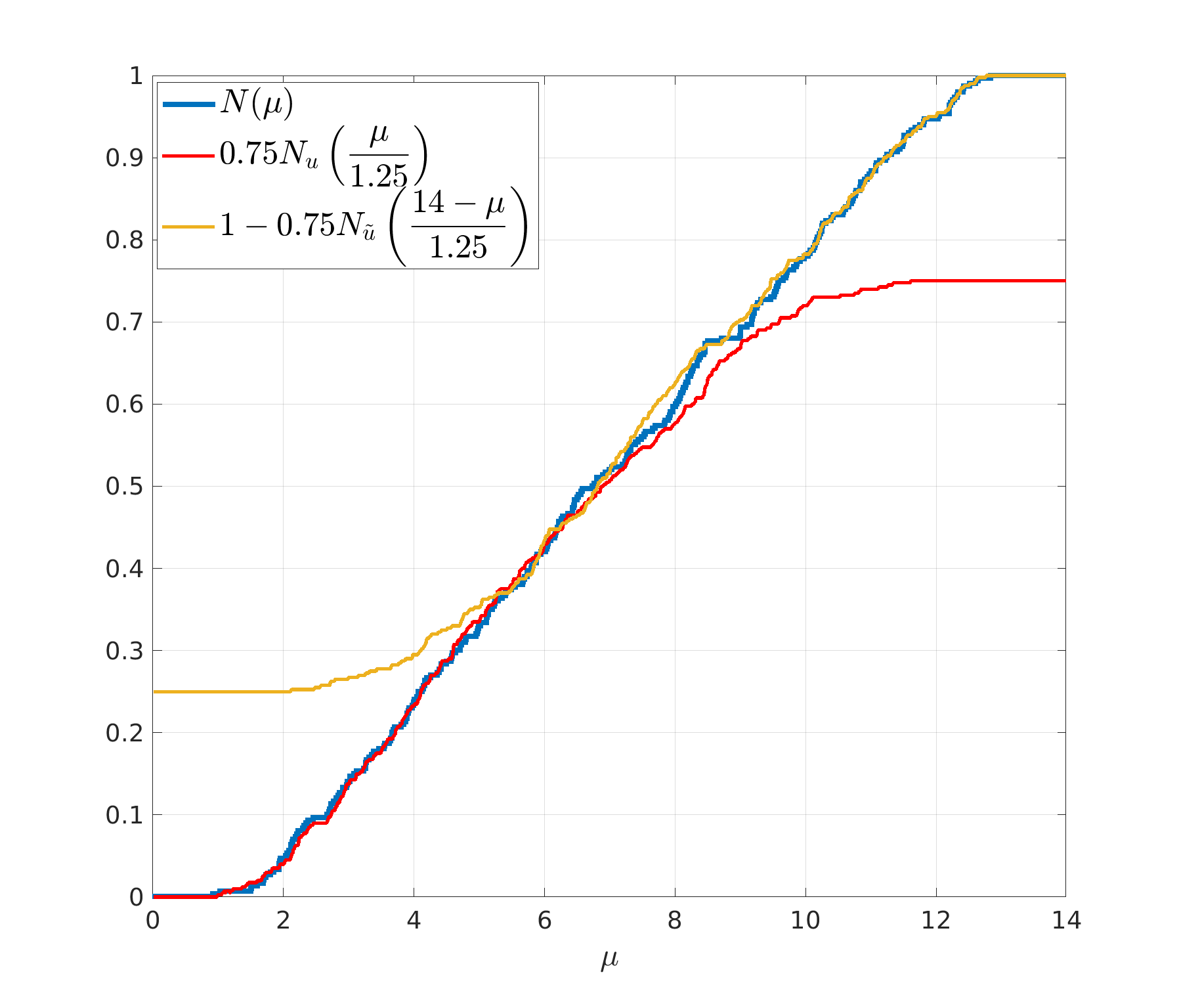}
     \includegraphics[width=0.48\textwidth]{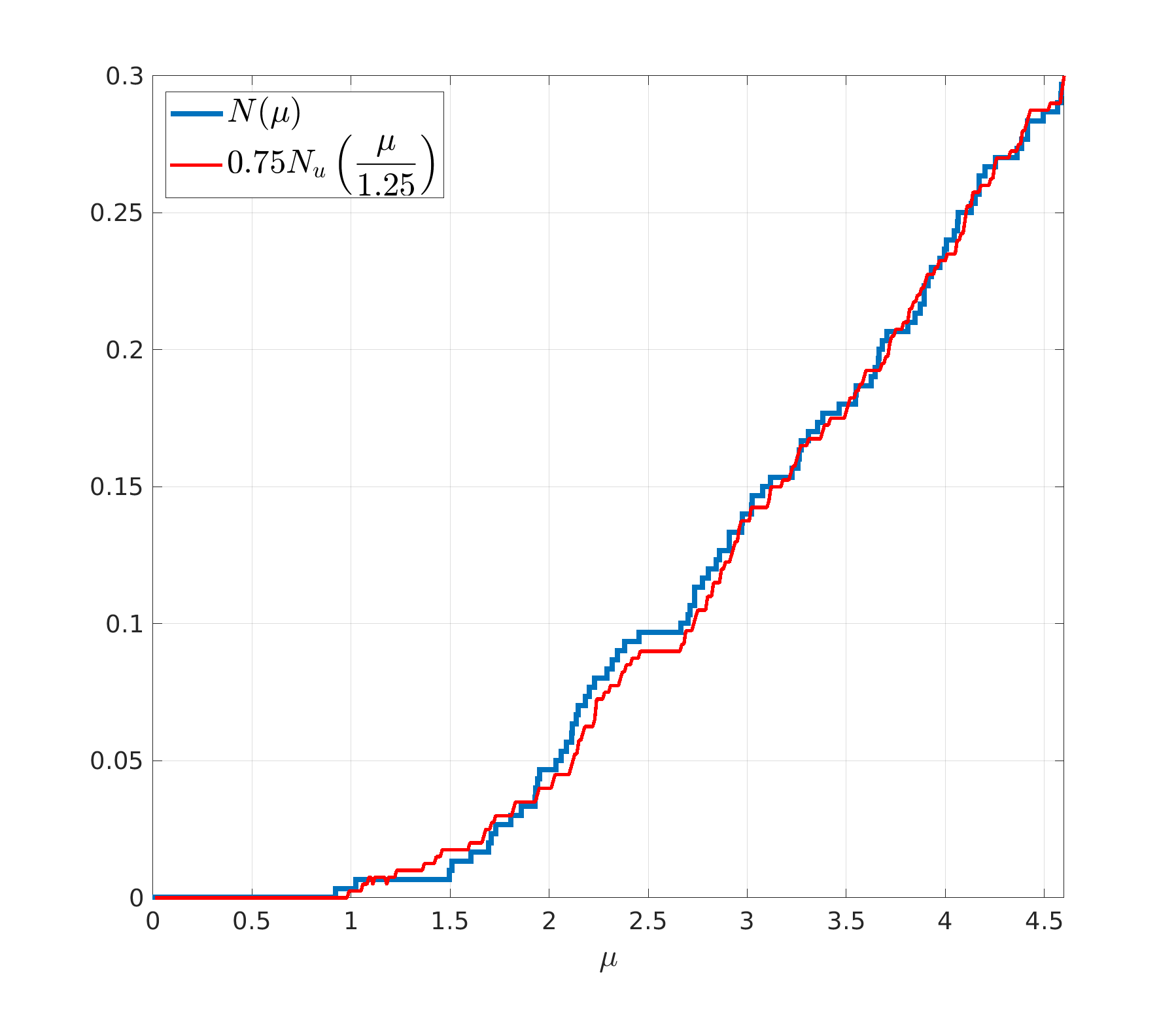}\caption{One-dimensional discrete Schr\"odinger operator on a periodic lattice $\Z/300\Z$, with a random potential $V=\{v_n\}$ uniformly distributed  in $[0,10]$. Comparison between the true eigenvalue counting function $N$, the scaled landscape box-counting function $N_u$, and the dual landscape box-counting function $N_{\wt u}$. The first plot shows the whole spectrum, while the second zooms in on the bottom spectrum. }
    \label{fig:NNu}
\end{figure}
We would like to point out that the landscape counting function, being a deterministic  rather than a probabilistic tool, even picks up a spectral gap around the energy $\mu \approx 2.5$ -- a feature which would not be feasible, for instance, via the Lifschitz tail estimates. These details would of course disappear in the limit of an infinite domain but they demonstrate a surprising precision of the Landscape Law compared to any other currently available method.

\vskip 0.08 in

The rest of the paper is organized as follows. We state preliminaries for tight-binding Hamiltonians and the discrete landscape theory in Section \ref{sec:pre}. In Section \ref{sec:det}, we study deterministic potentials and prove Theorem \ref{thm:NNu-intro}, \ref{thm:NNu-intro-2}, \ref{thm:NNu-intro-scaling}, and Corollary \ref{cor:periodic}. In Section \ref{sec:Anderson}, we concentrate on the Anderson model. We first prove the Lifshitz tail estimates for $N_u$  and finally conclude Theorem \ref{thm:NNu-iid-intro} in Section \ref{sec:LSlaw-iid}. Section \ref{sec:dual} is a discussion of the dual landscape theory. In the Appendix, we include some technical estimates for discrete harmonic functions and a well known probability result called  Chernoff-Hoeffding bound. The key properties of the landscape function strongly rely on the foundations of the theory of  elliptic PDEs. Many of these results require different techniques on $\Z^d$ compared to their continuous analogues. For instance, because of the lack of rotational symmetry and dilational invariance, many estimates for the Poisson kernel and the Green's function are not known on a lattice, and are technically difficult to prove. A substantial portion of the paper is devoted to the discrete analogues of these elliptic estimates, and we hope they will be of independent interest.

\vskip 0.1 in

\noindent{Acknowledgments.}  Arnold  is  supported  by the NSF grant DMS-1719694 and Simons Foundation grant
	601937, DNA.  Filoche is supported by Simons Foundation grant  601944, MF. Mayboroda is supported by NSF DMS 1839077 and the Simons Collaborations in MPS 563916, SM.   Wang is supported by Simons Foundation grant 601937, DNA.   Zhang is supported in part by the NSF grants DMS1344235, DMS-1839077, and Simons Foundation grant 563916, SM.


\section{Preliminaries}\label{sec:pre}
In the tight-binding model, the Hilbert space is taken as the space of sequences
$
    \ell^2(\Z^d)=\left\{\{ \phi_i\}_{i\in\Z^d}\, |\, \sum_{i\in\Z^d}|\phi_i|^2<\infty\right\}
$
where we may think of $\phi=\{\phi_i \}_{i\in\Z^d}$ either as a function $\phi=\phi(i)$ on $\Z^d$ or as a sequence $\{\phi_i\}$ indexed by $i\in\Z^d$. {The $\Z^d$ lattice is equipped with the $1$-norm: }
\begin{equation}\label{eq:norm1}
    |n|_1:=\sum_{i=1}^d|n_i|,
\end{equation}
 which reflects the graph structure of $\Z^d$. We will also frequently need the infinity (maximum) norm 
 \begin{equation}\label{eq:norminfty}
    |n|_\infty:=\max_{1\le i\le d}|n_i|.
\end{equation}
 Two vertices $m=(m_1,\cdots, m_d), n=(n_1,\cdots, n_d)\in \Z^d$ are called the nearest neighbors  if $|m-n|_1=1$. We also say that   nearest neighbors $m,n$ are connected by an edge of the discrete graph $\Z^d$. We denote by $e_i=(0,\cdots, 0,1, 0,\cdots, 0),\,i=1,\cdots, d,$ the elements of the canonical basis of $\Z^d$. For $\phi=\{\phi_n\}_{n\in \Z^d}\in \R^{\Z^d}$,  its $i$-th directional  (forward) difference  $\nabla_i\phi:\, \Z^d \to \R$ is  defined as \begin{equation}\label{eq:nabla-i}
    \nabla_i\phi_n=\phi_{n+e_i}-\phi_{n},\ \  i=1,\cdots, d, 
\end{equation}
and its gradient $\nabla \phi: \Z^d \to\R^d$ is
\begin{equation*}
   \nabla \phi_n=\left(\nabla_1\phi_n,\nabla_2\phi_n,\cdots,\nabla_d\phi_n\right).
\end{equation*}
We also denote the dot product and the induced norm of the resulting vectors by
$(\nabla g \cdot \nabla f)(n)=\sum_{i=1}^d \nabla_i g_n\cdot \nabla_i f_n$, and $|\nabla f|(n):=  \sqrt{ (\nabla f \cdot \nabla f)(n)}.$  
 The discrete (graph) Laplacian $\Delta$ on $\Z^d$ is defined as usual, acting on $\phi=\{\phi_n\}_{n\in \Z^d}$, via
\begin{equation}\label{eq:Lap}
    (\Delta\phi)_n=\sum_{|m-n|_1=1}\left(\phi_m-\phi_n\right)=\sum_{1\le i\le d}\left(\phi_{n+e_i}+\phi_{n-e_i}\, -\, 2\phi_n\right) .
\end{equation}
For a real sequence $\{v_n\}_{n\in\Z^d}$ on $\Z^d$, the potential $V$ is a multiplication operator acting on $\phi\in\ell^2(\Z^d)$ as $(V\phi)_n=v_n\phi_n$. The operator $-\Delta+V$ is called the discrete Schr\"odinger operator on $\Z^d$. If one takes $v_n=v_n(\omega)$  as independent, identically distributed random variables (in some probability space),  the random operator $-\Delta+V(\omega)$ is usually referred to as the Anderson model. We refer readers to \cite{AW,Ki} for more details and a complete introduction to tight-binding Hamiltonians and the Anderson model. 

Throughout the rest of the paper, we consider the discrete Schr\"odinger operator $-\Delta+V$ restricted to a finite domain in $\Z^d$.  Let $\Lambda=(\Z/K\Z)^d\cong\{\bar 1,\bar 2,\cdots,\bar K\}^d$, where $K\in\N$ and  $\bar k$, $1\leq k\leq K,$ is the congruence class, modulo $K$. For simplicity, we often treat $\Lambda$ as a subset of $\Z^d$.  Slightly abusing the notation, we denote by $|\cdot|_1$ the induced 1-norm of $\Z^d$ on the congruence class $\Lambda$, where, for example, we consider two points $(1,n_2,\cdots,n_d)$ and $(K,n_2,\cdots,n_d)$ to be   nearest neighbors and to have distance one from each other in $\Lambda$.  From now on, we will concentrate on the finite dimensional subspace $\ell^2({\Lambda}) $ of $\ell^2(\Z^d)$. We frequently write $\cH:= \ell^2({\Lambda}) \cong \R^{K^d}$ for simplicity. The linear space $\cH$ is equipped with the usual inner product on $\R^{K^d}$, which is denoted by $\ipc{\cdot}{\cdot}=\ipc{\cdot}{\cdot}_{\cH}$.
It is easy to check that for $\phi \in \cH$,
\begin{equation}\label{eq:bdry-peri}
     \phi_{n}=\phi_{n+Ke_i},\ n\in \Lambda,\ i=1\cdots,d ,
\end{equation}
which specifies the periodicity of $\phi$.
 
We assume that  $v_n\ge0$ is real valued and non-constant, and that $\min_{\Lambda}v_n\ge0$.   We set $V_{\max}:=\max_{\Lambda}v_n>0$ and let $H=H_{\Lambda}$  be the restriction of $-\Delta+ V$ to $\cH$:
\begin{equation}\label{eq:opH}
    (H_\Lambda \phi)_n=-(\Delta \phi)_n+(V\phi)_n=-\, \sum_{|m-n|_1=1}\, \left(\phi_m-\phi_n\right)\, +\,  v_n \phi_n,\ \ n\in {\Lambda} .
\end{equation}

Similar to the continuous case, the discrete Hamiltonian can be written in its  Dirichlet form on the periodic lattice $\Lambda$:
\begin{equation*}
    \ipc{\phi}{H\phi}_{\cH}=\sum_{n\in\Lambda}\|\nabla \phi_n\|^2+\sum_{n\in\Lambda}v_n\phi_n^2,
\end{equation*}
where $\|\nabla \phi_n\|^2:=\sum_{i=1}^d|\nabla_i \phi_n|^2$. 

It is easy to check that all eigenvalues of $H$ in \eqref{eq:opH} are contained in $[0,4d+V_{\max}]$ for any finite $K$. For the  Anderson model $H_\infty=-\Delta+V(\omega)$ acting on the entire space $\ell^2(\Z^d)$, it is well known that the spectrum  $\sigma(H_\infty)$ is (almost surely) the non-random set $[0,4d]+{\rm supp}\,V\subset[0,4d+V_{\max}]$. 
 
A linear operator acting a finite dimensional space may be viewed as a matrix. For example,  $H_\Lambda=-\Delta+V$ acting on $\Lambda=\Z/K\Z$ in \eqref{eq:opH}  may  be identified with the sum of the two $K\times K$ matrices, 
\begin{equation}\label{eq:LapPeri}
    -\Delta  =\begin{pmatrix}
	2 & -1 &0  & \cdots &  -1 \\
	-1 &   2 &\ddots   &  \vdots \\
	0 & \ddots& \ddots&\ddots  & 0 \\
	\vdots  &\ddots & \ddots&  2   &-1 \\
	-1  &\cdots & 0&-1 & 2
\end{pmatrix}, \quad \ \ V=\begin{pmatrix}
	  v_1 & 0 &0  & \cdots &  0 \\
	0 &   v_2 & 0 &\ddots   &  \vdots \\
	0 & \ddots& \ddots&\ddots  & 0 \\
	\vdots  &\ddots & \ddots&  v_{K-1}   &0 \\
	0  &\cdots & 0&0 &   v_K
\end{pmatrix}.
\end{equation}

It is easy to verify that $H$ is invertible. 
Moreover, by the maximum principle (see Lemma \ref{lem:maxP-periodic}), all the matrix elements of its inverse are positive, $H^{-1}(i,j)>0$ for all $i,j\in \Lambda$. Therefore, there is a unique positive vector $u\in \ell^2(\Lambda)$ solving the equation $(Hu)_n=1,n\in\Lambda$. 
The equation will be referred to as the \emph{landscape equation} and the solution $u=\{u_n\}_{n\in {\Lambda}}$,  will be called the {\emph{landscape function}}. The function $u$ thus defined is the discrete analogue of the landscape function in \cite{FM} in the continuum setting. The discrete landscape function was first introduced in \cite{LMF}, for a one dimensional lattice with zero boundary conditions. It was studied on a higher dimensional lattice with periodic boundary conditions in \cite{WZ}. The following result can be found in \cite{WZ}.
\begin{theorem}[Theorem 2.10, Lemma 2.12 in \cite{WZ}]\label{thm:landscape}
Assume that $v_n \ge 0$ and is not identically zero.  
Let $u=\{u_{n}\}\in \ell^2({\Lambda}) $ be the  unique solution of  the landscape equation $(Hu)_n=1$. Then 
\begin{equation}\label{eq:u-lower}
    \min_{ {n\in {\Lambda}} }u_n\ge \frac{1}{\max_{n\in\Lambda}v_n} >0.
\end{equation}
\end{theorem}

As shown in \cite{ADFJM-CPDE} for the continuous case and in \cite{WZ} for the discrete case, $1/u:=\{1/u_n\}_{n\in\Lambda}$ serves as an effective potential via the following landscape uncertainty principle: 
\begin{theorem}[Lemma 2.14 in \cite{WZ}]
For any $f\in \ell^2(\Lambda)$,
\begin{equation}\label{eq:uncert}
\ipc{f}{H  f}_\cH 
=\, \sum_{n\in {\Lambda}}\,\sum_{1\le i\le d}\, u_{n+e_i} u_n\,  \left(\nabla_i\frac{ f_{n}}{u_n}\right)^2\, 
\, +\,\sum_{n\in {\Lambda}}\, \frac{1}{u_n}f_n^2 
\,\ge\, \sum_{n\in {\Lambda}}\frac{1}{u_n}f_n^2 ,
\end{equation}
where $\nabla_i\frac{ f_{n}}{u_n}=f_{n+e_i}/u_{n+e_i}-f_{n}/u_{n}$.
\end{theorem}

Let us introduce a few more notations. For $a,b\in\Z,a\le b$, we denote by $\llbracket a,b \rrbracket=\{a,a+1,\cdots,b\}$ consecutive integers from $a$ to $b$. We will frequently work with cubes in $\Z^d$, and  their images  in $\Lambda=(\Z/K\Z)^d$. 
For $r\in\N,Q=\llbracket 1,r \rrbracket^d$, we say that $Q$ is a cube in $\Z^d$  of side length $\ell(Q)=r$, which is the cardinality  of $Q$ projected in each direction. We denote by $|Q|=\cards{Q}$ the total cardinality  of $Q$  and call it the volume of $Q$ when  it is clear. For any $a\in \Z^d$,  $a+Q$ is the translation of $Q$ in $\Z^d$, respectively having the same side length and volume. For a cube $Q$, we denote by $3Q$ the cube concentric with $Q$ of side length $3\ell(Q)$ 
\begin{equation}\label{eq:3Q}
    3Q:=\bigcup_{\substack{1\le i \le d \\ k_i=0,\pm \ell(Q)}}\, \left(Q+k_1e_1+k_2e_2+\cdots +k_de_d\right), 
\end{equation}
see Figure \ref{fig:3Q}.
\begin{figure}
    \centering
    \includegraphics[width=0.3\textwidth]{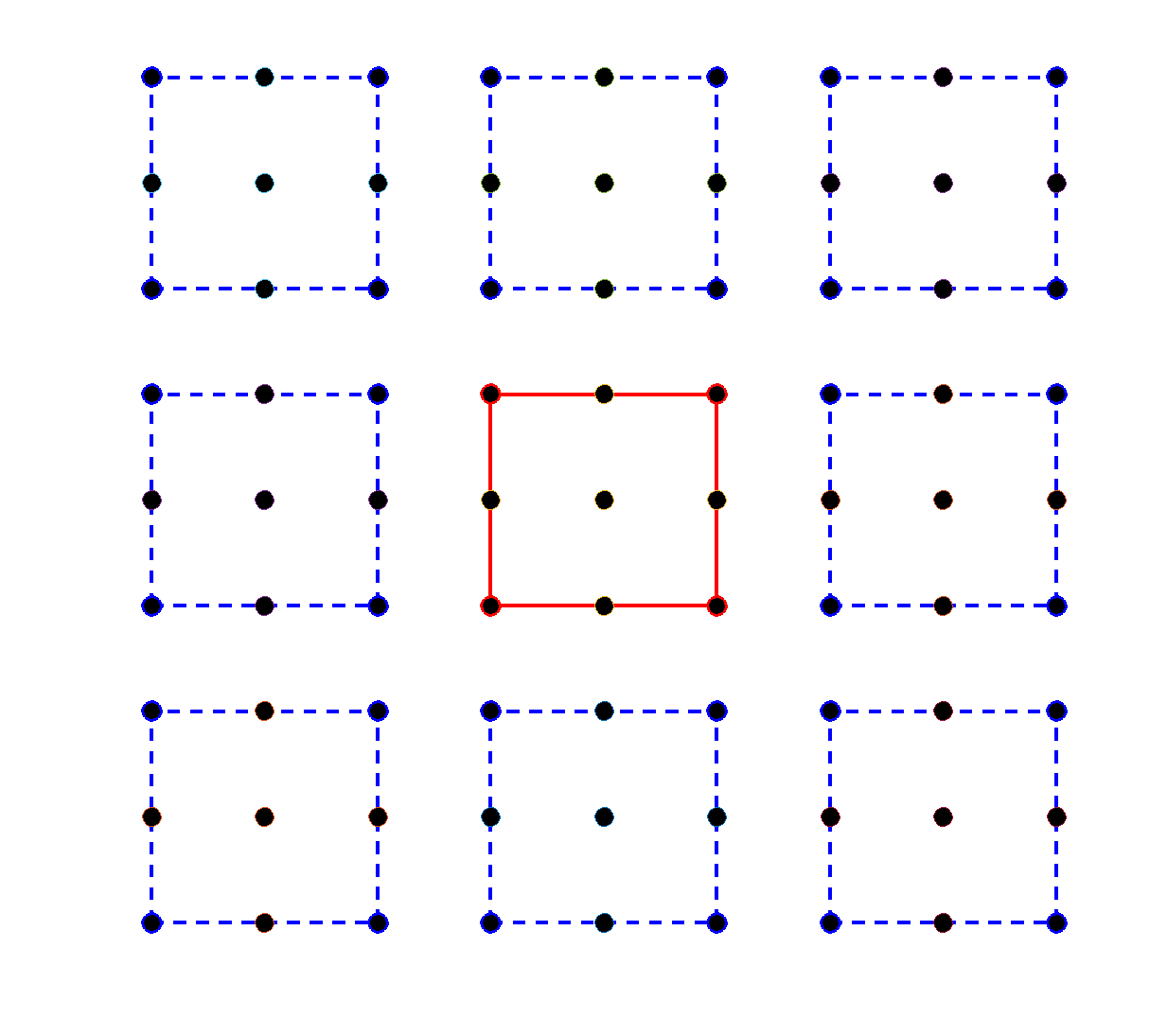}
    \caption{The center cube (in red) $Q=\{4,5,6\}^2$ has side length $3$. $Q$ is surrounded by its translations (in blue) $Q+ae_1+be_2,a,b=0,\pm3$. The union of all the small cubes is $3Q=\{1,\cdots,9\}^2$ of side length $9$.}
    \label{fig:3Q}
\end{figure}
Let $\partial Q$ be the (inner) boundary of $Q$:
\begin{equation}\label{eq:bdry-inner}
    \partial Q=\left \{\, n\in Q:\, n+e_i\not\in Q\ {\rm or}\  n-e_i\not\in Q\ {\rm for \ some}\, 1\le i\le d\, \right\},
\end{equation}
and let $\partial^\circ Q\subset \partial Q$ be the flat part of the boundary, removing all the   corners:
\begin{equation}\label{eq:bdry-no-corner}
\partial^\circ Q =\left\{n\in \partial Q:\,   n+e_i\not\in Q\ {\rm or}\  n-e_i\not\in Q\  {\textrm{for only one}}\ 1\le i \le d\right\}. 
\end{equation}

 For an   integer interval  $I=\llbracket a,a+r-1 \rrbracket$ of side length $r\ge 3$, we denote by $I/3:=\llbracket a+\cl{r/3},a+\cl{r/3}+\fl{r/3}-1 \rrbracket$ the middle third interval of $I$. For a cube $Q=I_1\times I_2\times \cdots \times I_d, I_i=\llbracket a_i,a_i+r-1 \rrbracket,i=1\cdots,d$, we denote by $Q/3$ the middle third cube of $Q$, defined as: 
\begin{equation}\label{eq:Q/3}
    Q/3:=(I_1/3)\times (I_2/3)\times \cdots \times (I_d/3).
\end{equation} 
It is easy to verify that $Q/3$ is the ``thin'' middle third part of $Q$ in the sense that $\ell(Q/3)=\fl{\ell(Q)/3}\le \ell(Q)/3$ and $3(Q/3)\subseteq Q$. The relation becomes $3(Q/3)= Q$ if $3 \mid \ell(Q)$.


\section{Landscape law: the general case and the self-improvement under the scaling condition}\label{sec:det}
In this section, we study Theorems \ref{thm:NNu-intro}, \ref{thm:NNu-intro-2}, and \ref{thm:NNu-intro-scaling}. Let us recall some of the notations first. Let $\Lambda\cong \llbracket 1,K \rrbracket^d$ be the periodic domain of side length $K$.  Let $N(\mu)$ be the (finite volume) integrated density of states (IDS) of $H$ on $\Lambda$, as defined in \eqref{eq:Ndef-intro}. 

Let $u=\{u_n\}$ be the landscape function of $H$ defined in the introduction.  For $\mu>0$, let $s(\mu)=\cl{\mu^{-1/2}}$ and let $N_u(\mu)$ be the landscape box counting function, defined with respect to the partition $\cP=\cP(s(\mu);\Lambda)$ as in \eqref{eq:Nudef-intro}, see Figures \ref{fig:Nu1d}, \ref{fig:Nu2d} for examples on $\Z^1,\Z^2$.
\begin{figure}
    \centering
    \includegraphics[width=0.8\textwidth]{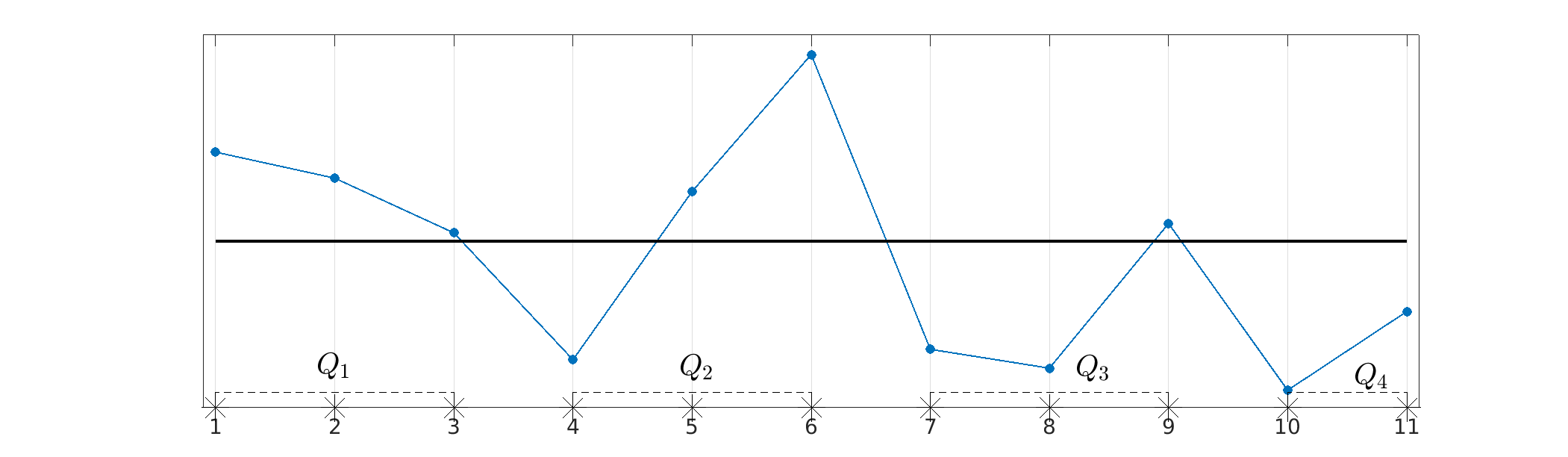}
    \caption{The effective potential $\{1/u_n\}_{n=1}^{11}$ is plotted in blue. The horizontal reference line  is $\mu=1/9$ (in black). The partition $\cP(3)=\{ Q_j\}_{j=1}^4$ contains four disjoint cubes. On $Q_2,Q_3,Q_4$, $\min (1/u_n)$ falls below $\mu$.}
    \label{fig:Nu1d}
\end{figure}
\begin{figure}
    \centering
    \includegraphics[width=0.6\textwidth]{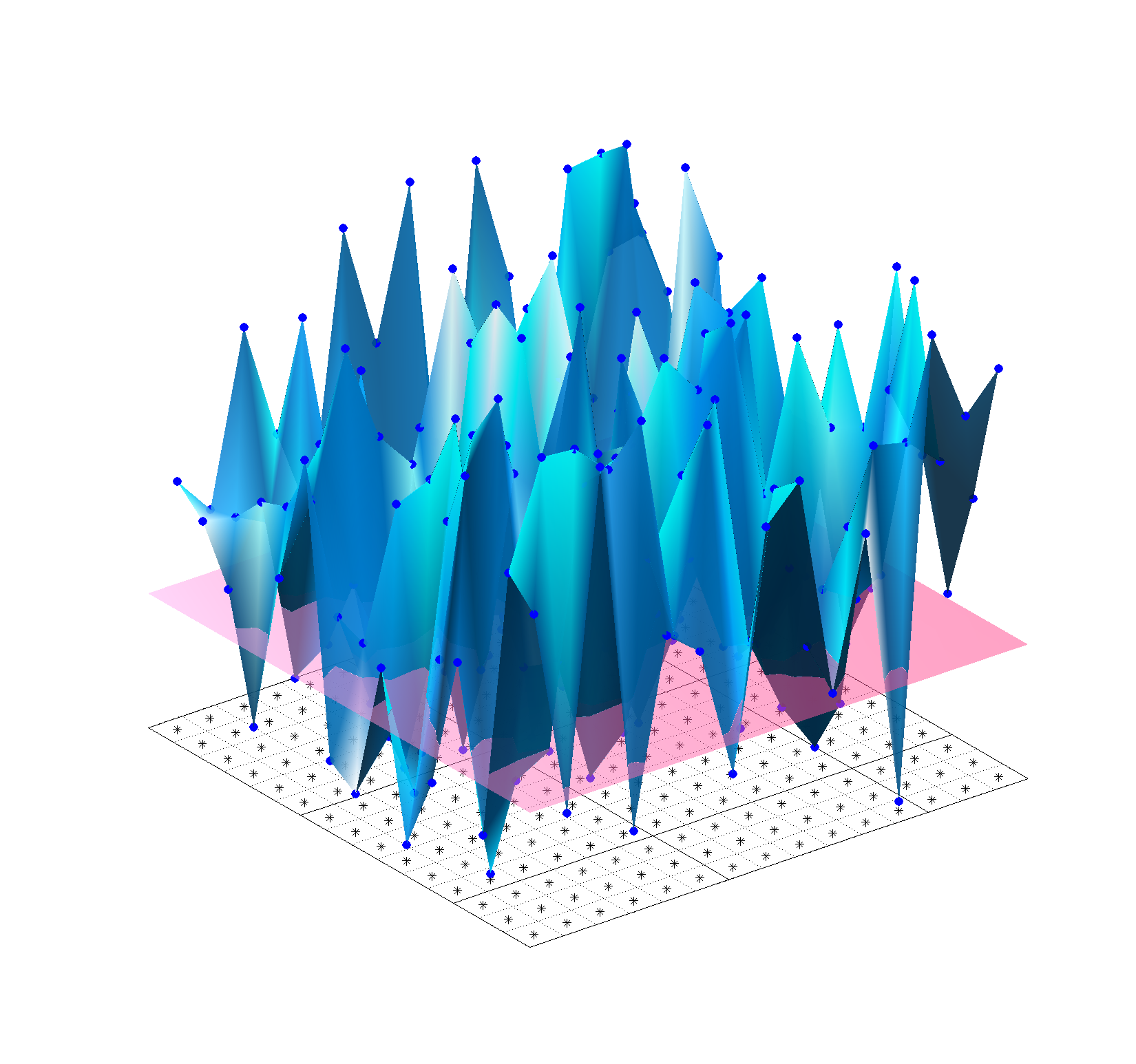}
    \caption{The effective potential (in blue)  
on a $\Z^2$ lattice. The reference energy  $\mu=1/36$ (in pink). The partition $\cP(6)$ contains $9$ disjoint boxes, four regular boxes have side length $6$, and five irregular boxes along the boundary.}
    \label{fig:Nu2d}
\end{figure}

\begin{remark}\label{rem:pat}
The landscape box counting function $N_u=N_u^{\cP}$ is defined with respect to  to the partition $\cP=\cP(r;\Lambda),r=\cl{\mu^{-1/2}}$. Consider any translation of $\cP$ by $a\in \llbracket 0,r-1 \rrbracket^d$, i.e., $\cP^a=\{a+Q,Q\in\cP\}$. Then  $\cP^a$ is also a partition of $\Lambda$ of size $r$ since $\Lambda$ is a periodic torus.  Each $a+Q\in \cP^a$ can be covered by finitely many cubes $Q'\in \cP$, and vice versa. The number of cubes from one partition needed to cover a cube from another partition is at most $3^d-1$.  Therefore, if we use $\cP^a$ to define a landscape box counting function $N_u^{\cP^a}$, the counting function will  differ by at most a factor of $3^{\pm d}$, i.e., for any $a\in \llbracket 0,r-1 \rrbracket^d$
\begin{equation*}
   3^{-d} N_u^{\cP^a}\le N_u^{\cP}\le 3^d N_u^{\cP^a}.
\end{equation*}
 Furthermore, if $r<r'$, then $\cP(r;\Lambda)$ is a finer partition than $\cP(r';\Lambda)$. Each $Q'\in \cP(r')$ can be covered by at most $(r'/r+2)^d$ cubes $Q\in \cP(r;\Lambda)$. Therefore, the number of $Q'\in \cP(r')$ such that $\min_{n\in Q'}\frac{1}{u_n}\,\le \mu$ will differ from the number of $Q\in \cP(r)$ such that $\min_{n\in Q}\frac{1}{u_n}\,\le \mu$ at most by a factor of $(r'/r+2)^d$. In other words, for $r<r'$, 
\begin{equation*}
 N_u^{\cP(r')}\le  N_u^{\cP(r)}\le \left ({r'}/{r}+2\right)^d\, N_u^{\cP(r')}.
\end{equation*} 
Based on the above discussion, we are allowed to estimate  $N_u(\mu)$  by either shifting the original partition $\cP\left(\cl{\mu^{-1/2}}\right)$ or tweaking the side length of the partition slightly. The change of the partition will  lead to  a different box-counting function, but the new counting function will differ from $N_u$ only by some multiplicative dimensional constants. This will be very useful in the proofs.
\end{remark}


\subsection{Upper bound. Proof of Theorem~\ref{thm:NNu-intro}}\label{subsec:ub}

Let $H=-\Delta+V$ be as in \eqref{eq:opH} acting on $\cH=\ell^2(\Lambda )\cong\R^{K^d}$. We denote by $\langle \, \cdot\, ,\,\cdot\,  \rangle$ the inner product on $\ell^2(\Lambda )$ induced by the usual one on $\R^{K^d}$. 

\vskip 0.08in 
\noindent {\bf Case I: $\mu<\frac{1}{4d}$}.   We will actually start with $\mu<1$, and then consider the rescaling $\frac{\mu}{4d}$ in the end.  
To get an upper bound for $N(\mu)$, it is enough to bound $\ipc{f}{Hf}$ from below on some subspace of $\cH$.  
For $r=\cl{\mu^{-1/2}}\le 2\mu^{-1/2}$ , let $\cP(r)=\cP(r;\Lambda)$ be the partition of side length $r$.  Let
\begin{equation*}
  {\mathcal F}:=\left\{Q \in \cP(r)\, :\ \ \min_{n\in Q }\frac{1}{u_n}\,\le \mu \right\} .
\end{equation*}
Let $S$ be the linear subspace of the vectors in $\cH$ whose average on each $Q \in \mathcal F$ is zero, i.e.,
\[S=\left\{f\in \cH\, :\,  \frac{1}{|Q |}\sum_{n\in Q }f_n=0,\ \ Q \in {\mathcal F} \right\}.\]
The subspace $S$ has $\cards{\mathcal F}$ many linear independent constraints since all $Q \in\cP$ are disjoint.  Therefore, $S$  has codimension $\cards{\mathcal F}$.

By the landscape uncertainty principle \eqref{eq:uncert},
\[\ipc{f}{Hf}=\sum_{n\in \Lambda }\, \Big(\|\nabla f_n\|^2+ v_n  f_n^2\Big)\, \ge\, \sum_{n\in \Lambda }\frac{1}{u_n}f_n^2. \]
We see that
\[\ipc{f}{Hf}\ge \sum_{n\in \Lambda }  \|\nabla f_n\|^2\ \  {\rm and}\ \ \ipc{f}{Hf}\ge\, \sum_{n\in \Lambda }\frac{1}{u_n}f_n^2, \]
which implies  
\[2\ipc{f}{Hf}\ge\sum_{n\in \Lambda }\, \Big(\|\nabla f_n\|^2+ \frac{1}{u_n}f_n^2\Big). \]
Therefore, for $f\in S$,  
\begin{equation}\label{eq:pf37}
    2\ipc{f}{Hf}\ge \sum_{Q \in \cP(r)}\sum_{n\in Q }\, \left(\|\nabla f_n\|^2+ \frac{1}{u_n}f_n^2 \right)\ge  \ \sum_{Q \in {\mathcal F}}\sum_{n\in Q }\, \|\nabla f_n\|^2\ +\ \sum_{Q \not\in {\mathcal F}}\sum_{n\in Q }\, \frac{1}{u_n}f_n^2.
\end{equation}
In the second sum, $1/u_n\ge \min 1/u_n>\mu$ since $Q \notin \cF$. Therefore, 
\begin{equation*}
\sum_{Q \not\in {\mathcal F}}\, \sum_{n\in Q }\, \frac{1}{u_n}\, f_n^2\, \ge \,  \mu\, \sum_{Q \not\in {\mathcal F}}\, \sum_{n\in Q }\, f_n^2   \ \  .
\end{equation*}
To bound the gradient term on the right-hand side  of \eqref{eq:pf37}, we need the discrete version of the Poincar\'e inequality (Lemma \ref{lem:PoincareZd} in Appendix \ref{sec:app}): for any cube $Q\in \cF $ of side length $\ell(Q )=r=\cl{\mu^{-1/2}}\le 2\mu^{-1/2}$, 
\begin{equation*}
   \sum_{n\in Q }\|\nabla f_n\|^2\ge \frac{2}{\ell(Q )^2  d}\, \sum_{n\in Q }\, (f_n- \bar f _{Q })^2\, \ge \,  \frac{2}{(2{\mu^{-1/2}) }^2\,  d}\, \sum_{n\in Q }f_n^2= \frac{\mu}{2d}\,  \sum_{n\in Q }f_n^2 . 
\end{equation*}
Notice that the last $Q\in \cP$ in each direction may not be a regular box of equal side length, the above estimate remains the same since the side length of the irregular box does not exceed $r=\cl{\mu^{-1/2}}$, see \eqref{eq:PoincareZd} in Lemma \ref{lem:PoincareZd}.

Putting these two parts together, one has for $f \in S$
\begin{align*}
    \ipc{f}{Hf}\ge\, \frac{1}{2}\, \sum_{Q \in \cP }\sum_{n\in Q }\, \left( \|\nabla f_n\|^2+ \frac{1}{u_n}f_n^2 \right)\, \ge\, &  \frac{1}{2}\, \sum_{Q \in {\mathcal F}} \frac{\mu}{2d} \, \sum_{n\in Q }f_n^2\ +\, \frac{1}{2}\, \mu \sum_{Q \not\in {\mathcal F}}\sum_{n\in Q}\, f_n^2  \\
 >\,  & \frac{\mu}{4d}\sum_{n\in \Lambda } f_n^2=\frac{\mu}{4d}\ipc{f}{f}. 
\end{align*}

Therefore,  using  the mini-max characterization of eigenvalues, the number of eigenvalues of $H$ below $\dfrac{\mu}{4d}$ is bounded from above by the codimension of the subspace $S$, which is equal to $\cards{ \mathcal F}$. Hence,
\begin{equation*}
    N\left(\frac{\mu}{4d}\right)\le \frac{\cards{\mathcal F}}{|\Lambda|}=N_u(\mu), \ \ {\rm for}\ \ \mu<1.
\end{equation*}
Equivalently, 
\begin{equation}\label{eq:N<Nu1}
    N\left(\mu \right)\le N_u(4d\mu), \ \ {\rm  for}\ \ \mu<\frac{1}{4d}. 
\end{equation}

\noindent {\bf Case II: $\mu\ge \frac{1}{4d}$.}   Similar to Case I, we will  work with $\mu\ge1$, and then consider the rescaling $\frac{\mu}{4d}$ in the end.   
The construction is similar to the previous case. Let 
\[{\mathcal F}=\left\{\, n\in \Lambda   :\ \ \frac{1}{u_n}\le \mu \right\}, \ \ S=\left\{\, f\in \cH\, :\, f_n=0\ \ {\rm if}\ \ n\in \cF\, \right\}.\]
 From the definition of $N_u$ in \eqref{eq:Nudef-intro} and the fact that $\cl{\mu^{-1/2}}=1$, we see that $N_u(\mu)= \cards{\mathcal F} /{|\Lambda|}$.

Due to \eqref{eq:uncert}, 
 \begin{equation*}
    \ipc{f}{Hf}\ge   \sum_{n\notin \cF}\, \frac{1}{u_n}\, f_n^2\,
    \ge  \sum_{n\notin \cF} \, \mu \,f_n^2 
    = \mu \, \sum_{n\in \Lambda}\, f_{n}^2 =\mu \ipc{f}{f}\ \  {\rm for \ all}\ f\in S, 
\end{equation*}
where we have used the fact that $f_n=0$ on $\cF$ for the first equality. 
Notice that in this case, we do not need the Poincar\'e inequality.

Therefore, for all $\mu\ge 1$, $
N\left({\mu}\right)\le N_u(\mu)$. 
Since $N(\mu)$ is non-decreasing, it implies that $N\left(\mu/(4d) \right)\le  N\left(\mu \right)\le N_u(\mu)$ for $\mu\ge 1$.  Equivalently, 
\begin{equation}\label{eq:N<Nu2}
    N\left(\mu \right)\le  N_u(4d\, \mu), \ \ {\rm for}\ \  \mu\ge \frac{1}{4d}.
\end{equation}

Combing \eqref{eq:N<Nu1} and \eqref{eq:N<Nu2}, we finish the proof for \eqref{eq:N<Nu-intro}. 
\ep


\subsection{General lower bound in the non-scaling case. Proof of Theorem~\ref{thm:NNu-intro-2}}\label{sec:non-scaling} Similarly to the upper bound, if one can bound $\ipc{f}{Hf}$ from above on a subspace of $\cH\cong\R^{K^d}$, then the eigenvalue counting function will be bounded from below by the dimension of this subspace.  

\begin{proof}
Let 
\[
    r= \cl{\left({c_H}/{32}\right)^{\frac{1}{4}}\mu^{-\frac{1}{2}}},
\]
where  $0<c_H<1$ is a dimensional constant given by the discrete Moser-Harnack inequality, see Lemma \ref{lem:MH-2}. 

Given $0<\alpha<({c_H/32})^{-1/4}/18$, let $R=\cl{\alpha^{-1}\mu^{-1/2}}\ge r$.  

\vskip 0.08 in
\noindent {\bf Case I:} $r\mid R \mid K$, i.e., $K=K_0R, \, R= R_0  r$ for some $K_0,R_0\in \N$. We will deal with the following three sub-cases for small, mild, and large $\mu$. 

{\bf Case I(a):} 
We first consider $\mu<\left({c_H}/{32}\right)^{\frac{1}{2}}$ and therefore $K\ge R\ge r\ge 2$.   
In this case, one has  $ \left({c_H/32}\right)^{\frac{1}{4}}\mu^{-\frac{1}{2}}\le r < 2 \left({c_H/32}\right)^{\frac{1}{4}}\mu^{-\frac{1}{2}}, \alpha^{-1}\mu^{-1/2}\le R< 2\alpha^{-1}\mu^{-1/2} $. Therefore,
 \begin{equation}\label{eq:R0}
\frac{1}{2}\, \alpha^{-1}\, ({32/c_H})^{1/4}  \le    R_0=\frac{R}{r}\le 2\alpha^{-1}\, ({32/c_H})^{1/4} . 
 \end{equation}

Consider the partition $\cP(R;\Lambda)$ of $\Lambda$ of side length $R$, and, for each $Q$ of side length $R=R_0r$, consider the finer partition $\cP(r;Q)$ of side length $r$. 
Clearly, the collection of all $q\in \cP(r;Q)$ for all $Q\in \cP(R;\Lambda)$ also forms a partition for $\Lambda $ of size $r$:
\begin{equation}\label{eq:pat-finer}
\cP(r;\Lambda)=\bigcup_{Q}\cP(r;Q)=\bigl\{q:\ q\in \cP(r;Q),\ Q\in \cP(R;\Lambda) \bigr\}. 
\end{equation}

Since $\alpha<({32/c_H})^{1/4}/18$, then $R/r=R_0\ge 9$ due to  \eqref{eq:R0}. 
For each $Q$, let $ \check q$ be a cube in $\cP(r;Q)$  such that ${\rm dist}_{\R^d}(\check q, c_Q)\le \sqrt d /2$, where $c_Q$ is the center of $Q$ in $\R^d$ and the distance is measured in $\R^d$.  We call such $ \check q$ a centric cube with $Q$. We see that a centric cube $\check q$ satisfies $3\check q\subset Q/3$, cf. \eqref{eq:3Q}, \eqref{eq:Q/3}, since $R_0\ge9$. The notion of centric cube   means  that we consider cubes in $\cP(r;Q)$ and $\cP(R;Q)$ as subsets of $\R^d$, and then pick $\check q$ to be a cube in $\cP(r;Q)$ closest to the center of $Q$ in $\R^d$. Note that the choice of a centric cube may not be unique. That would not effect the estimates below and the translation arguments in \eqref{eq:sumN012} due to Remark \ref{rem:pat}.

For any $0<\alpha<({32/c_H})^{1/4}/18$, let 
 \begin{equation}\label{eq:F'}
 {\mathcal F}'=\left\{Q\in \cP(R;\Lambda)\, :\ \ \min_{n\in \check q}\frac{1}{u_n}\,\le \,\mu 
 \quad {\rm and} \quad \min_{n\in  Q}\frac{1}{u_n}\,\ge \, \alpha^2\, \mu
 \right\}.
 \end{equation}
Given $Q\in \cP(R;\Lambda)$, let $Q/3\subset Q$ be the middle third of $Q$ as usual. Let $\chi^Q=\{\chi^Q_n\}_{n\in \Lambda}\in \cH$ be a discrete  cut-off function supported on $Q$ and such that
 \begin{equation}\label{eq:chiQ}
 \chi^Q_n =1  \ {\rm if}\ n\in Q /3,\quad  \chi^Q_n =0 \ {\rm if}\ n \notin Q,\quad  \chi^Q_n  \in (0,1)\,  {\rm  otherwise},
 \end{equation}  and 
 \begin{equation} \nonumber
 |\chi^Q_{n+e_i}-\chi^Q_{n}| \le \frac{3}{R}, \ {\rm if}\ n,n+e_i\in Q,\quad |\chi^Q_{n+e_i}-\chi^Q_{n}|
 =0, \ {\rm if}\ n \ {\rm or}\ n+e_i \notin Q, 
 \end{equation}
for $i=1,\cdots,d$. 
 Such a cut-off function is the discrete analogue of a smooth bump function in the continuous case. We include an explicit construction in Appendix \ref{sec:cutoff} for reader's convenience. 
 
Let $S'$ be the linear subspace of $\cH\cong\R^{K^d}$ which is spanned by the cut-offs of $u$ to each $Q $ in $\mathcal F '$. More precisely,  
we define
\begin{equation*}
   S'= {\rm span}\big\{u^{Q }=\{u^{Q }_n\}_{n\in \Lambda}\in\cH:\ u^{Q }_n=u_n\, \chi^Q_n, \ Q \in{\mathcal F}' 
 \big\}.
\end{equation*}
The subspace $S'$ has dimension $\cards{\mathcal F'}$ since all $Q $ are disjoint and $u_n>0$. 

We aim to estimate ${\ipc{u^{Q }}{H\,u^{Q }}}/{\ipc{u^{Q }}{u^{Q }}}$ from above for each $u^{Q }$ in $S'$. 
First, by the landscape uncertainty principle \eqref{eq:uncert}, 
\begin{align}
    \ipc{u^{Q }}{H u^{Q }} 
    =&\sum_{n\in  \Lambda }\, \sum_{1\le i\le d} u_{n+e_i} u_n\left(\chi^Q_{n+e_i}-\chi^Q_{n}\right)^2
    \, +\, \sum_{n\in  Q }\,\frac{1}{u_n}\left(u_n\chi^Q_{n}\right)^2   \nonumber  \\
    \le &  \sum_{1\le i\le d}\sum_{n,n+e_i\in  Q }\, u_{n+e_i} u_n\left(\frac{3}{R}\right)^2 \, +\, \sum_{n\in  Q }\,u_n  \le  {9d}\, {R^{d-2}} \sup_{Q }u_n^2+R^d \sup_{Q }u_n. \label{eq:uHu-upp}
\end{align}
On the other hand, recall that $3\check q \subset Q /3$, so that
\begin{equation*}
    \ipc{u^{Q }}{u^{Q }} =\sum_{n\in  Q }\,\left(u_{n}\chi^Q_n\right)^2\ge \sum_{n\in  Q /3}\,u^2_{n}\ge \sum_{n\in 3\check q }\,u^2_{n}. 
\end{equation*} 
Since $-(\Delta u)_n=1-v_nu_n\le 1$, applying the discrete Moser-Harnack inequality in  Lemma \ref{lem:MH-2},  Appendix \ref{sec:Moser-Harnack}, to $u_n$ on the smaller cubes $\check q $ with $\ell(\check q )=r$, one has 
\begin{equation}\label{eq:uu-lower}
    \sum_{n\in 3\check q }\,u^2_{n}\ge  r^d\Bigl( c_H\sup_{\check q }u^2_n-r^4\Bigr). 
\end{equation}
By the definition of ${\mathcal F}'$ in \eqref{eq:F'}, one has 
\begin{equation}\label{eq:u-sup-alpha}
\sup_{n\in Q }u_n\le \alpha^{-2}\mu^{-1},\ \ \  \sup_{n\in \check q }u_n\ge \mu^{-1}.
\end{equation}
Notice that $r=\cl{\left({c_H/32}\right)^{\frac{1}{4}}\mu^{-\frac{1}{2}}} $ implies $\left({c_H/32}\right)^{\frac{1}{4}}\mu^{-\frac{1}{2}}\le r<2\left({c_H/32}\right)^{\frac{1}{4}}\mu^{-\frac{1}{2}}$, i.e.,
\begin{equation}\label{eq:r-18}
r^4\le \frac{1}{2}{c_H}\mu^{-2}, \ \ r^{-2}\le \sqrt{{32/c_H}}\,\mu .
\end{equation} Therefore, putting \eqref{eq:uHu-upp} and \eqref{eq:uu-lower} together and using $R=R_0r$, 
\begin{align}
    \frac{\ipc{u^{Q }}{H u^{Q }}}{\ipc{u^{Q }}{u^{Q }}}
    \le&\, \frac{{9d}\,{R^{d-2}}  \sup_{Q }u_n^2+R^d \sup_{Q }u_n }{r^d\left( c_H\sup_{\check q }u^2_n-r^4\right)} \le \, \frac{9d\,{R^{d-2}_0r^{d-2}} \alpha^{-4}\mu^{-2}+R^d_0\,r^d\,  \alpha^{-2}\mu^{-1}}{ r^d\left( c_H\mu^{-2}-{\frac{1}{2}{c_H}\mu^{-2}}\right)}  \nonumber\\
     \le & \,C\left(R^{d-2}_0\,   \alpha^{-4}+R^d_0\,  \alpha^{-2}\right)\mu 
    \le  \,C_2\, \alpha^{-d-2}\,  \mu ,    \label{eq:temp01}
\end{align}
where $C_2=C_2(d,c_H)$. In the last line, we used \eqref{eq:R0} with $R_0\lesssim \alpha^{-1}$.

Then, by orthogonality of the $u^{Q }$ for $Q \in{\mathcal F}'$, we get for 
\begin{equation}\label{eq:cN0}   \cN_0:=\card{{\textrm{eigenvalues}}\ \lambda\ {\rm of}\  H\ {\rm such \ that}\ \lambda\le C_2\alpha^{-d-2}\, \mu}
\end{equation}
the estimate
\begin{multline} \label{eq:N012}
   \cN_0
   \ge \, {\rm Card}(\cF') \\
   = {\rm Card}\left\{Q \in \cP(R;\Lambda)\, :\ \ \min_{n\in \check q }\frac{1}{u_n}\,\le \,\mu 
   \quad {\rm and} \quad \min_{n\in  Q }\frac{1}{u_n}\,\ge \, \alpha^2\, \mu
   \right\}   \\
\ge\, {\rm Card} \left\{ \ Q \in \cP(R;\Lambda)\, :\ \ \min_{n\in \check q }\frac{1}{u_n}\,\le \,\mu 
 \right\}\\ -{\rm Card}\left\{ \ Q \in \cP(R;\Lambda)\, : \min_{n\in  Q }\frac{1}{u_n}\,\le  \, \alpha^2\, \mu
 \right\}   
 :=\,\cN_1-\cN_2.
\end{multline}

Given an integer $j$, $|j|\le \fl{R_0/2}$, we consider a translation $T^{i,j}:\Lambda\to \Lambda$, by the vector $jre_i$, i.e., $T^{i,j}(n)=n+jre_i$ for any $n\in \Lambda$. For the partition $\cP(R;\Lambda)$, denote by $\cP^{i,j}(R;\Lambda)$ the partition translated by $T^{i,j}$. Recall that $\cP(r;\Lambda)$ is the finer partition of side length $r$, see \eqref{eq:pat-finer}. Denote by $\cP^{i,j}(r;\Lambda)$ the translation of $\cP(r;\Lambda)$, which again is a refinement for $\cP^{i,j}(R;\Lambda)$.  For any $\check q  \subset Q \in \cP(R;\Lambda )$, it is easy to check that 
$
Q \subset \bigcup _{i,j}T^{i,j}(\check q ).
$
 In other words, the collection of all $T^{i,j}(\check q )$ will cover the entire $Q $, provided enough tanslations of $\check q$ (at most $\fl{R_0/2}$ many).

 For each translated centric cube $T^{i,j}(\check q )$ in the corresponding $T^{i,j}(Q )\in \cP^{i,j}(R;\Lambda)$, we repeat the construction of $\cF'$ and $\cS'$ starting from \eqref{eq:F'}. By exactly the same argument as for \eqref{eq:N012},  
\begin{equation}\label{eq:N012-trans}
    \cN_0\ge  {\cN}^{i,j}_1- {\cN}^{i,j}_2,
\end{equation}
where $\cN_0$ is the same as in \eqref{eq:cN0} since the eigenvalue counting will be the same for all the translations, and 
\begin{align*}
  {\cN}^{i,j}_1=& \,  {\rm Card} \left\{ \ T^{i,j}(Q )\in \cP^{i,j}(R;\Lambda)\, :\ \ \min_{n\in T^{i,j}(\check q )}\frac{1}{u_n}\,\le \,\mu 
 \right\}, \\
 {\cN}^{i,j}_2=&\,{\rm Card}\left\{ \ T^{i,j}(Q )\in \cP^{i,j}(R;\Lambda)\, : \min_{n\in  T^{i,j}(Q )}\frac{1}{u_n}\,\le  \, \alpha^2\, \mu
 \right\}.
\end{align*}
Recall that $\check q $ has side length $r$ and is located \red{near} the center of each $Q $. We repeat the above process exactly $\sigma$ times so that $\bigcup_{j,i}\, T^{i,j}(\check q )\,=Q $, and therefore, $\bigcup_Q\bigcup_{j,i}\, T^{i,j}(\check q )=\Lambda$, which is exactly the fine partition $\cP(r;\Lambda)$ of the entire domain. One translated $T^{i,j}(\check q )$ corresponds to exactly one small cube $q$ in the original partition  $\cP(r;\Lambda)$. The number $\sigma$ of the translations we need can be bounded from above by $\sigma \le \left(2\fl{R_0/2}\right)^d\le  R_0^d\le \,  C\alpha^{-d}$.  Notice that for all $j,i$, we have $ {\cN}^{i,j}_2<\, C  \cN_2$ for some dimensional constant $C$  because of Remark \ref{rem:pat}. 
Then, summing up \eqref{eq:N012-trans} over all possible translations, 
\begin{align}
  \sigma  \cN_0 \ge & \, \sum_{i,j}{\cN}^{i,j}_1\ -\ \sigma  {\cN}^{i,j}_2 \label{eq:sumN012}\\
  = & \, \sum_{i,j} {\rm Card} \left\{ \ T^{i,j}(Q )\in \cP^{i,j}(R;\Lambda)\, :\ \ \min_{n\in T^{i,j}(\check q )}\frac{1}{u_n}\,\le \,\mu 
  \right\} -\ \sigma  {\cN}^{i,j}_2 \nonumber \\
   = & \, \sum_{j,i} {\rm Card} \left\{ \ T^{i,j}(\check q ) :\ \ \min_{n\in T^{i,j}(\check q )}\frac{1}{u_n}\,\le \,\mu 
  \right\} -\ \sigma  {\cN}^{i,j}_2 \nonumber  \\
  \ge & \,  {\rm Card} \left\{  q \in \cP(r;\Lambda)  :\ \ \min_{n\in  q }\frac{1}{u_n}\,\le \,\mu 
  \right\} -\ C\sigma  {\cN}_2  \nonumber .
\end{align}
Therefore,  the upper bound on the number of the translations $\sigma\le \, C \alpha^{-d}$  implies
\begin{equation}\label{eq:tmp144}
    C \alpha^{-d}  \cN_0 \ge   {\rm Card} \left\{  q \in \cP(r;\Lambda)  :\ \ \min_{n\in  q }\frac{1}{u_n}\,\le \,\mu 
\right\} -\wt C\alpha^{-d}   \cN_2. 
\end{equation}

Notice that the partition in the counting of $\cN_2$ is of side length $R=\cl{\alpha^{-1}\mu^{-1/2}}=\cl{(\alpha^{2}\mu)^{-1/2}}$, which is exactly the side length needed in the definition of $N_u(\alpha^2\mu)$.  On the other hand, the partition in the counting of $\cN_1$ is of side length $r=\cl{\left({c_H/32}\right)^{\frac{1}{4}}\mu^{-\frac{1}{2}}}$. The side length  needed in the definition of $N_u(\mu)$ should be $r'=\cl{\mu^{-1/2}}$, which is larger than $r$  used in $\cN_1$. But the two counting functions defined by $\cP(r)$ or $\cP(r')$ will only differ by a dimensional factor since $1\le r'/r\le 2({{32/c_H}})^{1/4}$, see Remark \ref{rem:pat}. Therefore,  
 \[ {\rm Card} \left\{ q \in \cP(r;\Lambda)\, :\ \ \min_{n\in  q }\frac{1}{u_n}\,\le \,\mu 
 \right\}\ge C^{-1}  {\rm Card} \left\{ q'\in \cP(r';\Lambda)\, :\ \ \min_{n\in  q'}\frac{1}{u_n}\,\le \,\mu 
 \right\}.\]
Then by \eqref{eq:tmp144}
\begin{align*}
   \alpha^{-d}\cN_0\ge
  c \, {\rm Card} \left\{ q'\in \cP(r';\Lambda)\, :\ \  \min_{n\in  q'}\frac{1}{u_n}\,\le  \,\mu 
 \right\} \\ 
 & \hspace{-3cm}-C  \alpha^{-d}   {\rm Card}\left\{ \ Q \in \cP(R;\Lambda)\, : \min_{n\in  Q }\frac{1}{u_n}\,\le  \, \alpha^2\, \mu
 \right\},
\end{align*}
which implies that 
\begin{equation}\label{eq:eq0}
    N(C_2\alpha^{-d-2}\mu)\ge c \alpha^dN_u(\mu)-C N_u(\alpha^2\mu), \ {\rm for \ all}\ 0<\mu<(c_H/32)^{1/2}, 
\end{equation}
provided that $0<\alpha<({32/c_H})^{1/4}/18$, $K\ge \cl{\alpha^{-1}\mu^{-1/2}}$.  

 We notice that the construction above also needs the entire domain to be large enough, i.e., $K\ge R\ge \alpha^{-1}\mu^{-1/2}$. 
The restrictions on $K$ can be removed easily. If $3r<K<\cl{\alpha^{-1}\mu^{-1/2}}$, then we can repeat the above construction by setting $R=K$ directly,  the proof for \eqref{eq:eq0} is exactly the same. If $K\le 3r\lesssim \mu^{-1/2}$ then there are at most $3^d$ boxes in the partition. Therefore, $N_u(\mu)\le 3^d/K^d$.  On the other hand, the argument for \eqref{eq:uHu-upp} can be used to show that the ground state eigenvalue $E_0$ of $H$ is bounded from above by $E_0\le C\mu$, with a dimensional constant $C$. Therefore, $N(C\mu)\ge 1/K^d \gtrsim N_u(\mu)$. Then \eqref{eq:eq0} holds trivially by picking $\alpha$ small, and the smallness only depends on the dimension.

{\bf Case I(b):}
Next we consider $\left({c_H}/{32}\right)^{\frac{1}{2}}\le \mu\le (3\alpha)^{-2}$. Note that the range of $\mu$ requires $\alpha<(32/c_H)^{1/4}/3$, which is fulfilled by the assumption of $\alpha$.  In this case, the side length $\ell(q)=r= \cl{\left({c_H}/{32}\right)^{\frac{1}{4}}\mu^{-\frac{1}{2}}}=1$ and $\ell(Q)=R=\cl{\alpha^{-1}\mu^{-1/2}}\ge 3$.  For each $Q\in \cP(R)$, we pick $\check q$ to be a centric cube with $Q$, then   construct $\cF',S'$ in the same way as  \eqref{eq:F'}. The upper bound \eqref{eq:uHu-upp} for $\ipc{u^Q}{H^Q}$ remains the same.  To bound $\ipc{u^Q}{u^Q}$ from below, one has trivially 
$
\ipc{u^Q}{u^Q}\ge \sum_{Q/3}u_n^2 \ge \sum_{\check q}u_n^2 \ge \mu^{-2}$ since $\check q \subset Q/3$.  The bounds on $\mu$ and the definition of $R$ imply 
\begin{equation*}
R\le 2\alpha^{-1}\mu^{-1/2}\le \, 2\left({c_H}/{32}\right)^{-\frac{1}{4}}\, \alpha^{-1}, \ {\rm and }\ R^{-2}\le \,  (\alpha^{-1}\mu^{-1/2})^{-2}= \alpha^{2}\mu.
\end{equation*}
Hence, 
\begin{align}
\frac{\ipc{u^Q}{H^Q}}{\ipc{u^Q}{u^Q}}\le \frac{{9d}{R^{d-2}}  \alpha^{-4}\mu^{-2}+R^d  \alpha^{-2}\mu^{-1} }{\mu^{-2}} 
\le & \, {9d}({R^{d}}\alpha^{-4})  R^{-2}+(R^d\alpha^{-2})  \mu   \nonumber \\
\le & \, {C}_2 \, \alpha^{-d-2}\, \mu, \label{eq:r=1R>3}
\end{align}
where the constant $ {C}_2$ only depends on $d$ and $c_H$. We obtain the bound as in \eqref{eq:temp01}, and therefore the same bound \eqref{eq:eq0} holds for all $0< \mu\le (3\alpha)^{-2}$. 

Repeating the proof of \eqref{eq:eq0} for $\wt \mu=c_1\alpha^{d+2}\mu \le (3\alpha)^{-2}$ with $c_1=C^{-1}_2$, we obtain
\begin{equation} \label{eq:mu<1}
    N( \mu)\ge c {\alpha^d}N_u(c_1 \alpha^{d+2} \mu)-C N_u(c_1 \alpha^{d+4} \mu), \ {\rm for \ all}\  \mu<\frac{1}{9}c_1\, \alpha^{-d-4}, 
\end{equation}
provided $0<\alpha<({32/c_H})^{1/4}/18$. \\

{\bf Case I(c):}
The remaining case is $ \mu\ge  (3\alpha)^{-2}$.  Since $\alpha<({32/c_H})^{1/4}/18$, then $\mu\ge  \left({c_H}/{32}\right)^{\frac{1}{2}}$.  In this case, the range of $\mu$ implies  $r= \cl{\left({c_H}/{32}\right)^{\frac{1}{4}}\mu^{-\frac{1}{2}}}=1$ and $R=\cl{\alpha^{-1}\mu^{-1/2}}< 3$. We construct $\cF',S'$ with cube $Q$ of side length $\wt R=9$ instead. Note that this will not change the counting for $\cN_0$ and $\cN_1$. The change will only result a different counting for $\cN_2$, which we denote by $\wt \cN_2$ the new counting using $\ell(Q)=\wt R=9$. Since $1\le \wt R/R\le 9$, one has $\wt \cN_2 \le \cN_2 \le 9^d \wt \cN_2$,  due to Remark \ref{rem:pat}.  
Similar to \eqref{eq:r=1R>3}, one has 
\begin{equation*}
\frac{\ipc{u^Q}{H^Q}}{\ipc{u^Q}{u^Q}}\le \, {9d}{9^{d-2}}\, \alpha^{-4}+9^d\alpha^{-2}  \mu \le \, ({9d}{9^{d-1}}+9^d)\, \alpha^{-2}\,  \mu,
\end{equation*}
since $\alpha^{-2}\le 9\, \mu$. Then we repeat the arguments for \eqref{eq:N012-trans}-\eqref{eq:eq0} using translations. Notice that the number $\sigma$ of translations needed in this case is bounded from above by a dimensional constant $\sigma\le C$. We obtain instead 
\begin{equation}\label{eq:eq00}
    N(C'_2\alpha^{-2}\mu)\ge \, c N_u(\mu)-C N_u(\alpha^2\mu)\ \ \ {\rm for \ all}\ \mu\ge (3\alpha)^{-2}. 
\end{equation}

Repeating the proof of \eqref{eq:eq00} for $\wt \mu=c'_1\alpha^{d+2}\mu \ge (3\alpha)^{-2} $ where $c'_1=(C'_2)^{-1}$,  we reach
\begin{equation}\label{eq:mu>1}
    N(\mu)\ge \, c N_u(c'_1 \alpha^{2}\mu)-C N_u(c'_1 \alpha^{4}\mu) \ \ \ {\rm for \ all}\  \mu>\frac{1}{9}c'_1\, \alpha^{-4}, 
\end{equation}
provided $0<\alpha<({32/c_H})^{1/4}/18$. 

Finally, we require further that $\alpha<(c_1/c_1')^{1/d}$ so that $c_\ast \alpha^{-4}<\frac{1}{9}c_1\, \alpha^{-d-4}$ where $c_\ast=\frac{1}{9}c'_1$. Therefore, the estimates \eqref{eq:mu<1} and \eqref{eq:mu>1} cover $\mu\le c_\ast \alpha^{-4}$ and $\mu>c_\ast\alpha^{-4}$ respectively. This completes the proof of Theorem \ref{thm:NNu-intro-2}. \\

\vskip 0.08in

{\noindent{\bf Case II:}}  either $r\not\mid R$ or $ R\not\mid K$, 
where 
$  r=\cl{\left({c_H}/{32}\right)^{1/4}  \mu^{-1/2}}$ and  $R=\cl{\alpha^{-1}\mu^{-1/2}}$ are the same in Case I. Without loss of generality we assume that 
\[
    K=(K_0-1)\,R+\wt R , \,0<\wt R <R, \ \ {\rm and}\ \  R=(R_0-1)\,r+\wt r, \,0<\wt r<r .
\]
The other two cases, where either $\wt R=R$ or $\wt r =r$, are similar. Recall the construction of the partition for $\cP(R;\Lambda)$ and $\cP(r;Q )$. In the last row and column of each direction we need to use a  rectangular box instead of a cube. We denote the regular cube of side length $R$ or $r$ still by $Q $ and $q $, and denote the remaining special rectangular boxes by $\wt Q $ and $\wt q $, whose side lengths are $\wt R$ and $\wt r$, respectively, in at least one direction, and write
\begin{equation}\label{eq:pat3}
    \cP(R;\Lambda)=\{Q \}\cup \{\wt Q \}, \ \ \cP(r;Q )=\{q \}\cup \{\wt q \}.
\end{equation}
Notice in this case, the union 
$\wt{\cP}=\bigcup_{Q\in \cP(R;\Lambda)} \cP(r;Q ) $ is not the original partition $\cP(r;\Lambda)$, but it is a finer one. Therefore, the counting function defined through $\wt{\cP}$ will be bounded from below by the counting function defined through $\cP(r;\Lambda)$.

In this case, we define $\cF'$ only using the regular cubes $Q $ and ignoring all the $\wt Q $. Then by exact the same construction, we obtain \eqref{eq:N012}, i.e., $\cN_0\ge \cN_1-\cN_2$. Next, we need to translate the partition $\cP(R;\Lambda)$ and $\cP(r;Q )$ by vectors of length $r$ several steps  in each direction, so that the centric small cubes $\check q $ can cover the large cubes $Q $. In the previous case, we needed at most $j\sim \fl{R_0/2}$ steps in each direction. In Case II, we want to 
continue the translation up to $\wt j\sim \fl{2R_0}$ steps in each direction, where the total number of translations in all directions is at most $\wt {\sigma}\le (2  2R_0)^d \le  C  \alpha^{-d}$. By doing this, the translated centric cubes $T^{i,j}(\check q )$ will cover $3Q $. In particular, they will cover all the irregular boxes $\wt Q $ near the boundary of the domain.  Notice that for the regular cubes $Q $,  translations up to $\wt j\sim \fl{2R_0}$ steps will cause an overlap, which leads to an overestimate in the corresponding sum $\sum_{i,j} \cN^{i,j}_1$ as in \eqref{eq:sumN012}.  But since $\wt j \le 2R_0$ and the $T^{i,j}(\check q )$  are contained in $5Q $ for each $Q$, the over-counting will be at most $5^d$ times more. In conclusion, we can obtain 
\begin{align*}
 & \wt \sigma \, \card{{\textrm{eigenvalues}}\ \lambda\ {\rm of}\  H\ {\rm such \ that}\ \lambda\le C_2\alpha^{-d-2}\, \mu} \\ \quad &
  \ge     \sum_{j,i} {\rm Card} \left\{  T^{i,j}(\check q ) : \min_{n\in T^{i,j}(\check q )}\frac{1}{u_n}\,\le \,\mu 
  \right\}   - \wt \sigma \, {\rm Card}\left\{ \ Q \in \cP(R;\Lambda)\, : \min_{n\in  Q }\frac{1}{u_n}\,\le  \, \alpha^2\, \mu
 \right\} \\ \quad &
  \ge    \frac{1}{5^d}{\rm Card} \left\{  q\in \wt {\cP}  :  \min_{n\in  q}\frac{1}{u_n}\,\le \,\mu 
  \right\} -\ \wt \sigma \, {\rm Card}\left\{ \ Q\in \cP(R;\Lambda)\, : \min_{n\in  Q}\frac{1}{u_n}\,\le  \, \alpha^2\, \mu
 \right\} .
\end{align*}
Note that in the last line the collection of pertinent cubes $q$ or $Q$ already includes the irregular boxes $\wt q $ or $\wt Q $ respectively. Together with the fact that $\wt \cP$ is finer than $\cP(r;\Lambda)$, we obtain that 
\begin{align*}
 & \alpha^{-d}  \card{{\textrm{eigenvalues}}\ \lambda\ {\rm of}\  H\ {\rm such \ that}\ \lambda\le C_2\alpha^{-d-2}\, \mu}  \\
  & \quad  \ge  {5^{-d}}{\rm Card} \left\{  q\in \cP(r;\Lambda)  :   \min_{n\in  q}\frac{1}{u_n} \le  \mu 
  \right\} -  C \alpha^{-d}  {\rm Card}\left\{  Q\in \cP(R;\Lambda)  : \min_{n\in  Q}\frac{1}{u_n} \le  \alpha^2\, \mu
 \right\},  
\end{align*}
which implies
$  N(C_2\alpha^{-d-2}\mu)\ge c  \alpha^dN_u(\mu)-C  N_u(\alpha^2\mu)$.  This is the same estimate as we obtained in \eqref{eq:eq0}. Replacing $\mu$ by $C_2^{-1}\alpha^{d+2}\mu$ as in the remaining arguments in Case I, we can obtain \eqref{eq:mu<1} and \eqref{eq:mu>1} in a similar manner.  
\end{proof}

\subsection{Lower bound in the scaling case. Proof of Theorem~\ref{thm:NNu-intro-scaling}}\label{sec:scaling}
Let 
\begin{equation}\label{eq:C''}
R=\cl{\left(\frac{1}{2C''}\right)^{\frac{1}{4}}\mu^{-\frac{1}{2}}}, 
\end{equation}
where $C''\ge1$ is a dimensional constant that will be specified later. 

We need to consider two cases, corresponding to $R \mid K$ and $R \not \mid K$. Similar to  the arguments in the previous subsection, the latter can be be reduced to the former by a translation argument.  For simplicity, we will only deal with the case $K=K_0  R$ for some $K_0\in \N$. Also, we first assume that $\mu$ is small and $R\ge 3$, so that the cube of side length $R$ is large enough  to construct cut off functions as in \eqref{eq:chiQ}. Otherwise, when $\mu$ is large and $R$ is small, we start with $\wt R=9$ and tweak the dimensional constants in the end by  Remark \ref{rem:pat} as in the previous subsection. 

For $\left(\frac{1}{2C''}\right)^{\frac{1}{4}}\mu^{-\frac{1}{2}}\ge3$, i.e., $\mu<\frac{1}{9}\left(\frac{1}{2C''}\right)^{-\frac{1}{2}}$, one has $R\ge3$. We consider the partition $\cP(R;\Lambda)$ consisting cubes of side length $R$ as usual.
Let 
\begin{equation*}
{\mathcal F}''=\left\{Q \in \cP(R;\Lambda)\, :\  \min_{n\in  Q }\frac{1}{u_n}\,\le \, \mu
\right\}, \   \ 
S''={\rm span}\left\{u^{Q }\in \cH:\,  
u^{Q }_n=u_n  \chi^Q_n,\, Q \in{\mathcal F}'' \right\}  ,
\end{equation*}
where $\chi^Q=\{\chi^Q_n\}$ is the cut-off function  as in \eqref{eq:cutoff} on each $Q $. The dimension of $S''$ equals $\cards{\mathcal F''}$.

Our goal, once again, is to establish estimates similar to \eqref{eq:temp01}. 
The upper bound for $\ipc{u^{Q }}{H u^{Q }}$ remains the same as we obtained in \eqref{eq:uHu-upp}:
\begin{equation}
\ipc{u^{Q }}{H u^{Q }}
\le {9d}\,{R^{d-2}}  \sup_{Q }u_n^2+R^d \sup_{Q }u_n .\label{eq:doub-pf1}
\end{equation}

It remains to obtain the lower bound for $\ipc{u^{Q }}{ u^{Q }}$. First, the Moser-Harnack inequality \eqref{eq:MH-inhomo} implies that
\begin{equation*}
\sum_{n\in 3Q }\,u^2_{n}\ge  R^d\Bigl( c_H  \sup_{Q }u^2_n-R^4\Bigr).
\end{equation*}
Then we apply the scaling condition \eqref{eq:scaling} twice, to write
\begin{align*}
\sum_{n\in 3Q }\,u^2_{n}\le C_S \Bigl(\sum_{n\in Q }u_n^2+R^{d+4}\Bigr)\le&\, C_S \left(C_S \Bigl(\sum_{n\in Q /3}u_n^2+(R/3)^{d+4}\Bigr)+R^{d+4}\right)\\
\le&\, C_S ^2\sum_{n\in Q /3}u_n^2+C_S 'R^{d+4}.
\end{align*} 
Hence,
\begin{equation*}
\sum_{n\in Q /3}u_n^2\ge C_S ^{-2}c_H  R^d\Bigl(\sup_{Q }u^2_n-C''  R^4\Bigr),
\end{equation*}
where $C''$ depends on $d$, $c_H$ and $C_S$. By the choice of $R$ in \eqref{eq:C''},  
\begin{equation*}
\left(\frac{1}{2C''}\right)^{\frac{1}{4}}\mu^{-\frac{1}{2}}\le R\le 2\left(\frac{1}{2C''}\right)^{\frac{1}{4}}\mu^{-\frac{1}{2}},
\end{equation*}
and by definition for $Q \in {\mathcal F}''$, $\sup_{Q }u^2_n\ge \mu^{-2}$. Then we have 
\begin{equation*}
\frac{1}{2}\sup_{Q }u^2_n\ge \frac{1}{2}\mu^{-2}\ge \frac{1}{2}  (2C'')  R^4=   C''R^4.
\end{equation*}
Therefore, 
\begin{align*}
\ipc{u^{Q }}{ u^{Q }}\, =\, \sum_{n\in Q }\, (u_n\, \chi_n^{Q })^2\, \ge \, \sum_{n\in Q /3}\, u_n^2\, \ge\, &\, C_S ^{-2}\, c_H\,   \, R^d\,   \, \Bigl(\sup_{Q }u^2_n-C''R^4\Bigr)\\
\, \ge\, & \, C_S  ^{-2}\, c_H\,   \, R^d\,   \, \Bigl(\, \frac{1}{2}\, \sup_{Q }u^2_n \Bigr).
\end{align*}

Putting the upper and lower bounds together, we have that 
\begin{align*}
 \frac{\ipc{u^{Q }}{H u^{Q }}}{\ipc{u^{Q }}{u^{Q }}}
\, \le\, &\, \frac{\, {9d}\, {R^{d-2}}\,   \sup_{Q }\, u_n^2\, +\, R^d\,  \, \sup_{Q }u_n}{\frac{1}{2}\, R^d\, C_S  ^{-2}\, c_H\, \sup_{Q }u^2_n}\\
\, =\, &\, C'_1\, R^{-2}\, +\, C'_2\, \min_{Q}\frac{1}{u_n}\le \, C'_1\,(4\sqrt{2C''}\, \mu\, )\, +C'_2\, \mu\, := \, C_3\, \mu\, ,
\end{align*}
where $C_3=\, 4\, C'_1\,\sqrt{\, 2C''}\, +C'_2$ depends only on $c_H,C_S $ and the dimension $d$. Therefore,
\begin{equation*}
\card{ {\textrm{eigenvalues}}\ \lambda\ {\rm of}\  H\ {\rm such \ that}\ \lambda\le C_3\mu}\ge \cards{{\mathcal F}''}.
\end{equation*}

 Notice that in the definition of ${\mathcal F}''$, the side length of the cube $R$ is smaller than the side length $\cl{\mu^{-1/2}}$ required in the definition of $N_u(\mu)$. The box counting using $\cP(R;\Lambda)$ can be bounded from below by the box counting using $\cP(\cl{\mu^{-1/2}};\Lambda)$, due to Remark \ref{rem:pat}. The above estimates also require $3Q\subset \Lambda$, i.e., $K\gtrsim \mu^{-1/2}$. The restriction on $K$ can be removed exactly in the same way as for the non-scaling case, by multiplying counting functions by a dimensional constant $\wt c$. Therefore, we obtain $N(C_3\mu)\ge \, \wt cN_u(\mu) $ for all $0<\mu<\frac{1}{9}\left(\frac{1}{2C''}\right)^{-\frac{1}{2}}:=b$. 
Equivalently, one has that   
\begin{equation}\label{eq:mu<C3b}
N(\mu)\ge \, \wt c N_u(C^{-1}_3\mu),\ {\rm for \ all}\ 0<\mu<\, C_3\, b. 
\end{equation}

Next, we consider $\mu\ge b$ and $R=\cl{\frac{1}{2C''})^{1/4}\mu^{-1/2}}\le 3$. In this case, we construct $\cF''$ and $S''$ using cubes of side length $\ell(Q)=\wt R=9$ and tweak the dimensional constants in the counting in the end by Remark \ref{rem:pat}. The upper bound for $\ipc{u^{Q }}{H u^{Q }}$ remains the same as in \eqref{eq:doub-pf1} 
\begin{equation}
\ipc{u^{Q }}{H u^{Q }}\le \, a\, ( \sup_{Q }u_n^2+\sup_{Q }u_n) \le \, a(\mu^{-2}+\mu^{-1}),
\end{equation}
for some dimensional constant $a>0$.  For the lower bound on $\ipc{u^Q}{u^Q}$, we apply the Harnack inequality \eqref{eq:Harnack} to $u_n$ on $Q$. We obtain 
\begin{equation}
\ipc{u^{Q }}{ u^{Q }}\ge \, \sum_{Q/3}u_n^2\ge \, \inf_{Q}u_n^2 \ge \,  c_V\,\sup_Q\, u_n^2\ge \, c_V\, \mu^{-2},
\end{equation}
for some constant $c_V$ depending on $V_{\max}$. Therefore, 
\begin{equation}
\frac{\ipc{u^{Q }}{H u^{Q }}}{\ipc{u^{Q }}{ u^{Q }}} \le \frac{a}{c_V}(1+\mu) \le \frac{a}{c_V}\, (b^{-1}+1)\mu :=c_V'\, \mu,
\end{equation}
for some constant $c'_V$ depending on $d$, $V_{\max}$ and $C''$. Then 
\begin{equation*}
N(c_V'\, \mu)\ge \, \wt c N_u(\mu),\ {\rm for \ all}\ \mu\ge b. 
\end{equation*}
Equivalently, one has  
\begin{equation}\label{eq:mu>b}
N(\mu)\ge  \, \wt c N_u(c'^{-1}_V\mu),\ {\rm for \ all}\ \mu\ge \,  c_V'\, b. 
\end{equation}
Clearly, in \eqref{eq:mu<C3b}, we can make $C_3\ge c_V'$ so that the estimates \eqref{eq:mu<C3b} and \eqref{eq:mu>b} cover all $\mu>0$.   This completes the proof of Theorem \ref{thm:NNu-intro-scaling}. 

\ep


\subsection{Lower bound for the periodic potential}\label{sec:peri}
In this part, we prove  Corollary \ref{cor:periodic} for a $\Z^d$ periodic potential $V=\{v_n\}$.  It is enough to show that the landscape function $u$ associated with the periodic potential satisfies the scaling condition \eqref{eq:scaling}. Informally, this is quite obvious. Indeed, at small scales (below $p_{\max}=\max_d p_d$) we simply use the Harnack inequality. The emerging constant is roughly of the order of $V_{\max}^{\,C p_{\max}}$ then. At large scales we simply use periodicity to reduce to small scales. Here are the details.

Suppose $V=\{v_n\}$ is $\Z^d$ periodic with period $\vec p=(p_1,\cdots, p_d)$. Let $\Gamma=\llbracket 1,p_1\rrbracket\times \cdots \times \llbracket 1,p_d\rrbracket \subset  \Z^d$ be the fundamental cell of $V$. Notice that the condition $p_i\mid K,i=1\cdots,d$,  guarantees  that $\Lambda=(\Z/K\Z)^d$ contains finitely many copies of $\Gamma$.    Let $H_{\Gamma}$ be the restriction of $H$ on $\Gamma$ with the periodic boundary conditions, and let $ u^\Gamma=\{ u^\Gamma_n\}_{n\in \Gamma}$ be the landscape function for $H_\Gamma$, i.e., $(H_\Gamma   u^\Gamma)_n=1,n\in \Gamma$. By the uniqueness of the landscape function (Theorem \ref{thm:landscape}),  $u=\{u_n\}_{n\in \Lambda}$  will be the periodic extension of  $ u^\Gamma=\big\{ u^\Gamma_n\big\}_{n\in \Gamma}$ to the entire domain $\Lambda$.  

For $s\in \N$ and $a\in \Lambda$, let $Q(s)=a+\llbracket 1,s\rrbracket^d$ be cube in $\Lambda$ of side length $s$. Consider $\Gamma+\, \vec p\, \Z^d$, a collection of disjoint translations (copies) of the fundamental cell  by $\vec p\, \Z^d$. 
Suppose $s>p^{\max}:=\max\{p_1,\cdots,p_d\}$. It is easy to verify that the maximal number of copies of $\Gamma$ in the collection $\Gamma+\, \vec p\, \Z^d$ which lie inside $Q(s)$ (we will call their union $T_1$) and the minimal number of copies of $\Gamma$ in the collection $\Gamma+\, \vec p\, \Z^d$ which cover $3Q(s)$ (we will call their union $T_2$) differs by a dimensional multiplicative constant. 
Therefore, denoting the number of copies of $\Gamma$ in $T_1$ by $t$ and using the periodicity of $u=\{u_n\}$ with respect to all translations of $\Gamma$, one has, 
\[
\sum_{Q(s)}\, u_n^2\, \ge \, \sum_{T_1}\, u_n^2\, =\, t\, \sum_{\Gamma}\, \big( u^\Gamma_n\big)^2, \quad \mbox{and} \quad \sum_{3Q(s)}\, u_n^2\, \le \,  \sum_{T_2}\, u_n^2\, \le \,  Ct\, \sum_{\Gamma}\, \big( u^\Gamma_n\big)^2,\]
which shows the scaling condition \eqref{eq:scaling} is true for relatively large cubes.

If $s<p^{\max}$, we simply observe that the landscape function $u$ satisfies 
$
    -(\Delta u)_n+v_nu_n\ge 0,$ and 
$
    -(\Delta u)_n\le 1\, \  \ n\in \Lambda.
$
A combination of the Moser-Harnack inequality \eqref{eq:MH-inhomo} and the Harnack inequality \eqref{eq:Harnack} implies that for some constant $C$ depending on $d$ and $V_{\max}$, one has 
\begin{equation*}
    \sup_{3Q(s)}\, u_n^2 \, \le \,  C^{s} \inf_{Q(s)}u_n^2 \,  \le \, C^{s}   \inf_{Q(s)/3}u_n^2 \, \le\,    C^{s} \sup_{Q(s)/3}\, u_n^2 \le \, C^{s} \Bigl(c_{H}^{-1} (s/3)^{-d}\, \sum_{Q(s)}u_n^2\, +\, c_{H}^{-1}(s/3)^4\Bigr).
\end{equation*}
Since $C^s\le C^{p^{\max}}$, one has  
$
\sum_{3Q(s)}u_n^2\le (3s)^d \sup_{3Q(s)} u_n^2\le \wt C\left(\sum_{Q(s)}u_n^2+s^{d+4}\right), 
$
where the constant $\wt C$ only depends on the dimension, $V_{\max}$ and $p^{\max}$. 

Therefore, for all cubes $Q(s)$, $\{u_n^2\}_{n\in \Lambda}$ satisfies the scaling condition \eqref{eq:scaling}. The estimates for $N$ in the periodic case then follow directly from Theorem \ref{thm:NNu-intro-scaling}. \ep


\section{Landscape law for the Anderson model}\label{sec:Anderson} There must be changes propagating from changes in the previous section, at treatment of small scales. I do not bother about it for now.  In the present section we will concentrate on  the Anderson model. To this end, we consider $V=\{v_n\}_{n\in \Lambda}$ with the values given by independent, identically distributed (i.i.d.) random variables, with common probability measure $P_0$ on $\R$, subject to the conditions stated in the beginning of Section~\ref{subs-disorder}. In particular, denoting by $F(\delta)=P_0(v_n\le \delta)$ the common cumulative distribution function of $v_n$, we have  $F(\delta)=0$ if $\delta<0$, $F(\delta)=1$ if $\delta\ge V_{\max}$, and there is a $\delta_\ast>0$, such that 
\begin{equation}\label{eq:ass-F}
    0<F(\delta)\le F(\delta_\ast):=F_\ast<1  \ \ {\rm for\ all}\ \ 0<\delta\le \delta_\ast.
\end{equation}
We note that $\delta_\ast$ can be picked to be less than $\min(1,V_{\max}/2)$ since $\inf {\rm supp}P_0=0$. Hence $F_\ast V_{\max}/2\le\E(v_n)\le F_\ast V_{\max}/2+V_{\max}$. Some constants used in the proof of this section receive their dependence on $\E(v_n)$ through $\delta_\ast$ and $F_\ast$. 
\subsection{Estimates for $N_u$ in the Anderson model}
We will study the following tail estimates for $N_u$ first. 
\begin{theorem}\label{thm:Nu-iid}
Let $V=\{v_n\}_{n\in \Lambda}$ be an Anderson-type potential as above. Then there are dimensional constants $c_3,c_4,\gamma_1,K_\ast>0$ such that 
\begin{equation}
\Ev{N_u(\mu)}\, \ge \, c_4\, \mu^{d/2}\, F(c_3\mu)^{\, \gamma_1\,  \mu^{-d/2}}  \ \ {\rm for \ all}\ K_\ast/K^2\le \mu\le 1. \label{eq:Nu>F}
\end{equation}
Furthermore, there are constants $C_3,C_4,\gamma_2$, and $\mu_\ast>0$ depending on $d,\delta_\ast,F_\ast$ only, such that 
\begin{equation}
  \Ev{N_u(\mu)}\,  \le\,  C_4 \, \mu^{d/2}\, F(C_3 \mu)^{\, \gamma_2\, \mu^{-d/2}}  \ \ {\rm for \ all}\ \mu<\mu_\ast. \label{eq:Nu<F}
\end{equation}
\end{theorem}
\begin{remark}
All the constants are independent of $V_{\max}$. 
\end{remark}
After we establish these tail estimates for $N_u$, we will combine them with the deterministic result Theorem \ref{thm:NNu-intro},\ref{thm:NNu-intro-2} to prove Theorem \ref{thm:NNu-iid-intro}, and \eqref{eq:N-tail-intro-new}. 

\bp 
{\noindent{\bf The proof of  \eqref{eq:Nu>F}.}} For $0<\mu\le 1$, let   $
    r=\cl{ \mu^{-\frac{1}{2}}  }. 
$
Let $\cP(r;\Lambda)$ be  partition of size $r$ as usual.
It is enough to assume that $K= K_0 r$ for some $K_0\in \N$, otherwise the counting can always be bounded from below by ignoring the irregular boxes in the last rows/columns of $\cP(r;\Lambda)$.  For all cubes $Q\in \cP=\cP(r;\Lambda)$, let $\zeta_Q=1$ if  $\min_{n\in Q}\frac{1}{u_n}\,\le \mu$ and $\zeta_Q=0$ otherwise. Direct computation shows that
\begin{align}
  \Ev{N_u(\mu)}=&  \frac{1}{K^d}\,\E\left(\, \card{Q\in \cP(r)\, :\ \ \min_{n\in Q}\frac{1}{u_n}\,\le \mu}\, \right)\nonumber \\
  =&\frac{1}{K^d}\,\Ev{\sum_{Q\in \cP(r)}\zeta_Q } 
  =\frac{1}{K^d}\,\sum_{Q\in \cP(r)}\Ev{\zeta_Q }
=  \frac{1}{K^d_0r^d}\sum_{Q\in \cP(r)}\Prr{ \min_{n\in Q}\frac{1}{u_n}\,\le \mu } .\label{eq:min-max}
\end{align}

For each $Q\in \cP$, consider translations of $Q$ by the vectors $k  r e_i$ for all $1\le i\le d$ directions and  $|k|\le m$. Here $m$ is some large integer that will be specified later. Let $\cB_m=\cB_m  (Q)$ be the union of these translated cubes,
\begin{equation*}
\cB_m =\bigcup_{|k|\le m,1\le i \le d}\left(Q+kre_i\right) .
\end{equation*}

Similarly to \eqref{eq:cutoff}, one can construct a discrete cut-off function $\chi=\{\chi_n\}\in \cH\cong \R^{K^d}$, supported on $\cB_{2m}$, and satisfying 
\begin{align*}
    &\chi_n=1, \ \ n\in \cB_{m},\ \ \ \ {\rm and}  \ \ \chi_n=0, \ \ n \notin \cB_{2m},   \\
    & 0\le \chi_n\le 1,\,  n\in \cB_{2m}\backslash \cB_{m}, \\
    &|\nabla_i\chi_n|=|\chi_{n+e_i}-\chi_n|=0, \,  n\notin \cB_{2m},  \\
    &|\nabla_i\chi_n|=|\chi_{n+e_i}-\chi_n|<\frac{1}{m\, \ell(Q)}<\frac{1}{m}\, \mu^{\frac{1}{2}}, \ \  n\in \cB_{2m} , \ 1\le i \le d .
\end{align*}
By the landscape uncertainty principle \eqref{eq:uncert}, one has
\begin{equation*}
  \min_{\cB_{m}}\frac{1}{u_n}
   \le \, \frac{1}{|\cB_{m}|}\sum_{\cB_{2m}}\|\nabla\chi_n\|^2+ \frac{1}{|\cB_{m}|}\sum_{\cB_{2m}}v_n  \chi_n^2 
     \le \, \frac{2^dd}{m^2}\, \mu\, +\, 2^d\, \max_{\cB_{2m}}\, v_n
    \le \, \frac{1}{2}\, \mu\, +\, 2^d\, \max_{\cB_{2m}}v_n,
\end{equation*}
provided $m^2\ge 2^{d+1}d$. 
Therefore, for all $Q\in\cP\bigl(\cl{\mu^{-1/2}};\Lambda\bigr)$,
\begin{equation*}
    \P\left\{ \min_{n\in \cB_{m}  }\frac{1}{u_n}\,\le  \mu \right\} 
   \ge  \P\left\{\max_{n\in \cB_{2m}  }v_n\le \frac{1}{2^{d+1}}\, \mu \right\}
    =  \big(F\left(c\mu\right)\big)^{|\cB_{2m}  |}
    \ge  \big(F\left(c\mu\right)\big)^{\, C\, \mu^{-d/2}}, 
\end{equation*}
where $c={2^{-d-1}}$ and $C=(4m+1)^d$. On the other hand, notice that all the translations $Q+kre_i$ still belong to $\cP(r)$. Then  we can rewrite $\cB_m=\cB_{m}  (Q)$ as
$
\cB_{m}  (Q)=\bigcup_{Q'\in \cP(r)\cap \cB_{m}  (Q) }Q'\, , 
$
which implies that for all $Q\in \cP(r)$
\begin{equation*}
  \sum_{Q'\in \cP(r)\cap \cB_{m}  (Q)} \P\Bigl\{ \min_{n\in Q'}\, \frac{1}{u_n}\,\le  \mu \Bigr\}
  \ge \P\Bigl\{ \min_{n\in \cB_{m}  (Q)}\, \frac{1}{u_n}\,\le  \mu \Bigr\} 
  \ge \, \bigl(F\left(c\mu\right)\bigr)^{\, C\, \mu^{-d/2}}.
\end{equation*}
Summing the  left-hand side of the above inequality over  all $Q\in \cP(r)$, one has
\begin{equation*}
\sum_{Q\in \cP(r)} \sum_{Q'\in \cP(r)\cap \cB_{m}  (Q)} \P\left\{ \min_{n\in Q'}\frac{1}{u_n}\,\le  \mu \right\}
=(2m+1)^d\sum_{Q\in \cP(r)} \P\left\{ \min_{n\in Q}\frac{1}{u_n}\,\le  \mu \right\} .
\end{equation*}

Combining this together with \eqref{eq:min-max}, one has 
\begin{align*}
\Ev{N_u(\mu)}
\ge&\, \frac{(2m+1)^{-d}}{K^d_0\, r^d}\,  \sum_{Q\in \cP(r)}\, \left( \sum_{Q'\in \cP(r)\cap \cB_{m}  (Q)} \P\left\{ \min_{n\in Q'}\frac{1}{u_n}\,\le  \mu \right\}\right)\\
\ge& \,
\frac{(2m+1)^{-d}}{K^d_0\, r^d}\, \sum_{Q\in \cP(r)} \, \big(F\left(c \mu\right)\big)^{\, C\, \mu^{-d/2}}  
\ge \, (2m+1)^{-d}\, 2^{-d}\, \mu^{\frac{d}{2}} \,  \big(F\left(c \mu\right)\big)^{\,C\,\mu^{-d/2}}. 
\end{align*}

Note that we also need to impose the condition on the size of domain so that $\cB_{2m}\subset \Lambda$, i.e., $K\ge (2m+1)r =C'\mu^{-1/2}$. 


\vskip 0.08 in \noindent {\bf{The proof of \eqref{eq:Nu<F}.}} Using \eqref{eq:min-max}, it is enough to bound $\P\bigl\{ \min_{n\in Q}\frac{1}{u_n}\,\le  \mu \bigr\} $ from above since 
\begin{equation}\label{eq:min-max-2}
    \Ev{N_u(\mu)}\le \, 
 \frac{1}{\cl{\mu^{-1/2}}^d}\, \max_{ Q\in \cP(\cl{\mu^{-1/2}})}\, \P\left\{ \min_{n\in Q}\, \frac{1}{u_n}\,\le \mu \  \right\}.
\end{equation}

This will be the most delicate part. We need several technical lemmas concerning the growth of the landscape function. Some of these estimates may have independent interest in the landscape theory. 

For any $r\ge 3$, let $B\subset \Lambda$ be cube of side length $\ell(B)=r$ and let $\check B =B/3$ be the middle third cube  as defined in \eqref{eq:Q/3}. We are going to show that there is a suitable $M$ (large, and only depending on the  expectation of the random variable), such that for any $\mu$ (small enough, depending on $\E(v_n)$), and any cube $B$ of side length $r=\cl{(4M\mu)^{-1/2}}$
\begin{equation}\label{eq:key}
    \P\Bigl\{ \min_{n\in \check B }\frac{1}{u_n}\,\le  \mu \Bigr\}\le \,  A_1  M^{d/2} F(A_2 M\mu)^{\, (M\mu)^{-d/2}/2}, 
\end{equation}
for some suitable constants $A_1,A_2$ (depending only on the dimension and $\E(v_n)$, and independent of $\mu$).

We start from the following deterministic statement. The lemma states that the landscape function $u$ is forced to grow at a certain rate if $V$ is reasonably non-degenerate. 
\begin{lemma}\label{lem:climb} 
Let $u=\{u_n\}$ be the landscape function  given by Theorem \ref{thm:landscape}. Let $B\subset \Lambda$ be a cube of side length $r:=\ell(B)\ge 3$, and let $\check B=B/3$ be the middle third as usual.   For any $0<\lambda<1$, there is $\eps_0(d,\lambda)>0$ such that for all $0<\eps<\eps_0$, there are constants   $C_P=C_P(\eps,\lambda,d)>0$,  $M=M(\eps,\lambda,d)>0$ and $r_\ast=r_\ast(\eps,\lambda,d)>0$, such that the following statement holds. If $B$ satisfies conditions
\begin{description}
\item[(i)] 
\begin{equation}\label{eq:card-J}
    \card{j\in B: v_j\ge C_Pr^{-2}}\ge \lambda |B|, 
\end{equation}
\item[(ii)] there is a $\xi \in \check B $ such that 
\begin{equation}\label{eq:Mr2}
    u_{\xi} \ge Mr^2,
\end{equation}
\end{description}
then for all $r\ge r_\ast$, there is a $\xi'\in \Lambda$ such that $|\xi'-\xi|_{\infty}\le \fl{\sqrt{1+\eps}r}$ and 
\begin{equation}\label{eq:MR2}
    u_{\xi'} \ge (1+\eps)\,  u_{\xi}\ge  M\fl{\sqrt{1+\eps}r}^2.
\end{equation}
\end{lemma}

\begin{remark}\label{rem:climb}
This lemma holds  for any $\lambda$, and works for any cube $B$ of side length $r$. The choice of $C_P$, $\eps$ and $M$ only depend on $\lambda$,  and is independent of the choice of $B$, neither on its size nor the position.  
\end{remark}
\begin{remark}\label{rem:climb2}
Note that this Lemma is a completely deterministic result. It has nothing to do with the randomness (structure of $v_n$). It can be applied to any $V$ and $u$ as long as $(Hu)_n=1$ locally on the lattice (containing the cube $B$ and its neighborhood). In our proof, Lemma \ref{lem:climb} will lead to some important probability estimates.  The small parameter $\lambda$ will be picked at the very end when we are about to prove \eqref{eq:key}. 
\end{remark}

We need some technical preparations for the average of $u_n$. We will frequently write $u(n)=u_n$ to make the notations of the sub-index  easier to read. For $\xi=(\xi_1,\cdots,\xi_d)\in \Lambda$, and $r \in \Z_{\ge 0}$, we denote by $Q(r;\xi)$ the box centered at $\xi$ of side length $2r+1$ :
\begin{equation*}
  Q(r;\xi)=\left\{m=(m_1,\cdots,m_d)\in\Z^d:\, |m-\xi|_\infty\le r\right\}.
\end{equation*}
We will omit the center $\xi$ (fixed) and write $Q(r)=Q(r;\xi)$ when  it is clear. We denote by $\partial Q(r)\subset Q(r)$ the inner boundary of $Q(r)$ as defined in \eqref{eq:bdry-inner},
 and by $\partial^\circ Q(r)\subset \partial Q(r)$ the boundary removing the ``corners'' as defined in \eqref{eq:bdry-no-corner}. 
For $r=0$, we have the ``degenerate cube'' $Q(0;\xi)\,=\,\partial Q(0;\xi)\, =\, \{\xi\}$.
Notice that $Q(r)=\cup_{\rho=0}^{r}\partial Q(\rho)$. 
Let $a_r$ be the average of $u_n$ on $\partial Q(r) $ with respect to the (discrete) Poisson kernel (see the definition and properties of $P_r$ in \eqref{eq:Pdef} in Appendix \ref{sec:poisson}):
\begin{equation}\label{eq:ar}
    a_r=\sum_{n\in \partial Q(r;\xi)}P_r(\xi,n)\,u_n,\quad  r\ge 1, \quad  a_0=u_\xi,\,
\end{equation}
in other words, a harmonic function with data $u$ on the boundary.
Let $A_r$ be the corresponding weighted average of $u_n$ on $Q(r)$:
\begin{equation}\label{eq:Ar}
    A_r\,=\,\frac{1}{|Q(r)|}\,\sum_{\rho=0 }^{r}\,|\partial Q(\rho)|\,a_\rho\,=\, \frac{1}{|Q(r) |}\,\sum_{n\in Q(r) }\, p_n\, u_n,
\end{equation}
where for any $n\in  Q(r)$, and  $\rho=|n-\xi|_\infty$
\begin{equation}\label{eq:pn}
    p_n=|\partial Q(\rho)|\, P_\rho(\xi,n). 
\end{equation}
By the properties of the discrete Poisson kernel, one has
\begin{equation}\label{eq:sumpnpn}
    \sum_{n\in \partial Q(r;\xi)}P_r(\xi,n)=1 \Longrightarrow   \sum_{n\in  Q(r;\xi)}p_n=|Q(r)|.
\end{equation}

The first two estimates are lower bounds on $a_r$ and $A_r$. 
\begin{lemma}\label{lem:arAr}
There is dimensional constant $C$, such that for any $\xi \in \Lambda$ and $r\ge 1$ 
\begin{align}
   &a_r\ge u_\xi-r^2,\label{eq:ar-lower}\\
&A_r\ge u_\xi-C r^2. \label{eq:Ar-lower}
\end{align}

\end{lemma}
\begin{proof}
Let $\tilde u$ be the landscape function for the free Laplacian on $Q(r)$ with zero Dirichlet boundary condition:
\begin{align*}
    \begin{cases}
    -(\Delta u')_n=1,\, n\in Q(r-1), \\
    (u')_n=0, \,n\in \partial Q(r).
    \end{cases}
\end{align*}
Let $u''_n=\frac{1}{2}r^2-\frac{1}{2d}\sum_{i=1}^d(n_i-\xi_i)^2$. Direct computation shows that 
\begin{align*}
    \begin{cases}
    -(\Delta u'')_n=1, n\in Q(r-1), \\
    (u'')_n\ge 0, n\in \partial Q(r).
    \end{cases}
\end{align*}
By the maximum principle (Lemma \ref{lem:maxP}), one has for all $n\in Q(r)$, ${u'}_n\le u''_n \le r^2$. Let $w_n$ be the harmonic function on $Q(r)$ with the boundary data equal to $u_n$, i.e.,
\begin{align*}
    \begin{cases}
    -(\Delta w)_n=0,\, n\in Q(r-1), \\
    (w)_n=u_n,\, n\in \partial Q(r).
    \end{cases}
\end{align*}
Then by the Poisson integral formula \eqref{eq:IBP}, $w_\xi=\sum_{n\in \partial Q(r;\xi)}P_r(\xi,n)u_n=a_r$. On the other hand,
\begin{align*}
    \begin{cases}
    -(\Delta (w+u'-u))_n=v_nu_n\ge 0,\, n\in Q(r-1), \\
    (w+u'-u)_n=0,\, n\in \partial Q(r).
    \end{cases}
\end{align*}
Therefore, $(w+u'-u)_n\ge 0$ for all $n\in Q(r)$. In particular, $
    u_\xi\le w_\xi+u'_\xi \le a_r +r^2,$ proving \eqref{eq:ar-lower}, and a similar statement is true for $1\le \rho \le r-1$, so that
\begin{equation*}
  |\partial Q(\rho)|  u_\xi \le |\partial Q(\rho)|\sum_{n\in \partial Q(\rho) }P_\rho(\xi,n)u_n + |\partial Q(\rho)| \,\rho^2 \le \sum_{n\in \partial Q(\rho) }p_nu_n + C \rho^{d+1}.
\end{equation*}
Summing over $1\le \rho \le r-1$, one has 
\begin{equation*}
| Q(r)| u_\xi= \sum_{\rho=1}^{r-1} |\partial Q(\rho)| u_\xi 
  \le \sum_{\rho=1}^{r-1}\sum_{n\in \partial Q(\rho) }p_nu_n + C \sum_{\rho=1}^{r-1}\rho^{d+1}\le  \sum_ {n\in Q(r;\xi)}p_nu_n + C  r^{d+2},
\end{equation*}
which implies $
  u_\xi 
  \le \,  A_r + C\, r^{2},$ as desired.
\end{proof}

\begin{lemma}\label{lem:Poisson-1} For any $d\geq 1$ and $0<\eta<1/4$ there is a constant $C_1=C_1(d)>0$ (independent of $\eta$) and a constant $C_2=C_2(\eta,d)>0$ such that for any cube  $Q(r)=Q(r;\xi)$ of side length $r\ge \, C_1/\eta$, there is  a subset $Q^\eta(r)\subset Q(r) $ such that  $|Q(r)\backslash Q^\eta(r)|\le \, C_1\, \eta \, r^d$ and 
\begin{equation}\label{eq:pnC2}
   p_n\ge \, C_2 \ {\rm for\ } n\in  Q^\eta(r).
\end{equation}
\end{lemma}
\begin{remark}
The lemma is true in all dimensions. However, for $d=1$, we actually do not need to remove any portion of the cube (an interval in $\Z$) to obtain \eqref{eq:pnC2} since the 1-d Poisson kernel $P_r$ is rather trivial (constantly $1/2$), and so is $p_n$. 
\end{remark}
\begin{proof}
Write $Q(r)=\cup_{\rho=0}^{r}\partial Q(\rho)$. The estimate follows from the lower bound of $P_\rho(\xi,n)$ on each $\partial Q(\rho)$ as long as $n$ is away from the edges (and the corners). Given $0<\eta<1$,  according to Lemma \ref{lem:poisson}, there exist  $c(\rho,\eta)$ and $\rho_0(\eta,d)$ such that for all $\rho\ge\rho_0$, $P_\rho(\xi,n)\ge c\rho^{1-d}$ on $\partial Q(\rho)$ except for $C\eta \rho^{d-1}$ many $n\in \partial Q(\rho)$, where $C$ only depends on the dimension. Therefore, $ p_n=|\partial Q(\rho)|P_\rho(\xi,n)\ge \wt c>0$ on $\bigcup_{\rho=\rho_0}^r\partial Q(\rho) $ except for $\sum_{\rho=\rho_0}^rC\eta\rho^{d-1}\le C\eta r^d$ many $n$. For $0\le \rho <\rho_0$, we then have $P_\rho(\xi,n)>c(\rho,d)$ except for those $n$ on the edges and corners, whose total cardinality is at most $C\rho^{d-2}$. This again implies $p_n\ge \min_{0\le \rho <\rho_0}c(\rho,d)$ for all $n \in \bigcup_{0\le \rho <\rho_0}\partial Q(\rho)$ except for $\sum_{\rho=0}^{\rho_0}\rho^{d-2} \lesssim r^{d-1} \lesssim \eta r^d$ many $n$. Therefore, $p_n\ge C_2$ for some constant $C_2$ only depending on $\eta$ and $d$, and the cardinality of the exceptional set of  $n\in Q(r)$ violating this, is at most $C_1 \eta r^d$. 
\end{proof}

With these two technical lemmas, we are ready to the 
\begin{proof}[Proof of Lemma \ref{lem:climb}]
Let $B$ and $\xi \in \check B$ be given as in Lemma \ref{lem:climb}, where $\ell(B)=r$. Clearly, $B\subset Q(r;\xi)$. Denote by $J$ the set in condition \eqref{eq:card-J}, that is, $
    J=\{j\in B: v_j\ge C_Pr^{-2}\}, $ 
where $C_P$ will be be picked later. Fix $0<\lambda<1$, we assume that $|J|\ge  \lambda |B|=\lambda r^d$. 
Let $A_r,p_n$ be defined in \eqref{eq:Ar}. Let $
    J_0=\{j\in J: u_j<\, \frac{1}{2}\, A_r\, \}$ 
and let
\begin{equation}\label{eq:SJ0}
    S(J_0 )=\sum_{n\in J_0}p_n,\ \ S(J^C_0)=\sum_{n\in Q(r)\backslash J_0}p_n=|Q(r)|-S(J_0)\, ,
\end{equation}
where the last equality follows from \eqref{eq:sumpnpn}.

Now we are ready to look for $\xi'$ satisfying \eqref{eq:MR2} in the following two cases:

\vskip 0.08 in 
\noindent {\bf{Case I:} }  $
    |J_0|\ge \frac{1}{2}\lambda |B|=\frac{1}{2}\lambda r^d.$ 
    
Let $C_1$ be the constant from Lemma \ref{lem:Poisson-1}, then pick $\eta=\min(\frac{\lambda}{4C_1},1/4)$ and let $C_2=C_2(d,\eta)$ be also the constant from Lemma \ref{lem:Poisson-1}. Then 
\begin{equation*}
    |J_0|=\sum_{n\in J_0\cap Q^\eta(r)}1+\sum_{n\in J_0\backslash Q^\eta(r)}1 
  \le \, C_2^{-1}\sum_{n\in J_0\cap Q^\eta(r)}p_n+C_1\eta r^d \le \,   C_2^{-1}\, \sum_{n\in J_0}p_n+\frac{1}{4}\lambda r^d.
\end{equation*}
Therefore, 
\begin{equation*} 
    S(J_0)=\sum_{n\in J_0}p_n\ge  \, C_2\frac{1}{4}\lambda\, r^d:=c_3(\lambda)\, r^d.
\end{equation*}

 Direct computation shows that 
\begin{equation*}
    |Q(r)|\,A_r=\sum_{n\in Q(r)}p_nu_n= \sum_{J_0}p_nu_n+\sum_{Q(r)\backslash J_0}p_nu_n 
    \le  \frac{1}{2} \,A_r\,  S(J_0)+\sum_{Q(r)\backslash Q_0}p_nu_n.
\end{equation*}
By the definition of $S(J_0^C)$ and \eqref{eq:SJ0},  this implies that 
\begin{equation*}
  \frac{1}{S(J_0^C)}  \sum_{Q(r)\backslash J_0}p_nu_n
  \ge \frac{|Q(r)|-\frac{1}{2} S(J_0)}{|Q(r)|-S(J_0)}A_r\ge  \Bigl(1+\frac{S(J_0)}{2|Q(r)|}\Bigr)A_r 
  \ge \bigl(1+c_4(d,\lambda)\bigr)\, A_r .
\end{equation*}
Therefore, there is one point $\xi' \in Q(r)\backslash J_0$ such that 
$u_{\xi'} \ge \left(1+c_4\right)\, A_r.$ By \eqref{eq:Ar-lower} and \eqref{eq:Mr2}, 
\begin{equation*}
    u_{\xi'} \ge  \left(1+c_4\right)\, (u_{\xi}-C_dr^2) 
    \ge   \left(1+c_4\right)\, \left(1-\frac{C_d}{M}\right)u_{\xi}  
    \ge   \left(1+\frac{c_4}{2}\right)\, u_{\xi},
\end{equation*}
provided $
    M>\frac{2}{c_4}(C_d+1):=M_0(d,\lambda) .
$
Therefore, \eqref{eq:MR2} holds  for $\eps<c_4/2:=\eps_0(d,\lambda)$. 

\vskip 0.08 in 
\noindent {\bf{Case II:} } $
    |J_0| \le \frac{1}{2}\lambda r^d
.$ 

Take $R=\fl{ \sqrt {1+\eps}\, r } $ for some small $\eps$, which gives $R^2-r^2\le \eps r^2$. Let $a_r,a_R$ and $G_r,G_R$ be the surface average and the Green's function on $Q(r)=Q(r;\xi),Q(R)=Q(R;\xi)$ respectively, as defined in \eqref{eq:ar},\eqref{eq:Gdef}. 
Applying the discrete Green's identity (integration by parts formula) \eqref{eq:IBP} on $Q(r)$ and $Q(R)$, one has 
\begin{equation*}
    u(\xi)= a_R-\sum_{m\in  Q(R-1)}G_{R}(\xi,m)\Delta u(m)=a_r-\sum_{m\in  Q(r-1)}G_{r}(\xi,m)\Delta u(m).
\end{equation*}
Then 
\begin{align*}
    a_R-a_r
=&\sum_{m\in  Q(R-1)\backslash Q(r-1)}G_{R}(\xi,m)\Delta u(m)+\sum_{m\in  Q(r-1)}\, \big(G_{R}(\xi,m)-G_{r}(\xi,m)\big)\, \Delta u(m)  \\
\ge& -\sum_{m\in  Q(R-1)\backslash Q(r-1)}G_{R}(\xi,m)-\sum_{m\in  Q(r-1)}\, \big(G_{R}(\xi,m)-G_{r}(\xi,m) \big)\\ 
&  \qquad \qquad \qquad \qquad \qquad  +\sum_{m\in  Q(r-1)}\, \big(G_{R}(\xi,m)-G_{r}(\xi,m)\big)v_mu_m.
\end{align*}
Notice that  $\Delta G_R(\xi,\cdot)=0$  in $Q(R-1;\xi)\backslash \{\xi\}$. By the maximum principle, Lemma \ref{lem:maxP-harmonic}, one has 
\begin{equation*}
    \max_{m\in Q(R-1)\backslash Q(r-1)}G_{R}(\xi,m) \le \max_{m\in \partial Q(r)} G_{R}(\xi,m),
\end{equation*}
where we used $G_R(\xi,m)=0$ for $m\in \partial Q(R)$ and the discussion for an annular region after Lemma \ref{lem:maxP-harmonic}. 

On the other hand, $\Delta \big(G_R(\xi,\cdot)-G_r(\xi,\cdot)\big)=0$ in $Q(r-1)$ and $G_R(\xi,n)-G_r(\xi,n)=G_R(\xi,n)$ for $n\in \partial Q(r)$. Then using once again the maximum principle, Lemma \ref{lem:maxP-harmonic}, for any $m\in Q(r-1)$,
\begin{equation*}
 \min_{m'\in \partial Q(r)} G_{R}(\xi,m')   \le\, G_{R}(\xi,m)-G_{r}(\xi,m) \le \max_{m'\in \partial Q(r)} G_{R}(\xi,m').
\end{equation*}
By the choice of $r$ and $R$, $\partial Q(r)$ is away from both $\partial Q(R)$ and the pole $\xi$. Then Lemma \ref{lem:GRr} implies that for $r$ large enough (depending only on $\eps$) and all $m\in \partial Q(r)$, 
\begin{equation*}
C_1r^{2-d}\le     G_{R}(\xi,m')\le C_2r^{2-d}, 
\end{equation*}
where $C_1,C_2$ only depend on $d$ and $\eps$. 
Therefore,
\begin{align*}
   a_R-a_r\ge & -C_3(R^d-r^d)r^{2-d}-C_4r^d r^{2-d}+r^{2-d}\sum_{m\in  Q(r-1)}\, v_mu_m  \\
   \ge & -C_5(d,\eps)r^2+r^{2-d}\sum_{m\in  J\backslash J_0}\, v_mu_m \\
   \ge &  -C_5(d,\eps)r^2+C_1r^{2-d}\,C_Pr^{-2}\frac{1}{2}A_r\big(|J|-|J_0|\big)\ge   -C_5(d,\eps)r^2+2\eps A_r,
\end{align*}
provided  that $
    C_P\gtrsim \eps /(C_1\lambda).$

Finally, by the lower bounds for $a_r,A_r$ in Lemma \ref{lem:arAr}  and the condition \eqref{eq:Mr2} on $u_\xi$, we have 
\begin{align*}
  a_R \ge  &\, u_\xi-r^2 -C_5(d,\eps)r^2+2\eps(u_\xi-Cr^2) 
   \ge   (1+2\eps)u_\xi-(1+C_5+C_6)r^2\\
   \ge &\, (1+2\eps)u_\xi-(1+C_5+C_6)\,\frac{1}{M}\,u_\xi\ge  (1+\eps)u_\xi,
\end{align*}
provided that $
    M\ge ({1+C_5+C_6})/{\eps}:=M_0(d,\eps,\lambda).$ 
Recall that $a_R=\sum_{\partial Q(R)}P_R(\xi,n)u_n$ and $\sum_{\partial Q(R)}P_R(\xi,n)=1$, therefore, there is $\xi' \in \partial Q(R;\xi)$ such that 
$
    u_{\xi'}\ge (1+\eps)u_{\xi},
$
which completes the proof of Lemma \ref{lem:climb} in the second case.
\end{proof}
Notice that  Lemma \ref{lem:climb} is deterministic and requires no randomness of $v_n$. A direct consequence is the following  estimate on the probability that $u_\xi$ grows.  

\begin{lemma}\label{lem:Pr} 
Let $V(\omega)=\{v_n\}_{n\in \Lambda }$ be the Anderson-type   potential as in Theorem \ref{thm:Nu-iid}. 
Fix $0<\lambda<1$, and retain $\eps<\eps_0,C_P,M, r_\ast$ from Lemma  \ref{lem:climb}. For any cube $Q\subseteq \Lambda$ of side length $\ell(Q)$, define the event
\begin{equation}\label{eq:OmegaI}
    \Omega(Q):=\Big\{\omega:\,   \cards{j\in Q : v_j\ge C_P\,\ell(Q)^{-2}}\le  \lambda |Q| \Big\}.
\end{equation}

Assume that $r_\ast\ge 15/\eps$, otherwise, just reset $r_\ast$ to be $\max\{15/\eps,r_\ast\}$.
For any $r_0\ge r_{\ast}$, set 
\begin{equation}\label{eq:rk}
    r_{k+1}=\fl{\sqrt{1+\eps}\, r_k}, \ \ k=0,1,\cdots,k_{max},
\end{equation}
where $k_{\rm max}$ is the largest integer $k$ such that $r_k < K$. Let $\Omega_k=\Omega\left(\llbracket1,r_k\rrbracket ^d\right), k=0,1,\cdots,k_{max}$, and  $\Omega_\infty=\Omega \left(\llbracket1,K\rrbracket ^d\right)$.

Then for any cube $B\subset \Lambda$ of side length $\ell(B)=r_0$, and  $\check B=B/3 $ as in \eqref{eq:Q/3}, 
\begin{equation}\label{eq:Pr-key}
    \P\Bigl\{\,\omega:\,  \max_{\xi\in\check B } u_{\xi} \ge Mr_0^2 \, \Bigr\} \, \le \Prr{ \Omega_\infty}+C\eps^{-d}\sum_{k=0}^{k_{\rm max}}\Prr{\Omega_k},
\end{equation}
for some dimensional constant $C$. 
\end{lemma}

\begin{proof} The idea is to repeatedly use Lemma \ref{lem:climb} to construct a sequence of growing cubes and stop when the final cube  exceeds the size of the entire domain. 

We start with $B_0=B$ of side length $\ell(B_0)=r_0$. Assume that $\max_{\check B _0} u({\xi}) \ge Mr_0^2$. We pick some $\xi_0  \in \check B _0$ such that $u(\xi_0) \geq Mr_0^2$. Then \eqref{eq:Mr2} is satisfied. Suppose furthermore that $\cE_0:=\Omega_0$ fails. Then \eqref{eq:card-J} holds  for $B_0$. All in all,  Lemma \ref{lem:climb}  gives a point 
$\xi_1$, such that 
$
  \xi_1\in  {\cC_{r_1}}:=\{n: \dist(n,\check B _0)\le r_1  \}, 
$
where $\dist$ is measured in $|\cdot|_\infty$ for $\Z^d$ lattice points 
and  
\[u(\xi_1) \geq (1+ \varepsilon) u(\xi_0)\ge Mr_1^2.\]
Additionally, we can require $r_0\ge 15/\eps$ which implies that $r_1>(1+\eps/3)r_0$.   Clearly, 
\[
    \cards{{\cC_{r_1}}}\le
    (2r_0+2r_1)^d\le ({12}/{\eps})^dr_1^d.\]
Therefore, ${\cC_{r_1}}$ can be covered by at most $n_\ast=\fl{36^d\eps^{-d}}+1$ disjoint cubes in $\Z^d$ of side length $\fl{r_1/3}$, namely, $\check B _1^{(1)},\check B _1^{(2)}, \cdots, \check B _1^{(n_\ast)}$.  Recall the definition of the middle third set in \eqref{eq:Q/3}.  Now extend each $\check B _1^{(j)}$ to a cube
 $B_1^{(j)}$ such that $B_1^{(j)}$ has side length $r_1$ and contains each $\check B _1^{(j)}$ as a middle third part for $j=1,2,\cdots,n_\ast$.   Note that these $B_1^{(j)}$ are not disjoint. But the overlap does not effect our estimate on the probability of the events from above. 
 
In order for the induction driven by Lemma \ref{lem:climb} to continue, we need to exclude the event that \eqref{eq:card-J} fails for all $B_1^{(j)}$.  To this end, we define
$
    \cE_1=\bigcup_{j=1}^{n_\ast}\Omega(B_1^{(j)}) .
$
Assume that $\cE_1$ fails, which implies that \eqref{eq:card-J} holds  for all $B_1^{(j)}$. Let $B_1=B^{(j)}_1$ be the one that contains $\xi_1$. Now for $\xi_1\in \check B _1^{(j)}\subsetneqq B_1^{(j)}$,  Lemma \ref{lem:climb}  gives
$\xi_2$ such that
\begin{equation}\label{eq:xi12}
    |\xi_2-\xi_1|\le \fl{\sqrt{1+\eps}\, r_1}=r_2, \qquad u(\xi_2) \geq (1+ \varepsilon) u(\xi_1)\ge Mr_2^2. 
\end{equation} 

Repeat the construction for $r_2>(1+\eps/3)r_1$ and $\xi_2\in 
    {\cC_{r_2}}:=\{n: \dist(n,\check B _0)\le r_1+r_2  \}$,
    \begin{equation*}
         \cards{{\cC_{r_2}}}
       \le (2r_0+2r_1+2r_2)^d \le  ({12}/{\eps})^d r_2^d. 
    \end{equation*}
Therefore,  ${\cC_{r_2}}$ can  be covered by at most $n_\ast=\fl{36^d\eps^{-d}}+1$ disjoint cubes of side length $\fl{r_2/3}$, $\check B _2^{(1)},\check B _2^{(2)}, \cdots, \check B _2^{(n_\ast)}$. Extend $\check B _2^{(j)}$ to $ B_2^{(j)}$ in the same way, and define
$
    \cE_2=\bigcup_{j=1}^{n_\ast}\Omega(B_2^{(j)}).
$
We assume that $\cE_2$ fails, and find $\xi_3$ by Lemma \ref{lem:climb}. Inductively,  at step $k$,  we assume that all the previous events, $\cE_1,\cE_2,\cdots,\cE_{k-1}$ fail, and  obtain $\xi_{k}, r_{k}$ satisfying \eqref{eq:MR2} and then define $B_{r_{k}},\cC_{r_k}$.  The same estimates as for $r_2,{\cC_{r_2}}$  hold  for all $k$, and hence,
\begin{equation}
  (1+\eps/3)^k r_{0} \le\cdots \le (1+\eps/3)r_{k-1} \le r_k\le \sqrt{1+\eps}\, r_{k-1} \le \cdots \le ({1+\eps})^{k/2}\, r_{0}\label{eq:Bkrk}
  \end{equation}
  and
\begin{equation}\cards{{\cC_{r_k}}}\le (2r_0+2r_1+\cdots 2r_k)^d \le 2r_k\  \frac{1}{1-(1+ \varepsilon/3)^{-1}}<(12/\eps)^d r_k^d.\label{eq:card-Bk}
  \end{equation}
  Then $\check B ^{(j)}_{k}$,  $B^{(j)}_{k}$ are defined in the same way as we did in the previous steps. Because of \eqref{eq:card-Bk}, for all $k$, we need $n_\ast=\fl{36^d\eps^{-d}}+1$ many $\check B ^{(j)}_{k}$ to cover ${\cC_{r_k}}$. Then we define the event 
$
    \cE_k=\bigcup_{j=1}^{n_\ast}\Omega(B_k^{(j)}).
$
Since $v_n$ are i.i.d., the  probability of each $\Omega(B_k^{(j)})$ is translation invariant, and only depends on the size of the cube $B_k^{(j)}$. In particular, $
    \P \bigl\{\Omega(B_k^{(j)})\bigr\}=\Prr{\Omega_k} . 
$
Therefore, for some dimensional constant $C$,
\begin{equation}\label{eq:PrEk}
    \Prr{\cE_k}\le \, n_\ast \Prr{\Omega_k}\le \, C \eps^{-d} \,\Prr{\Omega_k},
\end{equation}
provided that $\eps<1$. 

We will continue the construction until we reach the $k_{\rm max}$-th step and obtain $\xi_{k_{\rm max}}, r_{k_{\rm max}}$, $B_{r_{k_{\max}}}$, $\check B ^{(j)}_{k_{\max}}$,  $B^{(j)}_{k_{\max}}$ and  $\cE_{k_{max}}$.  

We need to apply Lemma \ref{lem:climb} two more times for the final step. However, according to our choice of  $k_{\max}$, assuming that $\cE_{k_{\max}}$ fails will already result a cube of side length $\fl{\sqrt{1+\eps}\, r_{k_{\max}}}\ge K$ which may exceed the maximal size of the entire domain $\Lambda$. To alleviate this issue, we need to enlarge the domain at this point for the last two steps, by making several copies of $\Lambda$\footnote{We can do this from the very beginning of the construction, but it will not make any difference until we reach the size of $K$.}.  We need 
$p^d$ many copies where $p:=\fl{3\sqrt{1+\eps}}+1$.

Let $\widetilde K=p K$ so that 
$
    \widetilde K>\fl{\sqrt{1+\eps}\,( 3K) }
$
and denote 
\[p\Lambda:=\llbracket1,\wt K\rrbracket^d=\bigcup_{j\in (\Z\,{\rm mod}\,p\Z)^d  }\left(\Lambda+j\,K\right).\]
We extend the potential $V=\{v_n\}_{n\in \Lambda }$ periodically to $ V'=\{ v'_n\}_{n\in p\Lambda}$, where 
$  v'_{m}= v_n,\ \ n\in \Lambda,$ and  $m=n\ {\rm mod} (p\Z)^d.$
Now we consider the landscape equation on $p\Lambda$ for $-\Delta+ V'$ with periodic boundary conditions. The enlarged system has a unique solution $ u'$ by Theorem \ref{thm:landscape}, which is a periodic extension of the original $u$ to $p\Lambda$. Now we can return to the construction at the $k_{\max}$-th step. 

Assume that the event $\cE_{k_{\max}}$ fails. Then \eqref{eq:card-J} holds for all possible $B_{k_{\max}}\subset \llbracket1,\wt K\rrbracket^d$ that may contain $\xi_{k_{\max}}$. We have
$
    u'_{\xi_{k_{\max}}}=u_{\xi_{k_{\max}}} \ge M r^2_{\xi_{k_{\max}}}.
$
Applying Lemma \ref{lem:climb} to $u^{\Gamma}$ on $B_{k_{\max}}\subset \llbracket1,  K\rrbracket^d$, we  obtain   $\wt \xi \in \llbracket1$,  $K\rrbracket^d$ such that 
\begin{equation}\label{eq:xi-wt}
    u'_{\wt \xi} \ge (1+\eps) u_{\xi_{k_{\max}}} \ge M \fl{\sqrt{1+\eps}r_{\xi_{k_{\max}}}}^2\ge  M K^2\, ,
\end{equation}
where the last inequality follows from the definition of $k_{\max}$. Now let $\xi_{\infty}\in \llbracket1,  K\rrbracket^d$ be a point where $u_k$ attains its maximum. Clearly, $\wt u_k$ also attains its maximum at $\xi_{\infty}$, 
\begin{equation}\label{eq:M-9}
    u_{\xi_{\infty}}=\max_{ k \in \Lambda}u_k=  \max_{ k \in p\Lambda} u'_k=  u'_{\xi_{\infty}}\, .
\end{equation}
Together with \eqref{eq:xi-wt}, we have 
$
    u_{\xi_{\infty}}= u'_{\xi_{\infty}}\ge  u'_{\wt \xi}\ge MK^2 . 
$
 Now consider $\check B _{\infty}=\Lambda$, which is the middle third of $B_{\infty}=3\Lambda\subset p\Lambda$. Let 
\begin{align*}
    \cE_{\infty}
    :=&\left\{\,   \cards{j\in 3\Lambda:  v'_j\ge C_P(3K)^{-2}}\le  \lambda |3\Lambda|\, \right\} \\
    =& \left\{\,   \cards{j\in \Lambda:  v_j\ge \frac{C_P}{9}\, K^{-2}}\le  \lambda K^d\, \right\} .
\end{align*}
Since $\cards{j\in \Lambda:  v_j\ge \frac{C_P}{9}\, K^{-2}} \ge \cards{j\in \Lambda:  v_j\ge C_P\, K^{-2}}$, one has
\begin{align}\label{eq:PrEinf}
    \cE_{\infty} \subseteq \left\{\,   \cards{j\in \Lambda:  v_j\ge C_P\, K^{-2}}\le  \lambda K^d\, \right\}=\Omega_\infty\, .
\end{align}
Now if $\cE_{\infty}$ fails, apply Lemma \ref{lem:climb} one last time to $ u'$ on $\check B _{\infty}\subsetneq B_{\infty}$\footnote{One also needs to take $M$ nine times larger, so that $M/9\ge M_0$ in \eqref{eq:M-9} to meet the requirement in Lemma \ref{lem:climb}.}. Then \eqref{eq:MR2} implies that there is a $\xi' \in \llbracket1,  \wt K\rrbracket^d$ such that $ u'_{\xi'}\ge (1+\eps) u'_{\xi_\infty}> u'_{\xi_\infty}$. This is a contradiction. Recall that it happens when we start with $u_{\xi_0}\ge Mr_0^2$ and assume that all $\cE_j$ fail. Therefore, at least one $\cE_j$ must be true. In other words, 
\begin{equation*}
    \Bigl\{\,  \max_{\check B _0} u_{\xi} \ge Mr_0^2 \, \Bigr\} \subset \cE_\infty\bigcup \Bigl( \bigcup_{j=1}^{k_{\max}} \cE_j \Bigr)
 \Longrightarrow    \P \Bigl\{  \max_{\check B _0} u_{\xi} \ge Mr_0^2 \Bigr\} \le \Prr{\cE_\infty }+\sum_{j=1}^{k_{\max}} \Prr{\cE_j}.
\end{equation*}
Together with \eqref{eq:PrEk} and \eqref{eq:PrEinf}, this 
completes the proof. 
\end{proof}

The next lemma allows us to estimate the probability of each term on the  right hand side of \eqref{eq:Pr-key}. 
\begin{lemma}\label{lem:Prr}
Let $V=\{v_j\}_{j\in \Lambda}$ be the Anderson potential as in Theorem \ref{thm:Nu-iid}. Let $F$ and $\delta_\ast$ be as in \eqref{eq:ass-F}.   For any $B\subset \Lambda$ and $0<\lambda<1$, if $\delta> 0$ is such that $1-\lambda-F(\delta) >0$, then
\begin{equation}\label{eq:pr1}
    \Pr{ \cards{j\in B : v_j\ge \delta }\le  \lambda |B| }
   \le  e^{-D(1-\lambda \| F)\,\, |B|}, 
\end{equation}
where 
\begin{equation}\label{eq:KLdiv}
    D(x\|y)=x\log  \frac{x}{y}+(1-x)\log  \frac{1-x}{1-y}
\end{equation}
is the Kullback–Leibler divergence between Bernoulli distributed random variables with parameters $x$ and $y$ respectively.

As a consequence, for any $r\in \N$,
\begin{equation}\label{eq:pr2}
    \Pr{ \cards{j\in \llbracket 1, r \rrbracket ^d : v_j\ge \delta}\le  \lambda r^d } \, \le \, \left(C(\lambda)F(\delta)^{1-\lambda}\right)^{r^d},
\end{equation}
where $C(\lambda)=(1-\lambda)^{-1}\lambda^{-\lambda}$.

Furthermore, there is a $\lambda_\ast>0$,  which only depends
on $F(\delta_\ast),$
such that for all $0<\delta\le \delta_\ast,0<\lambda \le \lambda_\ast$ and any $r\in \N$, one has 
\begin{equation}\label{eq:pr3}
    \Pr{ \cards{j\in \llbracket 1, r \rrbracket ^d : v_j\ge \delta}\le  \lambda r^d } \, \le \, \left(F(\delta)\right)^{{r^d}/{2}}.
\end{equation}

\end{lemma}

\begin{remark}\label{rem:lambda-ast}
    The $\lambda_\ast$ can be taken as of order $(1-F(\delta_\ast))^2$, see \eqref{eq:lambda-ast}. 
\end{remark}

\begin{proof}
Let $\zeta_j$ be the characteristic function for the event $v_j\leq \delta$, i.e., $ \zeta_j=1$ for $v_j\leq \delta$ and $ \zeta_j=0$ otherwise. 
Since $\{v_j\}$ are i.i.d. random variables, all $\zeta_j$ are i.i.d. Bernoulli random variables, taking values in $\{0,1 \}$, with common expectation $\Ev{\zeta_j}=\Prr{v_j\leq \delta}=F(\delta)$. Let  $S_B:=\sum_{j\in B } \zeta_j. $ 
By the Chernoff–Hoeffding Theorem, \cite{Ho} (see Lemma \ref{lem:Chernoff}),
\begin{equation}\label{eq:SI}
    \Pr{ S_B\ge (1-\lambda)|B|} \le e^{-D(1-\lambda \| F)\, |B|},
\end{equation}
where $F=F(\delta)$ and $
    D(x\|y)$ is as in \eqref{eq:KLdiv}. 
Then \eqref{eq:pr1} follows directly from \eqref{eq:SI} since \[|B|-\cards{j\in B : v_j> \delta }=\cards{j\in B : v_j\le \delta } =\sum_{j\in B } \zeta_j=S_B. \]

Examining the  the Kullback–Leibler divergence with parameter $1-\lambda$ and $F$, one has
\begin{align*}
    D(1-\lambda \| F)= (1-\lambda)\log  \frac{1-\lambda}{F}+\lambda \log  \frac{\lambda}{1-F}  
    \ge & \log \left((1-\lambda)^{1-\lambda}\lambda^{\lambda}\right)\, -\, \log  F^{1-\lambda} \\
    \ge & \log \left((1-\lambda)\lambda^{\lambda}\right)\, -\, \log  F^{1-\lambda}, 
\end{align*}
where we used $1-F <1$  and $(1-\lambda)^{1-\lambda}\ge 1-\lambda$. 
Therefore, 
\begin{equation} \label{eq:temp4}
    \Pr{ \cards{j\in [1,r]^d\cap\Z^d: v_j\ge \delta}\le  \lambda r^d }  \le  e^{-D(1-\lambda \| F)\, r^d}
     \le \left((1-\lambda)^{-1}\lambda^{-\lambda} F^{1-\lambda} \right)^{r^d},
\end{equation}
which yields \eqref{eq:pr2}. 

Let $q=1-F(\delta_\ast)\in (0,1)$. For $0< \delta \le \delta_\ast$,
\begin{equation} \label{eq:delta-ast}
    F(\delta)\le F(\delta_\ast)=1-q<1\ \Longrightarrow \  \log F(\delta)\le \log (1-q)<0.
\end{equation}
On the other hand, it is easy to check that 
\begin{equation}\label{eq:F0}
    \lim_{\lambda\to 0^+} \frac{\log \left((1-\lambda)\lambda^{\lambda} \right)}{1/2-\lambda } =0.
\end{equation}
Then there is a $\lambda_\ast \le 1/2$  such that  for all $0<\lambda\le \lambda_\ast$,  
\begin{equation}
     \frac{\log \left((1-\lambda)\lambda^{\lambda} \right)}{1/2-\lambda } >  \log (1-q) \ge  \log F(\delta) 
     \Longrightarrow  (1-\lambda)^{-1}\lambda^{-\lambda}  <  \left( F(\delta) \right)^{ \lambda-1/2}.  \label{eq:temp5}
\end{equation}
Combined with \eqref{eq:temp4}, this gives
\[\Pr{ \cards{j\in \llbracket 1, r \rrbracket ^d : v_j\ge \delta}\le  \lambda r^d } \, \le \, \left( F(\delta) \right)^{r^d/2}, \]
which completes the proof of Lemma \ref{lem:Prr}. 

One can be more specific regarding the exact value of $\lambda_\ast$.  For $0<\delta\le \delta_\ast$, 
    $\log F(\delta) \le \log(1-q)< -\min(q,\frac{1}{2})$. 
    If $\lambda<1/4$, then
 $
        \log(1-\lambda)>-2\lambda>-\sqrt{\lambda}$, and  $
        \lambda\,\log\lambda >-\sqrt{\lambda}.$
Let 
    \begin{equation}\label{eq:lambda-ast}
        \lambda_\ast= \left(\frac{1}{8}\min(q,1/2)\right)^2.
    \end{equation}
    Then for $\lambda<\lambda_\ast\le 1/4$, 
    \begin{equation*}
       0> \frac{\log \left((1-\lambda)\lambda^{\lambda} \right)}{1/2-\lambda }\ge  4\left(\log ((1-\lambda)+\lambda\log \lambda\right) 
       \ge  -8\sqrt \lambda 
       \ge  - \min(q,\frac{1}{2})\ge \log F(\delta), 
    \end{equation*}
    which gives \eqref{eq:temp5} similarly to the argument above. 
\end{proof}

Combining Lemma \ref{lem:Pr} and Lemma \ref{lem:Prr} leads to
\begin{lemma}\label{lem:Pr-M}
Let $\delta_\ast$ and $\lambda_\ast$ be  as in Lemma \ref{lem:Prr}. Fix $\lambda\le \lambda_\ast$, and take $\eps<\eps_0(\lambda,d)$ and  $C_P,M,r_{\ast}$ as in Lemma \ref{lem:Pr}. Then for any cube $B\subsetneq \Lambda$ of side length $\ell(B)=r\ge r_\ast$, and its middle third part $\check B $, one has  
\begin{equation}\label{eq:Pr-key1}
    \P\Bigl\{\,  \max_{\xi\in\check B } u_{\xi} \ge Mr^2 \, \Bigr\} \, \le \frac{C\eps^{-d}}{1-F(\delta_\ast)} \left( F(C_Pr^{-2}) \right)^{r^d/2}
\end{equation}
for some dimensional constant $C>0$. 
\end{lemma}
\begin{remark}
The exponent $r^d/2$ can be made arbitrarily close to $r^d$, by taking $\lambda$ smaller, however, it will also result a large factor $1/\lambda$ in front of $F$. 
\end{remark}

\begin{proof}
Let $r_0=r$ and define the sequence $r_k$ as in Lemma \ref{lem:Pr} and $r_\infty:=K$. Let $\delta_k=C_P r_k^{-2},k=0,\cdots, k_{\max}$, and $\delta_\infty=C_PK^{-2}$. By the construction of $r_k$ and \eqref{eq:Bkrk}, one has 
\begin{equation*}
    r^d_k\ge (1+\eps/3)^{dk} r^d_0\ge (1+kd\eps/3)r^d_0\ge r^d_0+2k, \ \ k=0,\cdots,k_{\max}.
\end{equation*}
For $k=0,\cdots,k_{\max}$ and $k=\infty$, one has $\delta_k \le \delta_0=C_Pr_0^{-2} \le \delta_\ast$
provided that $r_0\ge \sqrt{C_P/\delta_\ast}$. 
 Notice that in the proof of Lemma \ref{lem:Prr}, by the choice of $\delta_\ast$ in \eqref{eq:ass-F}, one has $F(\delta_0)<F(\delta_\ast):=1-q$. Therefore, 
 \begin{equation*}
   F(\delta_k)\le 1-q<1, \ \ k=0,\cdots,k_{\max}, \ {\rm and}\  k=\infty,
\end{equation*}
 since the distribution $F$ is non-decreasing. Now apply Lemma \ref{lem:Prr} to all $r_k$. Combining  \eqref{eq:pr3} with \eqref{eq:Pr-key}, one has 
 \begin{multline*}
     \P\left\{\,  \max_{\xi\in\check B } u_{\xi} \ge Mr_0^2 \, \right\} \, \le
     F(\delta_\infty)^{K^d/2}+C\eps^{-d}\sum_{k=0}^{k_{\max}}F(\delta_k)^{r_k^d /2} \\
     \le  F(\delta_0)^{r^d_0/2}+C\eps^{-d}\sum_{k=0}^{k_{\max}}F(\delta_0)^{(r_0^d+2k)/2}
      \le  F(\delta_0)^{r_0^d/2}+C\eps^{-d}F(\delta_0)^{r_0^d/2}  \frac{1}{1-F(\delta_0)}\\
      \le   F(\delta_0)^{r_0^d/2}\left(1+\frac{C\eps^{-d}}{q}\right):=C(d,\eps,\delta_\ast) F(\delta_0)^{r_0^d/2},
 \end{multline*}
 which is the desired bound. 
\end{proof}

Now we are ready to complete:
\begin{proof}[Proof of \eqref{eq:Nu<F} in Theorem \ref{thm:Nu-iid}]
Let $\delta_\ast$ and $\lambda_\ast$ be given by Lemma \ref{lem:Pr-M}. Fix $\lambda\le \lambda_\ast$, take $\eps<\eps_0(\lambda,d)$ and  $C_P,M,r_\ast$ as in Lemma \ref{lem:Prr}.

For any $\mu \le 1/(4M)$, let $r=\cl{(4M\mu)^{-1/2}  }$ so that $\mu^{-1}/4<Mr^2\le \mu^{-1}$. To apply Lemma \ref{lem:Pr-M}, one also needs to ensure that $r\ge r_\ast$, which requires $\mu$ to be taken in the range $\mu\le \mu_\ast= 1/(Mr_\ast^2)$. 

Now for any cube $B$ of side length $r=\cl{(4M\mu)^{-1/2}  }$ and its middle third part $\check B$, Lemma \ref{lem:Pr-M} implies that
\begin{equation*}
\P\Bigl\{\,  \max_{\xi\in\check B } u_{\xi} \ge \mu^{-1} \, \Bigr\}  \le     \P\Bigl\{\,  \max_{\xi\in\check B } u_{\xi} \ge Mr^2 \, \Bigr\} \, \le \, C \cdot\left( F(C_Pr^{-2}) \right)^{r^d/2}, 
\end{equation*}
where $C>0$ in given by \eqref{eq:Pr-key1}  depending on $d,\eps$ and $F(\delta_\ast)$. Then
\begin{equation}
\P\Bigl\{\,  \min_{n\in\check B } \frac{1}{u_{n}} \le \mu \, \Bigr\}  \le \, C  \left( F(C_Pr^{-2}) \right)^{r^d/2}
\le  \, C  \big( F(4C_PM\,\mu) \big)^{(M\mu)^{-d/2}/2}. \label{eq:Pr-J}
\end{equation}

Notice that $\ell(\check B)\ge r/6\ge (4M\mu)^{-1/2}/6$. 
Recall that the cubes used in the definition of $N_u$   have side length $\cl{\mu^{-1/2}}$. Any $Q\in \cP\bigl(\cl{\mu^{-1/2}};\Lambda\bigr)$  can be covered by at most $\left(\frac{\cl{\mu^{-1/2}}}{(4M\mu)^{-1/2}/6}\right)^d+1\le \,  C'  M^{d/2}$ disjoint cubes of side length $\ell(\check B)=\cl{\cl{(4M\mu)^{-1/2}}/3}$ for some $C'$ which only depends on $M,d$. Notice also that the estimate \eqref{eq:Pr-J} is independent of the position of $\check B $, and can be applied to all cubes $\check B $ of the same size. Therefore, 
\begin{equation*}
  \P\left\{ \min_{n\in Q}\frac{1}{u_n}\,\le  \mu \right\} \le \, C' M^{d/2}\, \P\left\{\,  \min_{n\in\check B } \frac{1}{u_{n}} \le \mu \, \right\}  \le \,  C'' \big( F(C_3\,\mu) \big)^{\gamma_2\mu^{-d/2}}
\end{equation*}
for any $Q$.
Together with \eqref{eq:min-max-2}, we obtain the desired upper bound 
\begin{equation*}
     \Ev{N_u(\mu)}\le 
 \frac{1}{\cl{\mu^{-1/2}}^d}\max_{Q\in\cP(\cl{\mu^{-1/2}})}\P\left\{ \min_{n\in Q}\frac{1}{u_n}\,\le \mu \right\}
\le \, C_4\mu^{d/2} \big( F(C_3\,\mu) \big)^{\gamma_2\mu^{-d/2}} 
\end{equation*}
for all $\mu\le \mu_\ast$. The constants $C_3,C_4,\gamma_2>0$ only depend on $d,M$ and $C_P$, which eventually only depend on $d$ and $F_\ast=F(\delta_\ast)$. 
\end{proof}

\subsection{Lifschitz tails for the integrated density of states}\label{sec:LSlaw-iid}
Putting together the general upper/lower bounds in Theorem \ref{thm:NNu-intro},\ref{thm:NNu-intro-2} for the deterministic case, and the Lifshitz tails in Theorem \ref{thm:Nu-iid} for the Anderson model, we have 
\begin{theorem}\label{thm:randomNNu}
Let $C_1$ be as in Theorem \ref{thm:NNu-intro} and $\delta_\ast$ be as in Theorem \ref{thm:Nu-iid}. Then there are  constants $c_5,c_6>0$ depending  on $d$, $\delta_\ast$ and $V_{\max}$ such that 
\begin{equation} \label{eq:NNu-And}
   c_5\, \E N_u(\, c_6\,\mu) \, \le \,  \E N(\mu)\, \le\,  \E N_u(C_1\,\mu), \ {\rm for\ all}\ \mu>0.
\end{equation}
If furthermore $\mu_\ast,K_\ast$ are as in Theorem \ref{thm:Nu-iid}, depending only on $d$ and $\delta_\ast$, and  $\mu<\mu_\ast$, then the estimate \eqref{eq:NNu-And} holds with constants $c_5,c_6$ which are independent of $V_{\max}$. 

If, in addition,  $K_\ast/K^2<\mu<\mu_\ast/C_1$, then there are  constants $\bar C_1,\bar c_1, \bar C_2,\bar c_1, \bar \gamma_1,  \bar \gamma_2$  depending only on $\delta_\ast$ such that 
\begin{equation}\label{eq:N-tail}
 \bar c_2 \mu^{d/2}F(\bar c_1 \mu)^{\bar \gamma_1\mu_2^{-d/2}} \,  \le \, \E N(\mu)\, \le \, \bar C_2\, \mu^{d/2}F(\bar C_1\, \mu)^{\bar \gamma_2\mu^{-d/2}}. 
\end{equation}

\end{theorem}

\begin{proof}
The upper bound in \eqref{eq:NNu-And} is the average of the upper bound in Theorem \ref{thm:NNu-intro}. We only need to study the lower bound with the help of Theorem \ref{thm:NNu-intro-2} and Theorem \ref{thm:Nu-iid}. 
Let $c_\ast,c_0,c_1,C_1$ and $\alpha<\alpha_0<1$, be as in Theorem \ref{thm:NNu-intro-2}. If $\mu\ge 4d+V_{\max}$ then $N(\mu)=1$ and the  left-hand side of  \eqref{eq:NNu-And} holds trivially.  Fix $\mu<4d+V_{\max}$, let us denote $\mu_2=c_1\alpha ^{d+2}\mu$ and $\mu_4=\alpha^2\mu_2=c_1\alpha^{d+4}\mu$. 
 Assume further that $\alpha\le \alpha_{00}:=\big(c_\ast/(4d+V_{\max})\big)^{-1/4}$. Then  \eqref{eq:NNu-intro-2} in  Theorem \ref{thm:NNu-intro-2} implies that  
\begin{equation}
     \E N(\mu)\ge   c_0\alpha^d\E N_u(\mu_2)-C_0\E N_u(\mu_4). \label{eq:temp10}
\end{equation}

Next, let $c_3,c_4,C_3,C_4,\gamma_1,\gamma_2>0$ and $\mu_\ast$ be given by Theorem \ref{thm:Nu-iid}. Then by \eqref{eq:Nu>F} and \eqref{eq:Nu<F}, if $\mu_2\le 1$ and $\mu_4\le \mu_\ast$, one has
\begin{equation}\label{eq:F4}
 \E{N_u(\mu_2)}\ge c_4 \mu_2^{d/2}F(c_3\mu_2)^{\gamma_1\mu_2^{-d/2}},\ \ {\rm and} \ \  \Ev{N_u(\mu_4)} \le \,  C_4\mu_4^{d/2}F(C_3 \mu_4)^{\gamma_2 \mu_4^{-d/2}}.
\end{equation}

Therefore, 
\begin{align}
     \E N(\mu)     \ge & \,c_4\alpha^d \mu_2^{d/2}F(c_3 \mu_2)^{\gamma_1\mu_2^{-d/2}} - C_4\mu_4^{d/2}F(C_3 \mu_4)^{\gamma_2   \mu_4^{-d/2}}\nonumber\\
     =& \, c_4{\mu_4^{d/2}} \left(\,  F(c_3\mu_2)^{\gamma_1\mu_2^{-d/2}} -C_4 F(C_3 \alpha^2 \mu_2)^{\gamma_2 \mu_4^{-d/2}}    \, \right) \label{eq:temp6}
\end{align} 
for $\mu_4\le \mu_\ast,\mu_2\le 1$ and $\mu<4d+V_{\max}$. This requires $\alpha\le \min\{\alpha_1, \alpha_2\}$, where 
\begin{equation*}
     \alpha_1:=\big(c_1(4d+V_{\max})\big)^{-1/(d+2)},\ \ {\rm and}\  \alpha_2:=\mu_\ast^{1/(d+4)} \big(c_1(4d+V_{\max})\big)^{-1/(d+4)}.
\end{equation*}

Let $\delta_\ast$ be as in \eqref{eq:ass-F}. If we assume, in addition, that $\alpha$ is smaller than both $\alpha_3$ and $\alpha_4$, 
\begin{equation*}
 \alpha_3:=\sqrt{c_3/C_3}, \ \ {\rm and}\  \alpha_4:=\delta_\ast^{1/(d+2)}  \big(c_3  c_1 (4d+V_{\max})\big)^{-1/(d+2)},
\end{equation*}
then for all $\mu<4d+V_{\max}$, one has $C_3\mu_4< c_2\mu_2\le \delta_\ast$.
Therefore, $0< F_4 \le F_2 \le F_\ast<1$, where 
$
F_4=F(C_3  \mu_4)=F(C_3\alpha^2 \mu_2)$, $F_2= F(c_3\mu_2)$, and $F_\ast=F(\delta_\ast).$
The difference term in \eqref{eq:temp6} is then bounded from below by
\[
    F_2^{\gamma_1\mu_2^{-d/2}}-C_4  F_4^{\gamma_2\,  \mu_4^{-d/2}}\ge  F_4^{\gamma_1\mu_2^{-d/2}}-C_4 F_4^{\gamma_2\,  \mu_4^{-d/2}}.
\]
We want to pick $\alpha$ small enough (independent of $\mu$) so that, 
\begin{equation}
F_4^{\gamma_1\mu_2^{-d/2}}-C_4 F_4^{\gamma_2\,  \mu_4^{-d/2}}\ \ge \ \frac{1}{2}F_4^{\gamma_1\mu_2^{-d/2}}, \label{eq:temp7} 
  \end{equation}
  that is,   
\begin{equation}  
(2C_4)^{-1} \ge \ F_4^{\gamma_2\,  \mu_4^{-d/2}-\gamma_1\mu_2^{-d/2}}=F_4^{\mu_2^{-d/2}  (\gamma_2\, \alpha^{-d} -\gamma_1\, )}.\label{eq:temp8} 
\end{equation}

Notice that $\mu_2<1$, hence, $\mu_2^{-d/2}>1$ and $\mu_2^{-d/2}(\gamma_2\alpha^{-d}-\gamma_1)>\gamma_2\alpha^{-d}-\gamma_1>0$ provided that $\alpha<(\gamma_2/\gamma_1)^{1/d}:=\alpha_5$. Then the fact that $F_4\le F_\ast<1 $ implies that 
\begin{equation*}
   F_4^{\mu_2^{-d/2}(\gamma_2\alpha^{-d}-\gamma_1)} \le F_\ast^{\mu_2^{-d/2}(\gamma_2\alpha^{-d}-\gamma_1)} \le F_\ast^{\gamma_2\alpha^{-d}-\gamma_1}.
\end{equation*}
Solving  $1\ge (2C_4)^{-1} \ge F_\ast^{\gamma_2\alpha^{-d}-\gamma_1}$ for $\alpha$, we observe that 
\begin{equation*}
  \gamma_2\alpha^{-d}-\gamma_1\ge \frac{\log(2C_4)}{\log F_\ast} \Longleftrightarrow \alpha\le \left(\gamma_2^{-1}\frac{\log(2C_4)}{\log F_\ast}+\gamma_2^{-1}\gamma_1  \right)^{-1/d}:=\alpha_6
\end{equation*}
would yield \eqref{eq:temp8}. 

Putting everything together, set
\begin{equation*}
\alpha=\alpha_\ast:=\min\{\alpha_0,\alpha_{00},\alpha_1,\alpha_2,\cdots,\alpha_6\}.
\end{equation*}
Then, for all $\mu<4d+V_{\max}$, 
\eqref{eq:temp6}, \eqref{eq:temp7} and \eqref{eq:F4} imply that 
\begin{equation*}
   \E N(\mu)
   \ge   \frac{1}{2} \,c_4{\mu_4^{d/2}} F_4^{\gamma_1\mu_2^{-d/2}}
   \ge  \frac{1}{2} \,c_4   C_4^{-1} \E N_u(\mu_4)
  =\frac{1}{2}c_4 C_4^{-1}\, \E N_u(c_1\alpha_\ast^{d+4} \mu) =:c_5\,\E N_u(c_6\ \mu)\, ,
\end{equation*}
which completes the proof for the  first inequality in \eqref{eq:NNu-And}.

It is also easy to verify that if we are only interested in small $\mu$, then all the $\alpha_i$ can be picked independently of $V_{\max}$. Therefore, the final constants $c_5,c_6$ are also independent of $V_{\max}$. In particular, let $\mu_\ast$ be as in Theorem \ref{thm:Nu-iid}. Then for all $c_\ast'\mu<C_1\mu<\mu_\ast$
\begin{equation*}
     \E N(\mu) \le  \E N_u(C_1\mu) \le \, C_4 (C_1\mu)^{d/2}F(C_3 C_1\mu)^{\gamma_2   (C_1\mu)^{-d/2}}=:\bar C_2 \mu^{d/2}F(\bar C_1\mu)^{\bar \gamma_2  \mu ^{-d/2}}
\end{equation*}
and 
\begin{equation*}
     \E N(\mu) \ge   c_5\E N_u(C_6\mu) \ge c_5C_4 (c_6\mu)^{d/2}F(c_3 c_6\mu)^{\gamma_1   (C_6\mu)^{-d/2}}=:\bar c_2 \mu^{d/2}F(\bar c_1\mu)^{\bar \gamma_2  \mu ^{-d/2}}
\end{equation*}
where the constants $\bar c_1,\bar c_2,\bar C_1,\bar C_2,\bar \gamma_1,\bar \gamma_2$ only depend  on $d$ and $\mu_\ast$, and  are independent of $V_{\max}$. 
\end{proof}

\subsection{Dual landscape and the top edge of the spectrum. }\label{sec:dual}
Let $H=-\Delta+V$ be as in \eqref{eq:opH-intro} acting on $\cH=\ell^2(\Lambda), \Lambda=(\Z/K\Z)^d$. In this part, we will briefly discuss the so-called dual landscape and see how it is applied to the eigenvalue-counting for high energy modes. We refer readers to Section 2.4 in \cite{WZ} for more details.  For $\varphi\in \cH=\ell^2({\Lambda}) $, we define a dual vector $\wt \varphi$ 
\begin{equation}\label{eq:dualU}
    \wt{\varphi}_n=(-1)^{s(n)}\varphi_n, \ n\in {\Lambda},
\end{equation}
where $s(n)=\sum_{j=1}^d n_j$ for $n=(n_1,n_2,\cdots,n_d)\in \Z^d$.   We assume, in addition, that $K$ is an even number so that $\wt \varphi _n=\wt \varphi _{n+Ke_i},n\in\Lambda,i=1\cdots,d$.  
Now suppose  $(\mu,\varphi)$ is an  eigenpair of $H=-\Delta+V$ in $\cH=\ell^2({\Lambda})$. A direct computation shows that 
\begin{equation}\label{eq:dual-eq}
    (-\Delta +V_{\max}-V)\,\wt{\varphi}=\wt \mu\, \wt{\varphi},
\end{equation}
 where  $V_{\max}-V=\{V_{\max}-v_n\}_{n\in\Lambda}$ is a non-negative potential  and 
 \begin{equation}\label{eq:dual-ev}
     \wt{\mu}=4d+V_{\max}-\mu\, .
 \end{equation}

In other words,   $(\mu,\varphi)$ is an eigenpair of $H $ if and only if $(\wt \mu,\wt{\varphi})$ is an eigenpair of a dual operator $\wt{H} :=-\Delta +V_{\max}-V$. This dual operator $\wt H$ is the same type of discrete Schr\"odinger operator as $H$, only with a different potential (and also taking values in $[0,V_{\max}]$). We can define a dual landscape function $\wt u$ satisfying $(\wt H \wt u)_n=1$, and a dual box-counting function $N_{\wt u}(\mu;\wt H)$  as in \eqref{eq:Nudef-intro} for $\wt H$. 
  
  It is easy to check that the cardinality of the eigenvalues of $H$ which are smaller than or equal to $\mu$ is the difference of the volume of $\Lambda$ and the cardinality of the eigenvalues of $\wt H$ which are smaller than $\wt \mu$.  Therefore, 
  \begin{equation}\label{eq:dualN}
      N(\mu;H)=1- N^{-}\bigl(\wt \mu;\wt H\bigr), 
  \end{equation}
  where $N(\cdot;H)$ and $N^{-}(\cdot;{\wt H})$ are the finite volume integrated density of states for $H$ and $\wt H$, respectively. Here, the counting $N^{-}(\cdot;{\wt H})$  is defined for eigenvalues strictly less than $\mu$, which is different from the definition of $N(\cdot;H)$ in \eqref{eq:Ndef-intro}.   If $V$ is  the Anderson-type   potential  with common distribution $F(\delta)=\P(v_n\le \delta)$,  then$V_{\max}-V$ is also an Anderson-type   potential,  with common distribution $\P(V_{\max}-v_n\le \delta)$. We denote by $\wt F(\delta)=\P(V_{\max}-v_n< \delta)=1-F(V_{\max}-\delta)$. We now  apply Theorems \ref{thm:NNu-intro}, \ref{thm:NNu-intro-2}, \ref{thm:NNu-iid-intro}  and \ref{thm:Nu-iid}  to the dual operator $\wt H$ and the dual counting function $N(\wt \mu;\wt H)$ for $\wt \mu$ near $0$. All the estimates still hold  if we replace $N,F$ by $N^{-}$ and $\wt F$. In particular, the first part of Theorem \ref{thm:NNu-iid-intro} implies that there are constants $\wt c_5,\wt c_6$  depending  on $d$, the expectation of the random variables, and $V_{\max}$ 
such that for all $\wt \mu>0$,
  \begin{equation} \label{eq:NNu-dual}
   \wt c_5\,\E N_{\wt u}\,(\wt c_6\,\wt\mu;\wt H) \le \,\E N^{-}\,(\wt\mu;\wt H)\le  \,\E N_{\wt u}(\,C_1\wt \mu;\wt H).
\end{equation}
Therefore, by \eqref{eq:dual-ev} and \eqref{eq:dualN}, one has for all $\wt \mu=4d+V_{\max}-  \mu$,
\begin{equation*}
1- \,\E N_{\wt u}\,(C_1\wt \mu;\wt H) \le    \E N(\mu; H) \le 1-  \wt c_5\,\E N_{\wt u}\,(\wt c_6\,\wt\mu;\wt H) ,
\end{equation*}
 which yields  \eqref{eq:NNu-top-intro} in Corollary \ref{cor:duallandscape}. 


\appendix
\section{Discrete Laplacian and harmonic functions}\label{sec:app}

\subsection{Maximum principle for sub-solutions}
\begin{lemma}[The maximum principle for subharmonic functions]\label{lem:maxP}
Let $Q=\llbracket a_1,b_1 \rrbracket\times\cdots\times \llbracket a_d,b_d \rrbracket\subset \Z^d$ be a box in $\Z^d$ and let the inner boundary $\partial Q$ be defined as in \eqref{eq:bdry-inner}, 
and let $
    \partial^\circ Q\subset \partial Q
$
be the flat part of the boundary 
as defined in \eqref{eq:bdry-no-corner}. Let $V=\{v_n\}_{n\in Q}$ be a non-negative potential on $Q$. 
A vector $f=\{f_n\}_{n\in Q}$ is called a sub-solution, on (the interior of) $Q$ if 
\begin{equation*}
    -(\Delta f)_n+v_nf_n\ge 0,\ \ n\in Q\backslash \partial Q.
\end{equation*}
If $f$ is  a sub-solution, then the minimum of $f_n$ in $Q\backslash E(Q)$ must be attained  on $\partial^\circ Q$, i.e.,  
\begin{equation}\label{eq:maxP-apx}
    \min_{n\in Q\backslash \partial Q}\, f_n\, \ge  \min_{n\in \partial^\circ Q}\, f_n.
\end{equation}
\end{lemma}
\begin{proof}
Let $m=\min_{n\in \partial^\circ Q}\, f_n$. It is enough to prove that whenever $f_n\ge 0$ for all $n\in\partial^\circ Q$, we have $f_n\ge 0$ for all $n\in Q\backslash \partial Q$. Suppose not, then 
$
    -a:=\min_{n\in Q\backslash \partial Q}\, f_n<0.
$
Let $j\in Q\backslash \partial Q$ be such that the minimum is attained, i.e., $f_{j}=-a<0$ and $f_{j\pm e_i}\ge f_j,1\le i \le d$. Then  $-(\Delta f)_j+v_jf_j\ge 0$ implies that
$
	2df_{j}+  v_j f_j\ge \sum_{1\le i\le d} (f_{j+e_i}+f_{j-e_i})\ge 2d\cdot (-a)
$. 
	Therefore, $
	f_{j\pm e_i}= f_j=-a,1\le i \le d$, and  $v_j=0$. 	If any of $j\pm e_i$ belongs to the flat boundary $\partial ^\circ Q$, then it is a contradiction with the assumption that $f_n\ge 0$ for all $n\in\partial^\circ Q$. If not, then we pick any of them and repeat the procedure until eventually, after a finite number of steps, we reach the boundary and arrive at the contradiction again.
\end{proof}
There will be several direct corollaries of the above maximum principle. We will simply list them as independent lemmas and omit  the details for the proof. 
\begin{lemma}[Positivity of solutions for periodic boundary conditions]\label{lem:maxP-periodic}
Let $\Lambda=(\Z/K\Z)^d$. If $V=\{v_n\}_{n\in \Lambda}$ is a non-negative potential which is not constantly zero and  $(-\Delta f+V f)_n\ge 0$ for all $n\in \Lambda$, then $f_n\ge 0$ for all $n\in \Lambda$.
\end{lemma}
See \cite{WZ},  Lemma 2.12, for the proof.

\begin{lemma}[Maximum principle for discrete harmonic functions]\label{lem:maxP-harmonic}
Let $Q=\llbracket a_1,b_1 \rrbracket\times\cdots\times \llbracket a_d,b_d \rrbracket\subset \Z^d$ be a box in $\Z^d$ and let
$\partial Q, \partial^\circ Q$ be defined as in Lemma \ref{lem:maxP}.  Suppose  $f=\{f_n\}_{n\in Q}$ is a discrete harmonic function on (the interior of) $Q$, i.e., 
\begin{equation*}
    (\Delta f)_n=0,\ \ n\in Q\backslash \partial Q.
\end{equation*}
Then for all $n\in Q\backslash \partial Q$
\begin{equation}\label{eq:maxP-harmonic}
\min_{m\in \partial  Q}\, f_m \le \min_{m\in \partial^\circ Q}\, f_m \le   f_n\,\le \max_{m\in \partial^\circ Q}\, f_m \le \max_{m\in \partial  Q}\, f_m  .
\end{equation}
\end{lemma}
This is a direct application of Lemma \ref{lem:maxP}, to $f$ and $-f$. We only state maximum principles as above for the boxes in $\Z^d$ for simplicity, it is not hard to check that the same conclusion would hold  for more general domains in $\Z^d$, as long as they are ``connected'' with respect to  the discrete Laplacian operator in a suitable sense. In particular, it   works for the ``annular'' domain given by the difference of two cubes, $A=Q_2\setminus Q_1$. To be precise, if the boundary of $A$ is defined in the same as in \eqref{eq:bdry-inner}, i.e., $\partial A=\{n\in A:\, n+e_i\not\in A\ {\rm or }\ n-e_i\not\in A\ {\rm for\ some\ }e_i\} $ and  $(\Delta f)_n=0$ for $n\in A\backslash \partial A$, then
\[
    \min_{m\in \partial  A}\, f_m  \le   f_n\, \le \max_{m\in \partial  A}\, f_m, \quad \mbox{ for all } n\in A\backslash \partial A.
\]

\subsection{The discrete Poincar\'e inequality }\label{sec:Poincare}
The result essentially can be generalized to any ``connected'' region in $\Z^d$. We only need the version on a rectangular domain. 
\begin{lemma}\label{lem:PoincareZd}
Let $Q=I_1\times\cdots \times I_d$ be a a rectangular domain in $\Z^d$, where $I_i=\llbracket a_i,b_i \rrbracket$ for some $a_i<b_i\in\Z$ and $\ell_i=\cards{I_i}\in \N$, $i=1\cdots,d$. 
For any (real-valued) sequence $\{f_n\}_{n\in Q}$, let $|Q|=\cards{Q}$ and $\bar f_Q= \frac{1}{|Q|}\sum_{n\in Q}f_n$. Then 
\begin{align}\label{eq:PoincareZd}
    \sum_{n\in Q}(f_n-\bar f_Q)^2\le \frac{d}{2}\,\ell^2_{\max}\sum_{n\in Q}\|\nabla f_n\|^2.
\end{align}
\end{lemma}
\begin{proof}
 
Without loss of generality, we assume that $Q=\llbracket 1,\ell_1 \rrbracket \times \llbracket 1,\ell_2 \rrbracket \cdots \times \llbracket 1,\ell_d \rrbracket$.  
It is enough to prove \eqref{eq:PoincareZd} for $\bar f_Q=0$. It is easy to check that 
\begin{align}\label{eq:PWpf1}
    \sum_{m\in Q} \sum_{n\in Q}(f_m-f_n)^2=2|Q|\sum_{n\in Q}f_n^2.
\end{align}

For $m=(m_1,\cdots,m_d)\in Q, n=(n_1,\cdots,n_d)\in Q$, let $\Gamma(n,m)=\{\gamma^1\to\gamma^2\to\cdots \to\gamma^{d+1}\}$ be a discrete path in $\Z^d$ connecting $n$ and $m$, defined taking the maximal steps along every coordinate. That is, all  vertices $\gamma^i\in Q$, are given  by $\gamma^1=n$, $\gamma^{i+1}=\gamma^{i}+t_i e_i,i=1\cdots,d$, $\gamma^{d+1}=m$, and the all edges $E^i$ connecting the consecutive vertices are  parallel to $e_i$.
Then \begin{align}\label{eq:PWpf2}
    (f_m-f_n)^2=\Bigl(\sum_{1\le i\le d} (f_{\gamma^{i+1}}-f_{\gamma^{i}})\Bigr)^2
    \le     d     \sum_{1\le i\le d} (f_{\gamma^{i+1}}-f_{\gamma^{i}})^2.
\end{align}

We claim that for each $i=1,\cdots,d$, 
\begin{align}\label{eq:PWpf3}
    \sum_{m,n\in Q}  (f_{\gamma^{i+1}}-f_{\gamma^{i}})^2\le \ell_{\max}^2\, |Q| \, \sum_{n\in Q}|\nabla_i f_n|^2,
\end{align}
where $\ell_{\max}=\max_{j}\ell_j$. 
Then \eqref{eq:PoincareZd} follows from \eqref{eq:PWpf1}-\eqref{eq:PWpf3}.

 We now prove \eqref{eq:PWpf3} for $i=1$. Fix $\gamma^{d+1}=m=(m_1,\cdots,m_d)$. Write $\gamma^1=n=(n_1,\check n)$, where $\check n=(n_2,\cdots,n_d)\in \llbracket 1,\ell_2 \rrbracket \times \cdots \times \llbracket 1,\ell_d \rrbracket:=\check Q$. Assume that $t_1=m_1-n_1\ge 0$. Then $\gamma^2=\gamma^1+t_1e_1=(m_1,\check n)$. Write $f_k=f(k_1,\cdots,k_d)$ for $k=(k_1,\cdots,k_d)\in\Z^d$. Direct computation shows that
\begin{align*}
    (f_{\gamma^{2}}-f_{\gamma^{1}})^2=\left(\sum_{k_1=n_1}^{m_1-1} f(k_1+1,\check n) -f(k_1,\check n)\right)^2
\le & |t_1|\sum_{k_1=n_1}^{m_1-1} \big(f(k_1+1,\check n) -f(k_1,\check n)\big)^2  \\
\le &  \ell_{\max} \sum_{k_1=1}^{\ell_1} \big(\nabla_1 f(k_1,\check n) \big)^2. 
\end{align*}
The same estimate holds for $t_1<0$. 
Therefore, fix $m$, summing over $n\in Q$ gives
\begin{align*}
   \sum_{n\in Q} (f_{\gamma^{2}}-f_{\gamma^{1}})^2= \sum_{\check n \in \check Q}  \sum_{n_1=1}^{\ell_1}(f_{\gamma^{2}}-f_{\gamma^{1}})^2 
\le & \Big(\sum_{n_1=1}^{\ell_1}\ell_{\max}\Big) \Big( \sum_{\check n\in \check Q}  \sum_{k_1=1}^{\ell_1} \big(\nabla_1 f(k_1,\check n) \big)^2\Big) \\
\le & \ell_{\max}^2 \sum_{n\in Q} |\nabla_1f_n|^2. 
\end{align*}
Then summing over $m\in Q$ gives
\begin{equation*}
   \sum_{m,n\in Q} (f_{\gamma^{2}}-f_{\gamma^{1}})^2 \le |Q| \ell_{\max}^2 \sum_{n\in Q} |\nabla_1f_n|^2, 
\end{equation*}
which proves \eqref{eq:PWpf3} for $i=1$. The cases $i=2,\cdots,d$ can be proved in a similar manner. This completes the proof of \eqref{eq:PWpf3} and Lemma \ref{eq:PoincareZd}.

\end{proof}

\subsection{Discrete cut-off functions}\label{sec:cutoff}
Let $Q$ be a cube of side length \ $R\ge3$ on $\Z^d$  and let  $j_{\max}=\fl{R/3}$.  
Let $\partial Q$ and $Q/3$ be given by \eqref{eq:bdry-inner} and \eqref{eq:Q/3}.  
Let the distance ${\rm dist}(n,m)=|n-m|_\infty$ be measured by the infinity norm on $\Z^d$. 
Let $\cI(j),1\le j\le j_{\max}$, be a $d-1$ dimensional subset of $Q$ which is distance $j$ away from $Q/3$:
\[\cI(j):=\left\{n\in Q  |  \dist(n,Q/3)=j\ \right\}.\]

By the definition of $Q/3$ in \eqref{eq:Q/3}, the side length of $Q/3$  satisfies $\ell(Q/3)=\fl{R/3}\le R/3.$  
It is easy to check that $Q/3$ and all $\cI(j)$ are pairwise  disjoint for $j=1,\cdots, j_{\max}$. And $ 3(Q/3)=Q/3 \bigcup \Bigl(\bigcup_{j=1}^{j_{\max}}\cI(j)\Bigr) \subset Q
$. 
 
Now we can define the cut-off function $\chi=\{\chi_n\}$ as 
\begin{align}\label{eq:cutoff}
\chi_n = 
\begin{cases}
1, & n\in Q/3,\\
1-\frac{3}{R}j,&  n\in \cI(j),\ j=1,\cdots, j_{\max}, \\
0, &  n\notin 3(Q/3).
\end{cases}
\end{align} 
It is easy to see that  $|\chi_{n+e_i}-\chi_{n}|\le  {3}/{R}$ if $n,n+e_i\in Q$, and $\chi_{n+e_i}-\chi_{n}=0$ otherwise, for all $1\le i \le d$.

\subsection{Dirichlet problem on a cube}\label{sec:poisson}
We study the Dirichlet problem for the discrete Laplacian on a cube in $\Z^d$. 
Recall the definitions of $Q(r;\xi)$, $\partial Q(r)$, $ \partial^\circ Q(r)$, for  $\xi=(\xi_1,\cdots,\xi_d)\in\Z^d$, and $r \in \Z_{\ge0}$, given right before Lemma~\ref{lem:arAr}.

\begin{lemma}[Green's formula]\label{lem:G0}
For any $f,g\in \cH=\ell^2({\Lambda}) $, 
\begin{align}\label{eq:G0}
  \sum_{n\in Q(r-1)}g_n (\Delta f)_n =&
  -
  \sum_{n,n+e_i\in Q(r) }(\nabla g)_n (\nabla f)_n
  + \sum_{{n\in \partial^\circ Q(r) } }g_n \frac{\partial f}{\partial {\bf  N}}(n)\\
 =& -
  \frac{1}{2}\sum_{\substack{n,m\in Q(r)\\ |m-n|=1 } }(g_m-g_n)(f_m-f_n)
  + \sum_{\substack{n\in \partial^\circ Q(r), \,m\in \partial Q(r-1):\\ |m-n|=1 } }g_n (f_n-f_m).
\end{align}
As a consequence,
\begin{align}
  \sum_{n\in Q(r-1)}g_n  (\Delta f)_n -&\sum_{n\in Q(r-1)}f_n  (\Delta g)_n \nonumber \\
  =&-\sum_{\substack{n\in \partial^\circ Q(r), m\in \partial Q(r-1)\\ |m-n|=1 } }f_n  (g_n-g_m)
  + \sum_{\substack{n\in \partial^\circ Q(r), m\in \partial Q(r-1)\\ |m-n|=1 } }g_n (f_n-f_m)\label{eq:G1}.
\end{align}
\end{lemma}

The Green’s formula for discrete graphs is rather standard, existing in various lecture notes, e.g. Theorem
	1.37 in \cite{Ba}, see also in \cite{Chung,Gu}. We omit the proof here.

Given $\{f_n\}_{n\in Q(r-1)}$ and $\{h_n\}_{n\in \partial Q(r)}$, we proceed to solve the linear system on $Q(r)$
\begin{align}\label{eq:general}
    \begin{cases}
    -(\Delta u)_n=f_n,\ \  n\in Q(r-1), \\
    u_n=h_n,\  n\in \partial Q(r).
    \end{cases}
\end{align}

The problem can be decomposed into the following two systems, which give us the discrete Poisson Kernel and the discrete Green's function for the Dirichlet Laplacian. 

The discrete Poisson kernel $P_r(n;m): Q(r)\times \partial Q(r)\to [0,1]$ is the unique solution to the system
\begin{align}\label{eq:Pdef}
\begin{cases}
&\Delta P_r(n,m)=0,\,   n \in Q(r-1),\\
    & P_r(n,m)=\delta_{m}(n), \  n\in  \partial Q(r),
\end{cases}
\end{align}
for a fixed $m\in \partial Q(r)$.
Similarly, the discrete Green's function with pole at $m$, $G_r(n,m): Q(r)\to [0,1]$ is the unique solution to the system
\begin{align}\label{eq:Gdef}
\begin{cases}
-\Delta G_r(n,m)=\delta_{m}, \ & n\in Q(r-1),\\
     G_r(n,m)=0, \ & n\in\partial Q(r),
\end{cases}
\end{align}
for a fixed $m\in Q(r-1)$,  
Consider $-\Delta$ with zero boundary condition as an invertible matrix of the size $|Q(r-1)|\times |Q(r-1)|$. Clearly, for $n,m\in Q(r-1)$, $G(n,m)=(\Delta)^{-1}(n,m)=G(m,n)$ since $\Delta$ is self-adjoint. 

Moreover, for fixed $m\in \partial Q(r)$ and $m'\in Q(r)$, if we apply the Green's formula \eqref{eq:G1} to $g_n=P_r(n,m)$ and $f_n=G_r(n,m')$, then 
\begin{equation*}
    \sum_{n\in Q(r-1)}P_r(n,m)  \big(-\delta_{m'}(n)\big) =\sum_{\substack{n\in \partial^\circ Q(r),\, n'\in \partial Q(r-1):\\ |n'-n|=1 } }\delta_m(n) \,\big(G_r(n,m')-G_r(n',m')\big),
\end{equation*}
which implies for any $m\in \partial Q(r)$ and $m'\in Q(r)$,
\begin{equation}\label{eq:PG}
   P_r(m',m)  =G_r(n',m')=G_r(m',n'), n'\in \partial Q(r-1), |n'-m|=1.
\end{equation}
 
Notice that this can be considered as the (negative) normal derivative of $G_r(m',\cdot)$ in the direction of the outward pointing to the surface of $Q(r)$.

Back to the system \eqref{eq:general}, using $P_r$ and $G_r$, we can solve the system
\begin{equation*}
    u_n=\sum_{m\in\partial Q(r)}P_r(n,m)h_m+\sum_{m'\in  Q(r-1)}G_r(n,m')f_{m'}.
\end{equation*}

In particular, we have the following integration by parts formula (Green's identity) for any $\{u_n\}_{n\in \partial^\circ Q(r)}$ at the center of the box $Q(r;\xi)$: 
\begin{equation}\label{eq:IBP}
    u_\xi=\sum_{m\in\partial Q(r)}P_r(\xi,m)u_m-\sum_{m'\in  Q(r-1)}G_r(\xi,m') (\Delta u)_{m'}.
\end{equation}
In particular, for any $\xi$ and $r$,
\begin{equation}\label{eq:Psum1}
  \sum_{m\in\partial Q(r)}P_r(\xi,m)=1.
\end{equation}


\subsection{Estimates on the Green's function and the Poisson kernel}
Retain the definitions in the previous section, for $R\in \N$, let $Q(R)=Q(R;\xi)\subset \Z^d$ be the discrete cube centered at $\xi\in\Z^d$ of side length $2R+1$, and let $G_R(\xi,n)$ be the discrete Green's function as defined in \eqref{eq:Gdef}. In this part, we study the behavior of the discrete Green's function away both from the pole $\xi$ and the boundary $\partial Q(R)$.  We will approximate the discrete Green's function by a continuous one to obtain the desired estimates.  Let us also recall some of the definitions for the continuous case. Fix $\xi\in \Z^d$, let ${\mathcal Q}_1=\xi+[-1,1]^d$ be a cube in $\R^d$ centered at $\xi$ of side length $2$. Let $\cG(\xi,\cdot)$ be the continuous Green's function on the cube $\cQ_1$ with zero boundary conditions:
\begin{align}\label{eq:G-conti}
    \begin{cases}
    -\Delta_{\rm c}\cG(\xi,x)=\delta^c_{\xi}(x),\\
    \cG(\xi,x)=0, x\in \partial {\mathcal Q}_1,
    \end{cases}
\end{align}
where $\Delta_c=\sum_{i=1}^d\frac{\partial^2}{\partial x_i^2}$ is the standard Laplacian on $\R^d$, and $\delta^c_{\xi}(x)$ is the Dirac delta function at $\xi$ in the distribution sense. 
 For any $R\in\N$, consider a square mesh of size $h=\frac{1}{R}$ on $\cQ_1$. Denote the collection of all the mesh points by
 \begin{equation}\label{eq:mesh}
     \Omega_h=\big\{\tau_n=\xi+n  h, n\in \llbracket -R,R \rrbracket^d\, \big\}.
 \end{equation} 
We see that $\Omega_h$ is indexed (one-to-one) by $n\in Q(R)$, and hence $\ell^2(\Omega_h)\cong\R^{(2R+1)^2}$.  
For $\tau_n\in \Omega_h$, let
\begin{equation}\label{eq:G-Gh}
  \cG^h(\xi,\tau_n):= \frac{1}{h^{d-2}} G_R(\xi,n)=\frac{1}{h^{d-2}} G_R(n,\xi).
\end{equation}
 It is easy to verify that the equation for $G_R(n,\xi)$ in \eqref{eq:Gdef} implies that 
 \begin{equation}\label{eq:Gh-eq}
\begin{cases}
-(\Delta  \cG^h)(\xi,\tau_n):=- \sum_{|n-m|=1}\big(\cG^h(\xi,\tau_m)- \cG^h(\xi,\tau_n)\big)={h^{2-d}}\delta_\xi(n), \quad  \ n\in Q(R;\xi), \\
 \cG^h(\xi,\tau_n)=0,  \quad n \in  \partial Q(R;\xi).
\end{cases}
\end{equation}
We see that $\cG^h(\xi,\cdot)\in \ell^2(\Omega_h)$ is the  finite difference approximation to the solution of the continuous problem \eqref{eq:G-conti}. The approximation can be quantified as follows.
\begin{lemma}\label{lem:A6}
There are positive dimensional constants $C,h_0$ such that if $h\le h_0$, then  
\begin{align}\label{eq:G-Gh2}
    |\cG^h(\xi,\tau_n)-\cG(\xi,\tau_n)|\le \, C |\log  h|^{d+3}h^2 \ \ \ {\textrm{for all}}\ \tau_n \in  \Omega_h\cap\left({\mathcal Q}_{1} \backslash {\mathcal Q}_{1/2}\right). 
\end{align}
\end{lemma}

\begin{remark}
Such an approximation is proved in a rectangular domain in $\R^2$ in \cite{La}. It was later generalized to the interior of a domain of any dimension with smooth boundary by Schatz and Wahlbin, see Theorem 6.1 \cite{SW}, using the finite elements approach. The method in \cite{SW} potentially can be generalized to any convex polyhedral domains, up to the boundary. Here we present a direct proof using the series expansion of $\cG^h$ and $\cG$. Similar estimate also holds  for $\cG^h(y,x)-\cG(y,x)$ where the pole $y$ is not far away from the center $\xi$. We will only deal with the case $y=\xi$ which will be enough for our use.
\end{remark}

\begin{proof}
 
Without loss of generality we assume that $\xi= 0$. In this case, the mesh points $\tau_n=nh$ and $\Omega_h=\frac{1}{R}Q(R)\subset \cQ_1$, where $h=1/R$. Due to the symmetry of the problem, it is also enough to prove \eqref{eq:G-Gh2} on the upper half cube $\cQ^+_1=\{x\in \cQ_1: x_d\ge 0\}$. Below, we construct the analytic series representations of $\cG(0,x)$ and $\cG^h(0,x)$ on $\cQ^+_1$. 

For the continuous case, the analytic expression of $\cG(0,x)$ is well known by the method of the partial eigenfunction representation. 
Throughout the rest of the proof, we denote  $k=(k_1,\cdots,k_{d-1})\in \Z_+^{d-1}$,  and $x=(\wt x, x_d)$ where $\wt x=(x_1,\cdots,x_{d-1})\in \R^{d-1}$.  Let 
\begin{equation}\label{eq:alphak}
f_k(\wt x)=\prod_{i=1}^{d-1}\sin\Big(\frac{k_i \pi}{2}(x_i+1)\Big),  \ {\rm and}\    \alpha_k=\frac{\pi}{2}\|k\|:=\frac{\pi}{2}\Big(\sum_{i=1}^{d-1}k_i^{2}\Big)^{1/2}.
\end{equation}   
By separation of variables, for $x=(\wt x, x_d)\in \cQ^+_1,0<x_d<1$, 
\begin{align}
    \cG(0,x)=&\sum_{k\in\Z_+^{d-1}} \frac{1}{\alpha_k\sinh(2\alpha_k)}\sinh(\alpha_k) \sinh\big(\alpha_k(1-x_d)\big) f_k(0) f_k(\wt x)   \nonumber \\
    =&\sum_{k\in\Z_+^{d-1}} \frac{\sinh(\alpha_k(1-x_d))}{2\alpha_k\cosh(\alpha_k)}\,\prod_{i=1}^{d-1}\sin\bigl(\frac{k_i \pi}{2}\bigr)\sin\bigl(\frac{k_i \pi}{2}(x_i+1)\bigr) \label{eq:G-coeff}.
\end{align}

The partial eigenfunction representation  can be used to derive a similar formula for $\cG^h(0,\tau_n)$ solving \eqref{eq:Gh-eq} on the finite dimensional space. We may abuse the notation and write $\cG^h(0,n)=\cG^h(0,\tau_n)=\cG^h(0,nh)$ when it is clear. Notice that $\cG^h(0,n)$ satisfies zero boundary condition on $\partial Q(R)$, it is enough to consider $\cG^h(0,n)$ as a discrete function only for $n\in Q(R-1)=\llbracket -R+1,R-1 \rrbracket^{d}$.  Similar to the notations for the continuous case, we write $n=(\wt n, n_d)$, where $\wt n=(n_1,\cdots,n_{d-1})\in \llbracket -R+1,R-1 \rrbracket^{d-1}$, and $n_d\in \llbracket -R+1,R-1 \rrbracket$. Denote by $T_R:=\llbracket 1,2R-1 \rrbracket^{d-1}$. We first construct a basis for the subspace $\ell^2(\llbracket -R+1,R-1 \rrbracket^{d-1})$.  
For $k=(k_1,\cdots,k_d)\in T_R$ and $\wt n=(n_1,\cdots,n_{d-1})\in \llbracket -R+1,R-1 \rrbracket^{d-1}$, let 
\begin{equation}\label{eq:fhk} f^h_k(\wt n)={\sqrt h}^{\,d-1}\prod_{i=1}^{d-1}\sin\bigl(\frac{k_i \pi}{2}(n_ih+1)\bigr). 
\end{equation}

We claim that $\{f^h_k\}_{k\in T_R}$ form a normalized basis for $ \R^{(2R-1)^{d-1}}$.  This can be verified by direct computations of the finite dimensional inner product. For $n_i,m_i\in \Z$, let $\theta= { \pi h(n_i-m_i)}/{2}, \varphi= { \pi (hn_i+hm_i+2)}/{2}$. 

If $n_i=m_i$, then 
\begin{multline}\label{eq:sin}
2 \sum_{k_i=1}^{2R-1}\sin\frac{k_i \pi(n_ih+1)}{2}\sin\frac{k_i \pi(m_ih+1)}{2}   =\sum_{k_i=1}^{2R-1}\cos(k_i\theta)- \sum_{k_i=1}^{2R-1}\cos(k_i\varphi) \\
   =   (2R-1)+\frac{1}{2}-\frac{\sin(2R-1/2)\varphi}{2\sin\frac{\varphi}{2}}  
   =2R.
\end{multline}

If $n_i\neq m_i$,  then $2R\theta=\pi(n_i-m_i), 2R\varphi=\pi(n_i+m_i)+2\pi R \in \pi\Z$. Hence, 
\begin{multline*}
 2 \sum_{k_i=1}^{2R-1}\sin\frac{k_i \pi(n_ih+1)}{2}\sin\frac{k_i \pi(m_ih+1)}{2}\\ = \frac{\sin(2R-1/2)\theta}{2\sin\frac{\theta}{2}}-\frac{\sin(2R-1/2)\varphi}{2\sin\frac{\varphi}{2}} =-\sin(\pi n_i)\sin(\pi m_i)=0.
\end{multline*}
Therefore, 
\[\sum_{k_i=1}^{2R-1}\sqrt h\sin\frac{k_i \pi(n_ih+1)}{2} \,  \sqrt h\sin\frac{k_i \pi(m_ih+1)}{2}= \delta_0(n_i-m_i),\]
which implies for $\wt n,\wt m\in \llbracket -R+1,R-1 \rrbracket^{d-1}$, 
\[\sum_{k\in T_R}f^h_k(\wt n)f^h_k(\wt m)=\prod_{i=1}^{d-1}\delta_0(n_i-m_i)=\delta_0(\wt n-\wt m). \]
This shows that the $(2R-1)^{d-1} \times (2R-1)^{d-1}$ dimensional matrix $U(k,\wt n):=f^h_k(\wt n)$ is unitary. Hence, its column vectors $\{f^h_k\}_{k\in T_R}$ form a normalized basis. 

Now we fix $n_d$ and expand $\cG^h(0,nh)$ with respect to the normalized basis  $\{f^h_k\}_{k\in T_R}$:
\begin{equation}\label{eq:Gh-exp}
\cG^h(0,n)=\cG^h(0,nh)= \sum_{k\in T_R} F_k(n_d)   f^h_k(\wt n),
\end{equation}
where $ F_k(n_d)$ is the coefficient function to be solved. We also expand the $d$ dimensional discrete delta function as  $\delta_{ 0}(\wt n,n_d)=\sum_k\delta_0( n_d) f_k^h(0)  f_k^h(\wt n)$.

 Write the discrete Laplacian as $\Delta=\wt{\Delta}+\Delta^d$, where $\wt{\Delta}$ is the second order difference with respect to  the first $d-1$ variables $\wt n=(n_1,\cdots,n_{d-1})$ and $\Delta^d$ is the second order difference with respect to  the  last variable $n_d$. Applying $\wt{\Delta}$ to $f^h_k(\wt n)$ gives
\[
    \wt{\Delta}f^h_k(\wt n)=f^h_k(\wt n)\,\sum_{i=1}^{d-1}\Bigl(2\cos\frac{k_i\pi h}{2}-2\Bigr).
\]
Combing with the expansion \eqref{eq:Gh-exp}, we obtain
\begin{multline*}
\Delta \cG^h(0,n)= \sum_{k\in T_R} F_k(n_d)   \Big(\wt{\Delta}f^h_k(\wt n)\Big) + \Big(\Delta^d F_k(n_d)\Big) f^h_k(\wt n)  \\
=\sum_{k\in T_R} \Big(F_k(n_d)  \sum_{i=1}^{d-1}\bigl(2\cos\frac{k_i\pi h}{2}-2\bigr) +  (\Delta^d F_k)(n_d) \Big) f^h_k(\wt n) .
\end{multline*}  
Hence, the $d$ dimensional equation $\Delta \cG^h(0,\cdot)=-h^{2-d}\delta_0(\cdot)$ can be reduced to a one dimensional difference equation of $F_k(n_d)$:
\begin{equation}\label{eq:Fk}
  F_k(n_d)\sum_{i=1}^{d-1}\Bigl(2\cos\frac{k_i\pi h}{2}-2\Bigr)+(  \Delta^d F_k)(n_d)=-h^{2-d} f_k^h(0) \delta_0(n_d).
\end{equation}

Define $\beta_k=\beta_k(h)$ to be the positive solution solving
\begin{equation}\label{eq:betak}
2\cosh \beta_k-2+\sum_{i=1}^{d-1}\left(2\cos\frac{k_i\pi h}{2}-2\right)=0 . 
\end{equation}
 
For $-R\le n_d\le R$, let  $v_1(n_d):=\sinh(\beta_k (R-n_d))$ and $v_2(n_d):=\sinh(\beta_k (R+n_d))$. Then $(\Delta^d v_j)(n_d)=v_j(n_d)(2\cosh \beta_k-2), j=1,2$. Hence,  $v_j,j=1,2$ solves the homogeneous part of \eqref{eq:Fk}:
\begin{equation}\label{eq:homo-v12}
v_j(n_d)\sum_{i=1}^{d-1}\Bigl(2\cos\frac{k_i\pi h}{2}-2\Bigr)+(  \Delta^d v_j)(n_d)=0,\ \ {\rm for\ all}\ -R+1\le n_d\le R-1,
\end{equation}
 with boundary conditions $v_1(R)=0=v_2(-R)$. Define  $H(n_d)=v_1(n_d)$ for $0\le n_d\le R$, and $H(n_d)=v_2(n_d)$ for $-R\le n_d< 0$. 
Then $H(n_d)$ satisfies \eqref{eq:homo-v12} for all $-R+1\le n_d\le R-1$ except $n_d=0$.  At $n_d=0$,  the definition of $H,v_1$ and $v_2$ implies 
\[
    H(0)\left(2-2\cosh \beta_k\right)+H(1)+H(-1)-2H(0)=-2\cosh(\beta_kR)\sinh(\beta_k).
\]
Finally, let 
$
F_k(n_d)=\frac{h^{2-d} f_k^h(0)}{2\cosh(\beta_kR)\sinh(\beta_k)}H(n_d)
$. Then $F_k(n_d)$
solves the inhomogeneous equation \eqref{eq:Fk} for all $-R+1\le n_d\le R-1$, with zero boundary condition $F_k(\pm R)=0$. 

Together with \eqref{eq:fhk} and \eqref{eq:Gh-exp}, we obtain the analytic series expansion of $ \cG^h(0,nh)$,  on the upper half cube $0\le n_d\le R$:
\begin{align}
    \cG^h(0,nh)=&\sum_{k\in T_R} F_k(n_d)   f^h_k(\wt n)  \nonumber\\
    =&\sum_{k\in T_R} \frac{ \sinh\big(\beta_k(R-n_d)\big)}{2R\sinh(\beta_k) \cosh(\beta_kR)} \left(\prod_{i=1}^{d-1}\sin\big(\frac{k_i \pi}{2}\big)\sin\big(\frac{k_i \pi}{2}(n_ih+1)\big)\right) \label{eq:Gh-coeff}.
\end{align}

It was proved in \cite{La} that for $d=2$, and $|n|>0$, one has $|\cG(0,nh)-\cG^h(0,nh)|\le Ch^2|n|^{-2}$. The method can be extended to  higher dimensions. We will not bother to give the full generalization of the exact singularity of  order $1/|n|^2$. We only need the version for $R\ge n_d\ge R/2$ with some logarithm corrections as in \eqref{eq:G-Gh2}. To do that,  it is enough to study the asymptotic behavior of $\beta_k(h)$ in \eqref{eq:betak} as $h\to 0$.  

Let $\beta_k$ be given by \eqref{eq:betak}. First, it was shown in \cite{GuMa} (see also in \cite{Gu}) that 
either $\beta_k\ge 1$ or $2R \beta_k\ge \|k\|_d$ for all $R=\frac{1}{h}$ and $ 1\le k_i\le 2R-1$. We sketch the proof for reader's convenience. If $\beta_k\le 1$, then 
\begin{equation}\label{eq:014}
   1+\beta_k^2\ge \cosh \beta_k  = d-\sum_{i=1}^{d-1}\cos\frac{\pi k_ih}{2}\ge d-\sum_{i=1}^{d-1}\left(1- \bigl(\frac{ k_ih}{2}\bigr)^2\right)= 1+\frac{\|k\|^2h^2}{4},
\end{equation}
which gives $\beta_k \ge \|k\|h/2$. In \eqref{eq:014},  we used the elementary inequalities $1+x^2 \ge\cosh x$, for $x\in [0,1]$ and $1-\cos \pi x \ge x^2$ for $x\in [0,1)$. Therefore, for any $x_d=n_d/R\in[1/2,1]$, the coefficients of $\cG^h$ decays exponentially 
\begin{equation*}
     \left|\frac{ \sinh(\beta_k(R-n_d))}{2R\sinh(\beta_k) \cosh(\beta_kR)}\right|\le \frac{1}{\|k\|}e^{-\|k\|/4},\ \ {\rm or}  \  \left|\frac{ \sinh(\beta_k(R-n_d))}{2R\sinh(\beta_k)\cdot\cosh(\beta_kR)}\right|\le \frac{1}{\|k\|}e^{-R/2}.
\end{equation*}

On the other hand, for any fixed $k\in T_R$, we want to expand $\beta_k(h)$ in $h$ explicitly. Let $\alpha_k= {\pi\|k\|}/{2}$ be as in \eqref{eq:alphak}, one has 
\begin{align*}
    \beta_k= \cosh^{-1}\left(d-\sum_{i=1}^{d-1}\cos\frac{\pi k_ih}{2}\right) 
    =&\cosh^{-1}\left(1+\frac{\pi^2\|k\|^2h^2}{8}+O(\|k\|^4h^4)\right) \\
    =&\log \left(1+\frac{\pi^2\|k\|^2h^2}{8}+{\frac{\pi \|k\|h }{2}+O(\|k\|^3h^3)}\right)\\
    =& \frac{\pi \|k\|h}{2}+O(\|k\|^3h^3)=\alpha_k  h+O(\|k\|^3h^3).
    \end{align*}
Therefore, $R \, \beta_k=\alpha_k+O(|\log  h|^3h^2)$ for $\|k\|\le C_d|\log  h|$ (with any constant $C_d$ only depending on the dimension). For any $x_d=n_d/R\in[1/2,1]$ and $t=1-x_d\in[0,1/2]$,  we compare the coefficients of $\cG$ and $\cG^h$ in \eqref{eq:G-coeff} and \eqref{eq:Gh-coeff} up to  $\|k\|\le C_d|\log  h|$.  Let $f(x)=\frac{\sinh(xt)}{\cosh(x)}$. Then $0\le f(x)\le 1, |f'(x)|\le 2$ for any $x\ge 0$. Therefore, $f(\beta_kR)=f(\alpha_k)+O(\|\log  h\|^3h^2)$. Notice that $1/({R\sinh(xh)})=\frac{1}{x}(1+O(x^2h^2)$ implies  
\begin{align*}
    \frac{1}{{R\sinh(\beta_k)}}=\frac{1}{\beta_kR}\bigl(1+O(\beta_k^2h^2)\bigr)=\frac{1}{\alpha_k}\bigl(1+O(|\log  h|^3h^2)\bigr).
\end{align*}
Putting all these together, we have  
\begin{align*}
    \frac{ \sinh(\beta_k(R-n_d))}{2R\sinh(\beta_k) \cosh(\beta_kR)}=\frac{\sinh(\alpha_k(1-x_d))}{2\alpha_k\cosh(\alpha_k)}+O(|\log  h|^3h^2).
\end{align*}

Then we split the series expression of $\cG(0,nh)-\cG^h(0,nh)$ into low frequency and high frequency part for   $a= 8|\log  h| = 8\log  R$,
\begin{align*}
   | \cG(0,nh)-\cG^h(0,nh)|\le& \, \sum_{\|k\|\le a}\left(\cdot\right)+\sum_{\|k\|>  a}\left(\cdot\right) \\
   \le&\,  \sum_{\|k\|\le a}O(|\log  h|^3h^2)+\sum_{\ell\ge a}\sum_{\ell \le \|k\|<\ell+1  }\frac{1}{\ell}e^{-\ell/4}\, +(2R-1)^{d-1}e^{-R}\\
   \le & \, a^dO(|\log  h|^3h^2)+a^{d-2}e^{-a/4}\, +(2R-1)^{d-1}e^{-R}\le C_d|\log  h|^{d+3}h^2,
\end{align*}
for sufficiently large $R\ge R_d$.  This completes the proof of \eqref{eq:G-Gh2}. 
\end{proof}

For any $0<\eps<1/4$, by the positivity and smoothness (away from the pole) of the the continuous Green's function $\cG$, there are $c_1(\varepsilon,d)>0,c_2(\eps,d)>0$ such that $2c_1\le \cG(\xi,x)\le c_2/2$ for all ${\mathcal Q}_{1-\varepsilon/4} \backslash {\mathcal Q}_{1/2}$. Combining this with the approximation in \eqref{eq:G-Gh2}, we have $ c_1\le \cG^h(\xi,\tau_n)\le c_2$ for $h<\wt h_0(\varepsilon,d)$ and $\tau_n=\xi+nh\in \Omega_h\cap\left({\mathcal Q}_{1-\varepsilon/4} \backslash {\mathcal Q}_{1/2}\right)$. 
Then by \eqref{eq:G-Gh}, one has $  c_1\le  R^{d-2}G_R(\xi,n)\le c_2$ for $R=1/{h}\ge 1/\wt h_0 $. 
Notice that $\tau_n=\xi+nh\in \Omega_h\cap\left({\mathcal Q}_{1-\varepsilon/4} \backslash {\mathcal Q}_{1/2}\right)$ is equivalent to 
\begin{align}\label{eq:tmp-apx1}
    n\in Q(R;\xi), \ \ {\rm and} \ R/2\le |n-\xi|_\infty<(1-\eps/4)R.
\end{align}

For any $0<\eps<1/4$ and $r\ge   15/\eps$, if we set $R=\fl{(1+\varepsilon)^{1/2}r}$, then it is easy to verify that $(1+\eps/3)r\le R \le (1+\eps/2)r$, which implies that 
$
 R/2\le    r \le (1+\eps/3)^{-1}R<(1-\eps/4) R.
$
In other words, if $n\in \partial Q(r;\xi)$, then $n$ will satisfy \eqref{eq:tmp-apx1}. In conclusion, we have obtained
\begin{lemma}\label{lem:GRr}
For any $0<\varepsilon<1/4$, and $r\in \N$, let $R=\fl{(1+\varepsilon)^{1/2}r}$. There are constants $c_1,c_2 ,r_0$ depending on $d$ and $\eps$ such that if $r\ge r_0$, then $
 c_1 r^{2-d} \le    G_R(\xi,m)\le c_2r^{2-d} 
$ for all $m\in Q(r;\xi)$
\end{lemma}

We are also interested in the behavior of $\cG$ and $G_r$ near the boundary.
\begin{lemma}\label{lem:A8}
Let $x^{\,c}=\xi+(0,\cdots,0,1)\in \partial \cQ_1$ be the center of the top surface of $\cQ_1$. For any $0<\eta<1/4$, let $T_{1-\eta}$ be the semi cube contained in $\cQ_1$, cetered at $x^c$, and away from the other surfaces of $\cQ_1$ by distance $\eta$:
\[
    T_{1-\eta}=\{x\in \cQ_1: \max_{1\le i\le d}|x_i-x_i^{\,c}|\le 1-\eta\}.
\]
There is a constant $c(\eta,d)$ such that for all $x=(x_1,\cdots,x_d)\in T_{1-\eta}$ , one has 
\begin{align}
    \cG(\xi,x)\ge c(\eta,d)|\xi_d+1-x_d|=c(\eta,d)\, {\rm dist}(x,\partial \cQ_1) \label{eq:cG-bdry}.
\end{align}
\end{lemma}
\begin{proof}
Consider a   larger semi-cube $T_{1-\eta/2}$ such that $
    T_{1-\eta}\subsetneq T_{1-\eta/2}\subsetneq T_{1}.
$
Let 
$
    h(x)=\xi_d+1-x_d.
$
Clearly, $\cG(\xi,x)$ and $h(x)$ are two strictly positive harmonic functions on the interior of $T_{1-\eta/2}$. By the comparison principle for harmonic functions, see, e.g. \cite{Da,Ke}, there is a constant $c$ only depends on $d$ and $\eta$ such that $
 {\cG(\xi,x)}/{h(x)} \ge c   {\cG(\xi,y)}/{h(y)}
$ for all $x,y \in T_{1-\eta}$. Take $y=(0,\cdots,0,\xi_d+2\eta)\in T_{1-\eta}$ so that $h(y)=1-2\eta$. Then  $ \cG(\xi,y)\ge C(\eta,d)$, and therefore, 
$
    {\cG(\xi,x)}\ge c    h(x)=c  \left(\xi_d+1-x_d\right),
$
where  $c$ only depends on $d$ and $\eta$. 
\end{proof}
 Combing Lemma \ref{lem:A6} and Lemma \ref{lem:A8}, we have
\begin{lemma}\label{lem:poisson}
Let $Q(r)=Q(r;\xi)$ be a cube in $\Z^d$ centered at $\xi \in \Z^d$, with side length $r$. Let $ P_r(\xi,m)$ be the discrete Poisson kernel on $Q(r)$ given by \eqref{eq:Pdef} with pole at the center $\xi$. Let $S\subset \partial Q(r)$ be the ``top surface'' of $Q(r)$:
\[S:=\big\{m=(m_1,\cdots,m_d)\in \partial Q(r):\ m_d=\xi_d+r\big\}.\]
 Let $n^c=(n^c_1,n_2^c,\cdots,\xi_d+r)\in S$ be the center of $S$. 
Suppose $0<\eta <1/4$. There are constants $c=c(\eta,d)$ and $r_0(\eta,d)$  depending only on the dimension and $\eta$ such that
\begin{align}\label{eq:Pr-1}
    P_r(\xi,m)\ge c \frac{1}{r^{d-1}}.
\end{align}
 for all $r\ge r_0$ and $m=(m_1,\cdots,m_d)\in S$ satisfying
$
    |m-n^c|_\infty\le (1-\eta)r. 
$ 
\end{lemma}
The same estimate holds for $P_r(\xi,m)$ on the other $2d-1$ surfaces $S^i=\big\{m\in \partial Q(r):\ m_i=\xi_i\pm r\big\},i=1\cdots,d-1$ when $m$ is $\eta$-away from the edges and the corners. 
\begin{proof}
It is enough to prove the result for $Q(r)=Q(r; 0)$ centered at $\xi=  0$. 
We consider the approximation \eqref{eq:G-Gh2} on $\cQ_1=[-1,1]^d$ with the mesh size $h=\frac{1}{r}\le h_0$. Retain the definitions of $\cG,\cG^h,\Omega_h$ and $\tau_n$ in Lemma \ref{lem:A6}. By the definition of the mesh $\Omega_h$,  for $n=(n_1,\cdots,n_{d-1},r-1)\in \partial Q(r-1) $, the mesh points  $\tau_n=n\frac{1}{r}\in \Omega_h\cap \left(\cQ_1\backslash\cQ_{1/2}\right)$. Then Lemma \ref{lem:A6} implies that 
\[
   | \cG(0,\tau_n)-\cG^{1/r}(0,\tau_n)|\le C\frac{(\log  r)^{d+3}}{r^2},
\]
for   $n=(n_1,\cdots,n_{d-1},r-1)\in \partial Q(r-1) $. 

For these  $n=(n_1,\cdots,n_{d-1},r-1)\in \partial Q(r-1)$, assume further that $|n_i|<(1-\eta)r,1\le i\le d-1$. Then for $\tau_n= n/r$ and $x^c= (0,\cdots,0,1)$, one has 
\[\tau_n-x^c=\big(\frac{n_1}{r},\cdots,\frac{n_{d-1}}{r}, \frac{r-1}{r}\big).\]
Then 
$\tau_n\in T_{1-\eta}$ provided that  $r\le 1/\eta$,  and hence, \eqref{eq:cG-bdry} implies that 
$
    \cG(0,\tau_n)\ge c  {\rm dist}(\tau_n,\partial \cQ_1)=c  \frac{1}{r}.
$
Therefore, for $r\ge r_\ast(d,\eta)$
\[
    \cG^{1/r}(0,\tau_n)\ge \cG(0,\tau_n)-C\frac{(\log  r)^{d+3}}{r^2} \ge   c \frac{1}{r}-C\frac{(\log  r)^{d+3}}{r^2}\ge  \wt c\,\frac{1}{r}.
\]
By \eqref{eq:G-Gh} (for $h=1/r$), we have  
\begin{equation}\label{eq:tmpff}
    G_r(0,n)=h^{d-2}\cG^{1/r}(0,\tau_n)\ge  \wt c\frac{1}{r^{d-1}},
\end{equation}
for  $n=(n_1,\cdots,n_{d-1},r-1)\in \partial Q(r-1) $ and $|n_i|<(1-\eta)r,i=1,\cdots,d-1$. 

Notice that if $m\in S\subset \partial Q(r)$ and $n\in \partial Q(r-1)$ with $|n-m|_1=1$, then $|n_i|<(1-\eta)r,i=1,\cdots,d-1$, and $n_d=r-1$.  Combing \eqref{eq:tmpff} with \eqref{eq:PG}, one has 
\[P_r(0,m)=G_r(0,n)\ge \wt c\frac{1}{r^{d-1}}.\]
This completes the proof of the lemma.  
\end{proof}


\subsection{Harnack type inequalities}\label{sec:Moser-Harnack}
We prove the discrete sub-mean value property and the Moser-Harnack inequality first. 
\begin{lemma}\label{lem:MH-1}
Suppose $f=\{f_n\}$ is a discrete nonnegative subharmonic function on $Q(r;\xi)$ in the sense that
\[
  -(\Delta f)_n=  -\sum_{1\le i\le d} (f_{n+e_i}+f_{n-e_i})+2df_{n}\le 0, \ n\in Q(r-1;\xi).
\]
There is a dimensional constant $C=C(d)>0$ such that 
\begin{equation}
    f_{\xi}\le\, C {r^{1-d}}\sum_{n\in \partial Q(r;\xi)}f_n\quad\mbox{and} \quad
    f_{\xi}\le\, C r^{-d}\sum_{n\in  Q(r;\xi)}f_n \label{eq:submean*}.
\end{equation}
As a consequence, for any cube $\Omega\subset \Z^d$, if $f_n$ is non-negative and  subharmonic on a domain containing the tripled cube $3\Omega$, then
\begin{align}\label{eq:MH-1}
|\ell(\Omega)|^d\sup_{\xi \in \Omega}f^2_{\xi}\le  C  \sum_{n\in 3\Omega}f^2_{n}.
\end{align}
\end{lemma}
\begin{proof}
Fix $\xi\in \Z^d$, let $\wt{f}$ be the discrete harmonic function on $Q(r;\xi)$ which coincides $f$ on $\partial Q(r;\xi)$, i.e., 
$
	(\Delta \wt f)=0, n\in Q(r-1;\xi)$, and $
	\wt{f}_n=f_n,   n\in \partial Q(r;\xi).
$  By the integration by parts formula \eqref{eq:IBP},  $
    \wt{f}_{\xi}=\sum_{m\in \partial Q(r)}P_r(\xi,m){f}_{m}, 
$
where $P_r(n,m)$ is the discrete Poisson kernel on $Q(r)$ from \eqref{eq:Pdef}. It was showed in \cite{Gu,GuMa} that there is a dimensional constant $C>0$ such that $
    P_r(\xi,m)\le C r^{1-d}$ for all $m\in \partial Q(r;\xi)$.

Let $w_n=\wt f_n-f_n$. Clearly, $-(\Delta w)_n=(\Delta f)_n\ge 0$ for all $n\in Q(r-1;\xi)$ and $w_n=0$ for $n \in \partial Q(r;\xi)$. By the maximum principle Lemma \ref{lem:maxP}, one has $w_n=\wt f_n-f_n\ge 0$ for $n\in Q(r-1;\xi)$. In particular, for all $0\le r'\le r$, 
\begin{equation}\label{eq:submean-surf-pf}
   f_{\xi}\le \wt f_{\xi} 
    =\sum_{m\in \partial Q(r';\xi)}P_{r'}(\xi,m){f}_{m}
    \le C r'^{1-d}\sum_{m\in \partial Q(r';\xi)}f_{m}
\end{equation}
since $f_n\ge 0$. We multiply \eqref{eq:submean-surf-pf} by $(r')^{d-1}$, and then 
sum for all $1\le r' \le r$ to obtain
\[
    f_{\xi}\sum_{r'=1}^r(r')^{d-1}   \le  C \sum_{r'=1}^r\sum_{n\in \partial Q(r';\xi)}f_{n} \le C\sum_{n\in Q(r;\xi)}f_{n}.
\]
Therefore, 
\begin{equation}\label{eq:submean-vol-pf} 
f_{\xi}\le \wt C r^{-d}\sum_{n\in Q(r;\xi)}f_{n},
\end{equation}
which proves \eqref{eq:submean*}. 

Furthermore, \eqref{eq:submean-vol-pf} and H\"older inequality imply that for any $\xi$ and $r$,
\begin{equation}\label{eq:pfA90}
r^df^2_{\xi}\le  \wt C \sum_{n\in Q(r;\xi)}f^2_{n},
\end{equation}
which easily yields \eqref{eq:MH-1}. \end{proof}

As a direct consequence, we have 
\begin{lemma}[Moser-Harnack inequality for sub-solutions]\label{lem:MH-2}
Let  $\Omega\subset \Z^d$ be a cube of side length $\ell(\Omega)$, and let $3\Omega$ be the tripled cube  \eqref{eq:3Q}. Suppose $g_n$ is a non-negative sub-solution to an inhomogeneous equation on a domain containing $3\Omega$, so that
$
-(\Delta g)_n\le 1$ and $g_n\ge 0$ for $ n\in 3\Omega.
$
Then
 there is a dimensional constant $c_H>0$ such that  
\begin{equation}\label{eq:MH-inhomo}
\sum_{n\in 3\Omega}g_n^2\ge |\ell(\Omega)|^d\, \Bigl(c_H \sup_\Omega g^2_n-|\ell(\Omega)|^4\Bigr).
\end{equation}
\end{lemma}
\begin{proof}
Suppose $3\Omega=\llbracket a_1,b_1 \rrbracket\times \cdots \times  \llbracket a_d,b_d \rrbracket$, $b_i-a_i+1=3\ell,i=1,\cdots,d$, where $\ell=\ell(\Omega)$ is the side length of $\Omega$. Denote by $|3\Omega|=3^d\ell^d$ its cardinality  as usual.  For $n=(n_1,\cdots, n_d)\in 3\Omega$, let 
 
\[
h_n=\frac{9}{8} \ell^2-\frac{1}{2d}\sum_{i=1}^d(n_i-a_i)(b_i-n_i).
\]
 
Direct computations show that 
\[
(\Delta h)_n=1\quad {\mbox{and }} \quad 
0\le h_n \le  \frac{9}{8}  \ell ^2, \quad n\in 3\Omega. \]

Let $f_n=h_n+g_n$. Then $-(\Delta f)_n =-(\Delta h)_n-(\Delta g)_n\le   0$ and $f_n\ge g_n \ge 0$. We can apply Lemma \ref{lem:MH-1} to the non-negative subharmonic function $f_n$. The estimate \eqref{eq:MH-1} implies that
\[
|\ell(\Omega)|^d\sup_{n\in \Omega}g^2_n\le|\ell(\Omega)|^d\sup_{n\in \Omega}f^2_n
\le  C  \sum_{n\in 3\Omega}f^2_{n}
\le   C\Bigl(|\ell(\Omega)|^{4+d}+ \sum_{n\in 3\Omega}g_n^2\Bigr),
\]
as desired.
  
\end{proof}

Next, we study the discrete Harnack inequality for  sup-solutions of a homogeneous Schr\"odinger equation with a bounded potential.

\begin{lemma}\label{lem:Harnack}
Suppose $v=\{v_n\}$, $v_n\le V_{\max}$, is a bounded potential. Let $f$ be a non-negative super-solution for the Schr\"odinger equation on a cube $\Omega\subset \Z^d$ so that
\begin{equation}\label{eq:sup-sln}
    -(\Delta f)_n+v_nf_n\ge 0, \ \ f_n\ge 0,\ \ n\in\Omega.
\end{equation}
There is  a constant $C>0$ depending on $d$ and $V_{\max}$ such that for any cube $Q\subset \Omega$ of side length $\ell(Q)$, 
\begin{equation}\label{eq:Harnack}
   \sup_{n\in Q}f_n\le C^{\ell(Q)}  \inf_{n\in Q}f_n .
\end{equation}

\end{lemma}

\begin{proof}
Assume that the finite dimensional vector $\{f_n\}_{n\in Q}$ attains its minimum and maximum at $m,n\in Q$ respectively. Connect $m,n$ by a discrete path $\{\gamma^j\}_{j=0}^s$ in $Q$, where $\gamma_0=m,\gamma_s=n$ and $\gamma^{j+1}=\gamma^j\pm e_i$ for some $i$. It is easy to check that the minimum number of steps needed to reach $n$ from $m$ is $s=|m-n|_1\le d\, \ell(Q)$. The upper bound for $v_n$ and \eqref{eq:sup-sln} imply   $
(2d+V_{\max})f_k\ge \sum_{|k'-k|_1=1}f_{k'},
$  for all $k\in \Omega$. Then 
\[
    f_m\ge \frac{1}{2d+V_{\max}}\sum_{|m'-m|_1=1}f_{m'}\ge \frac{1}{2d+V_{\max}}f_{\gamma^1}\ge \frac{1}{(2d+V_{\max})^2}f_{\gamma^2}\ge \cdots \frac{1}{(2d+V_{\max})^s}f_{\gamma^s}.
\]
Therefore,  
\[
   \inf_{n\in Q}f_n \ge  (2d+V_{\max})^{-d\, \ell(Q)} \sup_{n\in Q}f_n, 
\]
which gives \eqref{eq:Harnack}.
\end{proof}


\section{Chernoff bound}
\begin{lemma}[Chernoff–Hoeffding Theorem, \cite{Ho}]\label{lem:Chernoff}
Suppose $B\subset \Z^d$ and  $\{\zeta_j\}_{j\in B}$ are i.i.d. Bernoulli random variables, taking values in $\{0,1 \}$ with common expectation $p=\Ev{\zeta_j}\in(0,1)$. Then for any $0<\lambda<1-p$, 
\begin{equation}\label{eq:Chernoff-app}
    \P \Bigl\{ \sum_{j\in B } \zeta_j\ge (1-\lambda)|B|\Bigr\} \le e^{-D(1-\lambda \| p)\, |B|} , 
\end{equation}
where 
\begin{equation}\label{eq:KL-div-app}
    D(x\|y)=x\log  \frac{x}{y}+(1-x)\log  \frac{1-x}{1-y}
\end{equation}
is the Kullback–Leibler divergence between Bernoulli distributed random variables with parameters $x$ and $y$ respectively.
\end{lemma}
We sketch the proof, following the arguments used in \cite{DFM} (which is also close to the original proof of Hoeffding),  for readers' convenience. 

\begin{proof}
Let $S=\sum_{j\in B } \zeta_j$. For any $t>0$,
\[
    \Pr{ S \ge (1-\lambda)|B|}= \Pr{e^{t S }\ge e^{t (1-\lambda)|B|}} 
    \le  e^{-t (1-\lambda)|B|}\, \Ev{e^{t S }}  
    =  e^{-t (1-\lambda)|B|}\, \left( \Ev{e^{t\zeta_1}} \right)^{|B|}.
\]

For any $j$, the expectation of $e^{t\zeta_j}$ is  $
\Ev{e^{t\zeta_1}}= e^t p + 1-p.
$ Therefore, for all $t>0$,
\begin{equation}
    \log \Pr{ S \ge (1-\lambda)|B|} 
  \le   -t (1-\lambda)|B|+|B|\log \left( e^t p + 1-p     \right) =: -|B|\, f(t)   \label{eq:temp3}.
\end{equation}
It is enough to optimize $f(t)$ in $t$. Under the condition of $1-\lambda-p>0$, the function $f(t)$ attains its only local maximum at 
$
    {t^\ast}=\log \left(\frac{1-\lambda}{\lambda}\frac{1-p}{p} \right).
$
Direct computations show that 
\[
    f(t^\ast)= t^\ast(1-\lambda)-\log \left( e^{t^\ast} p + 1-p     \right)
    = (1-\lambda)\log \frac{1-\lambda}{p}\, +\lambda\log \frac{\lambda}{1-p}=D(1-\lambda\|p)
\]
where $D(x\|y)$ is defined in \eqref{eq:KL-div-app}.  Combing with \eqref{eq:temp3}, we have 
\eqref{eq:Chernoff-app}.
\end{proof}

\Addresses

\end{document}